\documentclass[a4paper,USenglish,cleveref,autoref,thm-restate,numberwithinsect]{lipics-arxiv-v2021}
\usepackage{mathtools}

\usepackage{tikz}
\usetikzlibrary{decorations.pathreplacing}
\usetikzlibrary{patterns,arrows,decorations.pathreplacing,decorations.pathmorphing,calc}

\definecolor{darkpastelgreen}{rgb}{0.01, 0.75, 0.24}
\definecolor{springgreen}{rgb}{0.75, 1.0, 0.0}

\definecolor{mYellow}{RGB}{255,237,0}

\tikzstyle{vertex}=[circle,fill=white,draw=black,minimum size=7pt,inner sep=0pt]
\tikzstyle{smallvertex}=[circle,fill=white,draw=black,minimum size=5pt,inner sep=0pt]

\newcommand{\figscalesmall}{0.8}
\newcommand{\figscalelarge}{0.8}

\newcommand{\NN}{\mathbb N}
\newcommand{\ZZ}{\mathbb Z}
\newcommand{\RR}{\mathbb R}

\newcommand{\CM}{\mathcal M}

\newcommand{\notleftright}{\mathrel{\ooalign{$\Leftrightarrow$\cr\hidewidth$/$\hidewidth}}}

\newcommand{\WL}[2]{\chi_{\sf WL}^{#1}\![#2]}
\newcommand{\WLit}[3]{\chi_{#2}^{#1}\![#3]}
\newcommand{\WLab}[1]{\chi_{\sf WL}^{#1}}
\newcommand{\undir}[1]{\textsf{undir}(#1)}

\DeclareMathOperator{\im}{im}
\DeclareMathOperator{\dist}{dist}
\DeclareMathOperator{\BP}{BP}

\DeclareMathOperator{\Fix}{Singles}
\DeclareMathOperator{\Disc}{Singles}
\DeclareMathOperator{\inte}{Int}
\DeclareMathOperator{\exte}{Ext}

\DeclareMathOperator{\Aut}{Aut}

\title{A Study of Weisfeiler--Leman Colorings on Planar Graphs}

\titlerunning{WL Colors on Planar Graphs}

\author{Sandra Kiefer}{Max Planck Institute for Software Systems, Saarland Informatics Campus, Saarbrücken, Germany}{sankie@mpi-sws.org}{https://orcid.org/0000-0003-4614-9444}{}
\author{Daniel Neuen}{School of Computing Science, Simon Fraser University, Burnaby, Canada}{dneuen@sfu.ca}{https://orcid.org/0000-0002-4940-0318}{}

\authorrunning{S.\ Kiefer and D.\ Neuen} % mandatory. First: Use abbreviated first/middle names. Second (only in severe cases): Use first author plus 'et al.'

\Copyright{Sandra Kiefer and Daniel Neuen} % mandatory, please use full first names. LIPIcs license is "CC-BY";  http://creativecommons.org/licenses/by/3.0/



\relatedversion{}

\nolinenumbers

\PlatformAcronym{arXiv}

\begin{document}

\maketitle

\begin{abstract}
 The Weisfeiler--Leman (WL) algorithm is a combinatorial procedure that computes colorings on graphs, which can often be used to detect their (non-)isomorphism. Particularly the $1$- and $2$-dimensional versions $1$-WL and $2$-WL have received much attention, due to their numerous links to other areas of computer science. 

 Knowing the expressive power of a certain dimension of the algorithm usually amounts to understanding the computed colorings.
 An increase in the dimension leads to finer computed colorings and, thus, more graphs can be distinguished.
 For example, on the class of planar graphs, $3$-WL solves the isomorphism problem.
 However, the expressive power of $2$-WL on the class is poorly understood (and, in particular, it may even well be that it decides isomorphism). 

 In this paper, we investigate the colorings computed by $2$-WL on planar graphs.
 Towards this end, we analyze the graphs induced by edge color classes in the graph.
 Based on the obtained classification, we show that for every $3$-connected planar graph, it holds that: a) after coloring all pairs with their $2$-WL color, the graph has fixing number $1$ with respect to $1$-WL, or b) there is a $2$-WL-definable matching that can be used to transform the graph into a smaller one, or c) $2$-WL detects a connected subgraph that is essentially the graph of a Platonic or Archimedean solid, a prism, a cycle, or a bipartite graph $K_{2,\ell}$.
 In particular, the graphs from case (a) are identified by $2$-WL.
\end{abstract}

\section{Introduction}

The Weisfeiler--Leman (WL) algorithm \cite{WeisfeilerL68} is a combinatorial procedure that, given a graph~$G$, computes a coloring on $G$ which respects (and sometimes also detects) the symmetries in the graph.
Its most prominent application is in theoretical \cite{Babai16,CaiFI92,Kiefer20} and practical approaches \cite{AndersS21,DargaLSM04,JunttilaK07,McKay81,McKayP14} to the graph isomorphism problem.
The original algorithm by Weisfeiler and Leman is the $2$-dimensional version and it colors pairs of vertices.
Its generalization yields for every natural number $k$ the $k$-dimensional WL algorithm $k$-WL, which iteratively refines a coloring of vertex $k$-tuples by aggregating local structural information encoded in the colors.
Its final output is a coloring that is \emph{stable} with respect to the criterion for partitioning the color classes, and graphs with different final colorings are never isomorphic.

Over the decades, fascination for the algorithm has persisted.
This is to a large extent due to the discovery of numerous connections to other areas in computer science that are still being explored.
For example, the algorithm has close links to linear and semidefinite programming~\cite{AtseriasM13,AtseriasO18,GroheO15}, homomorphism counting~\cite{DellGR18,Dvorak10}, and machine learning \cite{AhmadiKMN13,Grohe21,MorrisRFHLRG19,ShervashidzeSLMB11,XuHLJ19}.
Its expressive power can be characterized via winning strategies for the players in a particular type of Ehrenfeucht-Fraïssé game \cite{CaiFI92,Hella96}.
Moreover, it is known that two graphs receive different final colorings with respect to $k$-WL if and only if the graphs can be distinguished via a formula in the counting-logic fragment $C^{k+1}$ \cite{CaiFI92,ImmermanL90}.

In this work, we focus on the original version $2$-WL, as introduced by Weisfeiler and Leman~\cite{WeisfeilerL68}.
Besides the connections outlined for $k$-WL above, $2$-WL has a precise correspondence to coherent configurations (see, e.g., \cite{ChenP19}).
Despite the simple and very natural concept behind the algorithm, its behavior is not well-understood and there is an extensive line of study to capture its expressive power.
For example, one branch of research aims at understanding which graph properties can be detected by $2$-WL.
In this direction, Fürer \cite{Furer17} as well as Arvind et al.\ and Fuhlbrück et al.\ \cite{ArvindFKV20,FuhlbruckKV21a} obtained insights concerning the ability of $2$-WL to detect and count small subgraphs.
Furthermore, the algorithm is able to detect $2$-separators in graphs and implicitly computes the decomposition of a graph into its $3$-connected components \cite{KieferN22}.

A related line of research analyzes which graphs are \emph{identified} by $2$-WL, i.e., on which graphs $2$-WL serves as a complete isomorphism test.
Positive examples include interval graphs \cite{EvdokimovPT00} and distance-hereditary graphs \cite{GavrilyukNP20} as well as almost all regular graphs \cite{Bollobas82}. In the light of the  upper bound of $3$ on the dimension of the algorithm needed to identify all planar graphs \cite{KieferPS19}, there is hope that the class of planar graphs can eventually be added to the list.
Towards a complete characterization of the expressive power of $2$-WL, Fuhlbrück, Köbler, and Verbitsky \cite{FuhlbruckKV21} developed an algorithmic characterization of the graphs of color class size at most $4$ that are identified by $2$-WL.

\subparagraph{Our Contribution}

In this work, we investigate $2$-WL on planar graphs. We are interested in analyzing the stable output coloring computed by $2$-WL and deducing symmetries and other properties of the input graph from properties of the coloring. 

As a starting point, we precisely characterize the planar graphs in which all edges receive the same color with respect to $2$-WL.
Since the coloring that $2$-WL computes is preserved by automorphisms, edge-transitive planar graphs clearly fall into this category.
As our first main result, we prove the converse of this statement: every planar graph in which all edges receive the same color with respect to $2$-WL is edge-transitive.
To show the implication, we reprove the classification of edge-transitive planar graphs (see, e.g., \cite{GrunbaumS87}) building solely on the $2$-WL coloring.

Using the classification, we continue to analyze the WL coloring on general planar graphs.
Since, by \cite{KieferN22}, the algorithm $2$-WL implicitly computes the graph decomposition into $3$-connected components, understanding $2$-WL on planar graphs essentially amounts to a study of $3$-connected planar graphs.
Here, we can exploit a theorem due to Whitney \cite{Whitney32}, which says that all embeddings of a $3$-connected planar graph are combinatorially equivalent. 

Our focus lies on the following three tasks: (i) classify the subgraphs induced by edges of the same $2$-WL color that can occur, (ii) analyze how these subgraphs interleave, and (iii) establish connections to properties of the entire graph $G$.

Let $G$ be a $3$-connected planar graph and let $C_E(G)$ denote the set of $2$-WL colors that correspond to edges of $G$.
For every $c \in C_E(G)$, denote by $G[c]$ the subgraph induced by all edges of $2$-WL color $c$.
To describe our results, it turns out to be useful to partition the edge colors into three types depending on the number of faces per connected component of $G[c]$.
We say that $c$ has \emph{Type I} if every connected component of $G[c]$ has one face, \emph{Type II} if every connected component of $G[c]$ has two faces, and \emph{Type III} if every connected component of $G[c]$ has at least three faces. (By the properties of $2$-WL, these types indeed cover all cases that can occur.)

First, we analyze the graphs induced by edge colors $c$ of Type III.
It is not hard to see that every edge in such a graph $G[c]$ receives the same $2$-WL color (when applying $2$-WL to $G[c]$), and thus, by our classification, $G[c]$ is edge-transitive.
However, it turns out that much stronger statements are possible, since, in the end, many edge-transitive planar graphs cannot appear as a graph $G[c]$.
For example, we show that $G[c]$ is always connected.
As our central result for colors of Type III, we obtain a precise classification of the possible graphs $G[c]$.
An interesting consequence of this classification is that the automorphism group $\Aut(G)$ of $G$ is always isomorphic to a subgroup of $\Aut(G[c])$.
More precisely, we show that fixing the images of all vertices of $G[c]$ uniquely determines the image of every vertex of $G$ under any automorphism of $G$.
Hence, by only looking at the subgraph induced by a single edge color of Type III, we obtain strong insights about the symmetries of the entire graph.

On the other side of the spectrum, we prove that if all edge colors are of Type I, then $G$ has \emph{fixing number} at most $1$, where the fixing number is the minimum number of vertices that need to be fixed pointwise so that the identity mapping is the only automorphism of $G$.
It is known that $3$-connected planar graphs have fixing number at most $3$, and there is a complete characterization of those graphs of fixing number exactly $3$ \cite{KieferPS19}. In our analysis of $3$-connected planar graphs $G$ in which all edge colors are of Type I, we only use $1$-WL to prove that $G$ has no non-trivial automorphisms after fixing a certain single vertex. This implies that $2$-WL identifies all such graphs.

If neither of the above cases applies, then there is an edge color of Type II.
Let us first remark that the graphs of many Archimedean solids fall into this category (whereas the edge colors in the graphs of all Platonic solids are of Type III). 
In such a situation, the graph of the Archimedean solid is defined by edge colors $c,d$ where one of the two colors has Type II.
With this in mind, towards solving task (ii), we analyze how edge colors of Type II interleave with other edge colors. 
More precisely, similarly as for Type III, we aim at identifying a connected subgraph defined by two colors $c,d \in C_E(G)$, where $c$ has Type II, that corresponds to one of the Archimedean solids, or stems from a small number of infinite graph families. 
We remark that, similar to the case of edge colors of Type III, if we have such a subgraph $G[c,d]$, then $\Aut(G)$ is isomorphic to a subgroup of $\Aut(G[c,d])$.

We show that either this goal can be achieved, or $G$ has fixing number $1$ or there is a \emph{WL-definable matching}. 
Such a matching is given by an edge color $c$ such that $G[c]$ is a matching graph (i.e., every vertex has degree $1$) and the endpoints of every edge receive different colors.
Such matchings also play a crucial role in the analysis of $2$-WL on graphs of color class size $4$ \cite{FuhlbruckKV21}, and contracting all matching edges preserves many crucial properties related to WL  such as the stable coloring, identifiability by WL, as well as the automorphism group of $G$.
As a result, finding a WL-definable matching is beneficial since we can proceed to a smaller graph without affecting the problem at hand.

\subparagraph{Towards the WL Dimension of Planar Graphs}
\label{par:wl-dim}

The WL dimension of a graph class $\mathcal{C}$ is the minimal $k$ such that $k$-WL identifies every graph from $\mathcal{C}$, i.e., $k$-WL serves as a complete isomorphism test for the class $\mathcal{C}$.
Many classes of graphs are known to have a finite WL dimension, for example, interval graphs \cite{EvdokimovPT00}, graphs of bounded rank-width \cite{GroheN19} as well as graphs of bounded genus \cite{GroheK19} and, more generally, all graph classes that exclude a fixed graph as a minor \cite{Grohe12,Grohe17}.

For planar graphs, the quest for bounds on their WL dimension was initiated by Immerman already over three decades ago \cite{Immerman87}.
In a first step, Grohe \cite{Grohe98} proved that the dimension is finite.
Analyzing Grohe's proof in detail, Redies \cite{Redies14} showed an upper bound of $14$ on the WL dimension of planar graphs.
This was further improved in \cite{KieferPS19}, where it is shown that already $3$-WL identifies all planar graphs, thus narrowing down the WL dimension of planar graphs to $2$ or $3$.
Moreover, it was recently shown that a constant dimension of the WL algorithm suffices to identify all planar graphs in a logarithmic number of refinement rounds \cite{GroheK21}, extending previous results for $3$-connected planar graphs \cite{Verbitsky07}. 
Still, the task to determine the precise WL-dimension of the class of planar graphs remains open.
A central motivation for our work is to determine whether $2$-WL identifies every planar graph.

\begin{figure}
 \centering
 \begin{tikzpicture}[scale=0.8]
  \draw[line width=1.6pt, mYellow, fill = mYellow!40] (0,0) circle (3.6cm);
  \draw[line width=1.6pt, mYellow, fill = white] (0,0) circle (2.7cm);
  \node at (30:3.15) {$C_4$};
  
  \node[vertex,red!80] (1) at (0,0) {};
  \node[vertex,red!80] (2) at (90:2.4) {};
  \node[vertex,red!80] (3) at (210:2.4) {};
  \node[vertex,red!80] (4) at (330:2.4) {};
  
  \draw[line width=1.6pt] (1) edge (2);
  \draw[line width=1.6pt] (1) edge (3);
  \draw[line width=1.6pt] (1) edge (4);
  
  \draw[line width=1.6pt,bend right=56] (2) edge (3);
  \draw[line width=1.6pt,bend right=56] (3) edge (4);
  \draw[line width=1.6pt,bend right=56] (4) edge (2);
  
  \draw[line width=1.6pt, mYellow, fill = mYellow!40] (150:0.3) to ($(150:0.3) + (90:1.9)$) to[bend right=45] ($(150:0.3) + (210:1.9)$) to (150:0.3);
  \draw[line width=1.6pt, mYellow, fill = mYellow!40] (270:0.3) to ($(270:0.3) + (210:1.9)$) to[bend right=45] ($(270:0.3) + (330:1.9)$) to (270:0.3);
  \draw[line width=1.6pt, mYellow, fill = mYellow!40] (30:0.3) to ($(30:0.3) + (330:1.9)$) to[bend right=45] ($(30:0.3) + (90:1.9)$) to (30:0.3);
  
  \node at (30:1.4) {$C_1$};
  \node at (150:1.4) {$C_2$};
  \node at (270:1.4) {$C_3$};
  
  \node[vertex,blue!80] (c1) at (30:0.6) {};
  \node[vertex,blue!80] (c2) at (150:0.6) {};
  \node[vertex,blue!80] (c3) at (270:0.6) {};
  \node[vertex,blue!80] (c4) at (90:3.0) {};
  \node[vertex,blue!80] (c5) at (210:3.0) {};
  \node[vertex,blue!80] (c6) at (330:3.0) {};
  \node[vertex,blue!80] (c7) at ($(90:2.4) + (235:0.8)$) {};
  \node[vertex,blue!80] (c8) at ($(90:2.4) + (305:0.8)$) {};
  \node[vertex,blue!80] (c9) at ($(210:2.4) + (65:0.8)$) {};
  \node[vertex,blue!80] (c10) at ($(210:2.4) + (355:0.8)$) {};
  \node[vertex,blue!80] (c11) at ($(330:2.4) + (115:0.8)$) {};
  \node[vertex,blue!80] (c12) at ($(330:2.4) + (185:0.8)$) {};
  
  \draw[line width=1.6pt, darkpastelgreen] (1) edge (c1);
  \draw[line width=1.6pt, darkpastelgreen] (1) edge (c2);
  \draw[line width=1.6pt, darkpastelgreen] (1) edge (c3);
  
  \draw[line width=1.6pt, darkpastelgreen] (2) edge (c4);
  \draw[line width=1.6pt, darkpastelgreen] (3) edge (c5);
  \draw[line width=1.6pt, darkpastelgreen] (4) edge (c6);
  
  \draw[line width=1.6pt, darkpastelgreen] (2) edge (c7);
  \draw[line width=1.6pt, darkpastelgreen] (2) edge (c8);
  \draw[line width=1.6pt, darkpastelgreen] (3) edge (c9);
  \draw[line width=1.6pt, darkpastelgreen] (3) edge (c10);
  \draw[line width=1.6pt, darkpastelgreen] (4) edge (c11);
  \draw[line width=1.6pt, darkpastelgreen] (4) edge (c12);
 \end{tikzpicture}
 \caption{The visualization shows a $3$-connected planar graph $G$ and an edge color $c \in C_E(G)$ (shown in black) such that $G[c]$ is isomorphic to $K_4$. Also, $C_1,C_2,C_3,C_4$ are the vertex sets of the connected components of $G - V(G[c])$.
 Using regularity constraints specific to the case $G[c] \cong K_4$, one can show that the graphs induced by the $C_i$ are all indistinguishable by $2$-WL.
 More strongly, it actually holds for all $i,j \in [4]$, all $v \in N(C_i)$ and all $w \in N(C_j)$ that $\{\!\{\chi(v,v') \mid v' \in N(v) \cap C_i\}\!\} = \{\!\{\chi{G}(w,w') \mid w' \in N(w) \cap C_j\}\!\}$, where $\chi$ denotes the $2$-WL coloring. (In the picture, this multiset contains one green edge, i.e., every $v \in V(G[c])$ has exactly one neighbor via a green edge in each adjacent $C_i$.)}
 \label{fig:decomposition}
\end{figure}
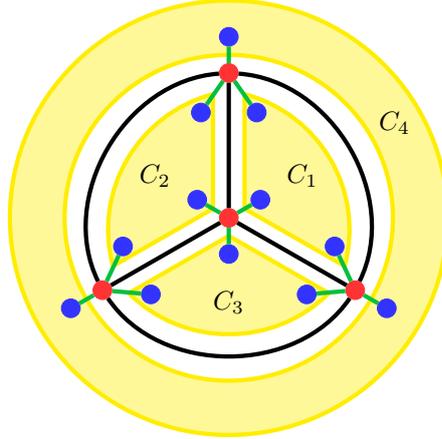

Our results suggest an inductive approach to this question.
Indeed, building on the fact that $2$-WL is able to detect the decomposition into $3$-connected components \cite{KieferN22}, we can restrict our attention to $3$-connected graphs.
Given a $3$-connected planar graph $G$, by combining the results described above, we always obtain that $G$ has one of the following:
\begin{enumerate}[label=(\Alph*)]
 \item\label{item:wl-dim-1} fixing number $1$ under $1$-WL, i.e., individualizing a single vertex and performing $1$-WL (after coloring all pairs by their $2$-WL color) 
 results in a discrete coloring,
 \item\label{item:wl-dim-2} a WL-definable matching, or
 \item\label{item:wl-dim-3} a connected subgraph induced by at most two edge colors that corresponds to a Platonic or Archimedean solid or stems from a small number of infinite graph families.
\end{enumerate}
In Case \ref{item:wl-dim-1}, the graph $G$ is identified by $2$-WL.
In Case \ref{item:wl-dim-2}, we can follow the strategy outlined in \cite{FuhlbruckKV21} and move to a smaller graph by contracting the definable matching.
Therefore, determining the WL dimension of the class of planar graphs boils down to defeating Case \ref{item:wl-dim-3}.
In this case, we obtain a connected subgraph $H$ that is defined by at most two edge colors $c$ and $d$ and which we can classify precisely.
Let $C_1,\dots,C_s$ denote the vertex sets of the connected components of $G - V(H)$, the graph $G$ with the vertices in $H$ removed (see also Figure \ref{fig:decomposition}).
Also, let $G'$ be a second graph that cannot be distinguished from $G$ by $2$-WL.
Let $H'$ denote the subgraph of $G'$ induced by $c$ and $d$ and let $C_1',\dots,C_{s}'$ denote the vertex sets of the connected components of $G' - V(H')$. Presupposing by induction that the statement holds for smaller graphs, we may assume that $2$-WL identifies the subgraphs induced by $C_1,\dots,C_s$.
This implies that $G[C_i]$ is isomorphic to $G'[C_i']$ for all $i \in [s]$ (possibly after reordering the sets $C_1',\dots,C_{s}'$).
It is not hard to see that $2$-WL identifies $H$ and, thus, $H$ is isomorphic to $H'$.
Now, ideally, we want to glue all these partial isomorphisms together to obtain a global isomorphism from $G$ to $G'$.
Towards this end, it is our intuition that the options for the interplay between $H$ and the sets $C_i$ are extremely limited due to $G$ being planar and $H$ being defined by few edge colors, which enforces strong regularity conditions on the interaction between $H$ and the surrounding graph. A formalization of such a strategy for all subcases that can appear in Case \ref{item:wl-dim-3} should yield that $2$-WL identifies every planar graph.

\section{Preliminaries}
\label{sec:prelim}

\subsection{Graphs}

\subparagraph{Basics}
An (undirected) \emph{graph} is a pair $G = (V(G),E(G))$ of a finite \emph{vertex set} $V(G)$ and an \emph{edge set} $E(G) \subseteq \big\{\{u,v\} \, \big\vert \, u \neq v \in V(G)\big\}$. 
All graphs considered in this paper are finite and simple (i.e., they contain no loops or multiple edges).
Unless stated explicitly otherwise, graphs are undirected.
For a directed graph $G'$, we write $\undir{G'}$ to denote its undirected version.
For $v,w \in V(G)$, we also write $vw$ as a shorthand for $\{v,w\}$.
The \emph{neighborhood} of~$v$ in $G$ is denoted by~$N_G(v)$ and the \emph{degree} of $v$ in $G$ is $\deg_G(v) \coloneqq |N_G(v)|$.
The \emph{closed neighborhood} of $v$ in $G$ is $N_G[v] \coloneqq N(v) \cup \{v\}$.
For $W \subseteq V(G)$, we also define $N(W) \coloneqq \left(\bigcup_{v \in W}N(v)\right) \setminus W$.
For a directed graph $G'$, we define $N_{G'}^+(v) \coloneqq \{w \in V(G') \mid (v,w) \in E(G')\}$ and $N_{G'}^-(v) \coloneqq \{w \in V(G') \mid (w,v) \in E(G')\}$.
Also, $N_{G'}(v) \coloneqq N_{G'}^+(v) \cup N_{G'}^-(v)$.
If the graph is clear from the context, we usually omit the index.

A \emph{walk of length $k$} from $v$ to $w$ is a sequence of vertices $v = u_0,u_1,\dots,u_k = w$ such that $u_{i-1}u_i \in E(G)$ for all $i \in [k] \coloneqq \{1,\dots,k\}$. 
A \emph{path of length $k$} from $v$ to $w$ is a walk of length $k$ from $v$ to $w$ for which all occurring vertices are pairwise distinct.
The \emph{distance} $\dist_G(v,w)$ between two vertices $v,w \in V(G)$ is the length of a shortest path between them.
For an undirected graph $G$ and $U,W \subseteq V(G)$, we define $E_G(U,W) \coloneqq \{uw \in E(G) \mid u \in U, w \in W\}$.
We denote by $G[W]$ the \emph{induced subgraph} of $G$ on the vertex set $W$, and define $G - W \coloneqq G[V(G) \setminus W]$.
A \emph{minor} of $G$ is a graph obtained by deleting vertices and edges of $G$ as well as contracting edges of $G$, i.e., replacing an edge $vw$ with a fresh vertex $u$ with $N(u) \coloneqq (N(v) \setminus \{w\}) \cup (N(w) \setminus \{v\})$.
A set $S \subseteq V(G)$ is a \emph{separator} of $G$ if $G - S$ has more connected components than $G$. 
A \emph{$k$-separator} of $G$ is a separator of $G$ of size $k$.
A vertex $v \in V(G)$ is a \emph{cut vertex} if $\{v\}$ is a separator of $G$.
The graph $G$ is \emph{$k$-connected} if it is connected and has no separator of size at most $k-1$.

In our definitions of vertex sets of graphs, we use the notation $\uplus$ to denote a formal disjoint union.
More precisely, for sets $U$ and $W$, the set $U \uplus W$ contains $|U| + |W|$ vertices, one distinct copy of each vertex in $U$ and one distinct copy of each vertex in $W$.
(For ease of notation, we refer to the vertices by their original names in $U$ and $W$ instead of renaming them first.)

\subparagraph{Colorings}
A \emph{vertex-colored graph} is a tuple $(G,\lambda)$ where $G$ is a graph and $\lambda\colon V(G) \rightarrow C$ is a \emph{vertex coloring}, a mapping from $V(G)$ into some set $C$ of colors.
We define the set of \emph{arcs} of a graph $G$ as $A(G) \coloneqq \{(v,v) \mid v \in V(G)\} \cup \{(v,w) \mid \{v,w\} \in E(G)\}$.
Observe that, for each $vw \in E(G)$, there are the two arcs $(v,w)$, $(w,v)$.
An \emph{arc-colored graph} is a tuple $(G,\lambda)$ where $G$ is a graph and $\lambda\colon A(G) \rightarrow C$ is a mapping from $A(G)$ into some set $C$ of colors. 
Similarly, a \emph{pair-colored graph} is a tuple $(G,\lambda)$ where $G$ is a graph and $\lambda\colon (V(G))^2 \rightarrow C$ is a mapping into some set $C$ of colors.

Typically, the set $C$ is chosen to be an initial segment $[n]$ of the natural numbers.
We say a coloring $\lambda$ is \emph{discrete} if it is injective, i.e., all color classes have size $1$. 
Finally, for a coloring $\lambda$ and distinct vertices $v_1,\dots,v_\ell$, we denote by $(G,\lambda,v_1,\dots,v_\ell)$ the colored graph where each $v_i$ for $i \in [\ell]$ is individualized.
To be more precise, if $\lambda$ is a vertex coloring, then $(G,\lambda,v_1,\dots,v_\ell) \coloneqq (G,\widetilde{\lambda})$ where $\widetilde{\lambda}(v_i) = (1,i)$ for all $i \in [\ell]$, and $\widetilde{\lambda}(v) = (0,\lambda(v))$ for all $v \in V(G) \setminus\{v_1,\dots,v_\ell\}$.
The definitions for arc and pair colorings are analogous.
We generally assume that all graphs are arc-colored even if not explicitly stated.
Every (uncolored) graph can be interpreted as an arc-colored graph by assigning to every diagonal arc $(v,v)$ the color $1$ and assigning to every non-diagonal arc the color $2$.

An \emph{isomorphism} from $G$ to another graph $H$ is a bijection $\varphi\colon V(G) \rightarrow V(H)$ that respects the edge relation, that is, for all~$v,w \in V(G)$, it holds that~$\{v,w\} \in E(G)$ if and only if $\{\varphi(v),\varphi(w)\} \in E(H)$.
The graphs $G$ and $H$ are \emph{isomorphic} ($G \cong H$) if there is an isomorphism from~$G$ to~$H$.
We write $\varphi\colon G\cong H$ to denote that $\varphi$ is an isomorphism from $G$ to $H$.
An \emph{automorphism} of $G$ is an isomorphism from $G$ to itself and we write $\Aut(G)$ to denote the group of all automorphism of $G$.
Isomorphisms between vertex-, arc-, and pair-colored graphs have to respect the colors of the vertices, arcs, and pairs.

\subparagraph{Planar Graphs}
A graph is called \emph{planar} if it can be embedded into the plane $\mathbb{R}^2$. It is easy to see that the class of all planar graphs is minor-closed.

The following well-known characterization of planar graphs in terms of forbidden minors is often useful to argue about structural properties of planar graphs.
We define $K_n$ to be the complete graph on $n$ vertices, and $K_{m,n}$ to be the complete bipartite graph with $m$ vertices on the left side and $n$ vertices on the right side.

\begin{theorem}[Wagner's theorem \cite{Wagner37}]
 \label{thm:wagner}
 A graph $G$ is planar if and only if neither $K_5$ nor $K_{3,3}$ is a minor of $G$.
\end{theorem}

A \emph{plane graph} is a graph embedded into the plane. As the following statement shows, all plane realizations of a planar graph have the same number of faces, i.e., regions bounded by edges.

\begin{theorem}[Euler's formula]
 Let $G$ be a connected plane graph with $n$ vertices, $m$ edges, and $f$ faces. Then $n-m+f = 2$.
\end{theorem}

It follows quite easily from Euler's formula that for every planar graph $G$, it holds that $m \leq 3n - 6$.
Also, if $n > 3$ and $G$ does not contain a triangle, then $m \leq 2n - 4$.
This implies the existence of vertices of small degree

\begin{corollary}
 \label{cor:planar-degree-bound}
 Let $G$ be a planar graph.
 Then there is some $v \in V(G)$ such that $\deg(v) \leq 5$.
 Also, if $G$ does not contain a triangle, then there is some $v \in V(G)$ such that $\deg(v) \leq 3$.
\end{corollary}

We will also fall back on the following famous theorem due to Whitney.

\begin{theorem}[Whitney's theorem \cite{Whitney32}]
Up to homeomorphism, a $3$-connected planar graph has a unique embedding into the plane.
\end{theorem}

The theorem allows us to speak about faces of $3$-connected planar graphs as abstract objects, since it implies that in a $3$-connected planar graphs, the set of faces does not depend on a specific embedding and thus, the faces can be viewed as combinatorial objects associated with $G$ and are uniquely defined by their sets of vertices $V(F)$ and the edges $E(F)$ bounding~$F$. We will therefore not draw a clear distinction between this combinatorial view and the topological view of $F$ as a region and just use whichever is most suitable for our purpose.
We use $|F| \coloneqq |V(F)| = |E(F)|$ to denote the \emph{size} of the face $F$.

\subsection{The Weisfeiler--Leman Algorithm}

Let~$\chi_1,\chi_2\colon (V(G))^k \rightarrow C$ be colorings of the~$k$-tuples of vertices of a graph~$G$. 
We say $\chi_1$ \emph{refines} $\chi_2$, denoted $\chi_1 \preceq \chi_2$, if $\chi_1(\bar v) = \chi_1(\bar w)$ implies $\chi_2(\bar v) = \chi_2(\bar w)$ for all $\bar v,\bar w \in (V(G))^{k}$.
The colorings $\chi_1$ and $\chi_2$ are \emph{equivalent}, denoted $\chi_1 \equiv \chi_2$,  if $\chi_1 \preceq \chi_2$ and $\chi_2 \preceq \chi_1$.

\subparagraph{The Algorithm}
Given a graph $G$, the algorithm $1$-WL iteratively computes an isomorphism-invariant coloring of the vertices of $G$.
In this work, we actually require an extension of $1$-WL, which also takes arc colors into account.
For an arc-colored graph $(G,\lambda)$, we define the initial coloring computed by the algorithm via $\WLit{1}{0}{G}(v) \coloneqq \lambda(v,v)$ for all $v \in V(G)$. 
This coloring is refined via $\WLit{1}{i+1}{G}(v) \coloneqq (\WLit{1}{i}{G}(v), \CM_i(v))$ where $\CM_i(v)$ is a multiset defined as
\[\CM_i(v) \coloneqq\Big\{\!\!\Big\{\big(\WLit{1}{i}{G}(w),\lambda(v,w),\lambda(w,v)\big) \;\Big\vert\; w \in N_G(v) \Big\}\!\!\Big\}.\]
By definition, $\WLit{1}{i+1}{G} \preceq \WLit{1}{i}{G}$ holds for all $i \geq 0$.
Hence, there is a minimal value $i_\infty$ such that $\WLit{1}{i_\infty}{G} \equiv \WLit{1}{i_\infty+1}{G}$.
We call the coloring $\WLit{1}{i_\infty}{G}$ the \emph{stable coloring} of $G$ and denote it by $\WL{1}{G}$.
The algorithm $1$-WL takes an arc-colored graph $(G,\lambda)$ as input and returns $\WL{1}{G}$.

We can also apply $1$-WL to a pair-colored graph $(G,\lambda)$. This can be done by defining $\widetilde{\lambda}(v_1,v_2) \coloneqq (1,\lambda(v_1,v_2))$ for all $v_1,v_2 \in V(G)$ with $v_1v_2 \in E(G)$, and $\widetilde{\lambda}(v_1,v_2) \coloneqq (0,\lambda(v_1,v_2))$ for all $v_1,v_2 \in V(G)$ with $v_1v_2 \notin E(G)$. 
Then we define $\WL{1}{G,\lambda} \coloneqq \WL{1}{H,\widetilde{\lambda}}$ where $H$ is a complete graph on vertex set $V(G)$.

Next, we describe the \emph{$k$-dimensional Weisfeiler--Leman algorithm} ($k$-WL) for $k \geq 2$.
For an input graph $G$, let $\WLit{k}{0}{G}\colon (V(G))^{k} \rightarrow C$ be the coloring where each tuple is colored with the isomorphism type of its underlying ordered subgraph.
More precisely, $\WLit{k}{0}{G}(v_1,\dots,v_k) = \WLit{k}{0}{G}(v_1',\dots,v_k')$ if and only if, for all $i,j \in [k]$, it holds that
$v_i = v_j \Leftrightarrow v_i'= v_j'$ and $v_iv_j \in E(G) \Leftrightarrow v_i'v_j' \in E(G)$.
If the graph comes equipped with a coloring, the initial coloring $\WLit{k}{0}{G}$ also takes the input coloring into account. 
More formally, for an arc coloring~$\lambda$, for $\WLit{k}{0}{G}(v_1,\dots,v_k) = \WLit{k}{0}{G}(v_1',\dots,v_k')$ to hold, we have the additional conditions $\lambda(v_i,v_j) = \lambda(v_i',v_j')$ for all $i,j \in [k]$ with $(v_i,v_j) \in A(G)$. 
For a pair coloring $\lambda$, we have the additional conditions $\lambda(v_i,v_j) = \lambda(v_i',v_j')$ for all $i,j \in [k]$.

We then recursively define the coloring $\WLit{k}{i}{G}$ obtained after $i$ rounds of the algorithm.
For $\bar v = (v_1,\dots,v_k) \in (V(G))^k$, set
$\WLit{k}{i+1}{G}(\bar v) \coloneqq \big(\WLit{k}{i}{G}(\bar v), \CM_i(\bar v)\big)$,
where
\[\CM_i(\bar v) \coloneqq \Big\{\!\!\Big\{\big(\WLit{k}{i}{G}(\bar v[w/1]),\dots,\WLit{k}{i}{G}(\bar v[w/k])\big) \;\Big\vert\; w \in V(G) \Big\}\!\!\Big\}\]
and $\bar v[w/i] \coloneqq (v_1,\dots,v_{i-1},w,v_{i+1},\dots,v_k)$ is the tuple obtained from substituting the $i$-th entry of $\bar v$ with $w$.
Again, there is a minimal~$i_\infty$ such that $\WLit{k}{i_{\infty}}{G} \equiv \WLit{k}{i_{\infty}+1}{G}$, and we set $\WL{k}{G} \coloneqq \WLit{k}{i_\infty}{G}$.
More generally, we say a coloring of $k$-tuples $\chi\colon (V(G))^{k} \rightarrow C$ is \emph{$k$-stable} (with respect to $G$) if $\chi \preceq \WL{k}{G}$ and $\chi$ is not strictly refined by applying one round of $k$-WL.

The algorithm $k$-WL takes a (pair- or arc-)colored graph $G$ as input and returns $\WL{k}{G}$.
Given graphs $G$ and $H$, the algorithm \emph{distinguishes} $G$ and $H$ if $\{\!\!\{\WL{k}{G}(\bar v) \mid \bar v \in (V(G))^k\}\!\!\} \neq \{\!\!\{\WL{k}{H}(\bar w) \mid \bar w \in (V(H))^k\}\!\!\}$.
We write $G \simeq_k H$ if $k$-WL does not distinguish $G$ and $H$.
Also, $k$-WL \emph{identifies} $G$ if it distinguishes $G$ from every other non-isomorphic graph.

\begin{definition}\label{def:determines-orbits}
 Let $G$ be a graph and let $k \geq 2$.
 Then \emph{$k$-WL determines arc orbits on $G$} if for every $(v_1,v_2) \in A(G)$, every graph $H$, and every $(w_1,w_2) \in A(H)$ such that $\WL{k}{G}(v_1,v_2,\dots,v_2) = \WL{k}{H}(w_1,w_2,\dots,w_2)$,
 there is an isomorphism $\varphi\colon G \cong H$ such that $\varphi(v_i) = w_i$ holds for both $i \in \{1,2\}$.
 
 Moreover, \emph{$k$-WL determines pair orbits of $G$} if for all $v_1,v_2 \in V(G)$, every graph $H$, and all $w_1,w_2 \in V(H)$ such that $\WL{k}{G}(v_1,v_2,\dots,v_2) = \WL{k}{H}(w_1,w_2,\dots,w_2)$,
 there is an isomorphism $\varphi\colon G \cong H$ such that $\varphi(v_i) = w_i$ holds for both $i \in \{1,2\}$.
\end{definition}

Observe that if $k$-WL determines arc or pair orbits of $G$, then it identifies $G$.
Indeed, if for a second graph $H$, there is no isomorphism from $G$ to $H$, the multisets of $\WLab{k}$-colors in the two graphs must be disjoint by Definition \ref{def:determines-orbits}.

\subparagraph{The Bijective Pebble Game}

To analyze the Weisfeiler--Leman algorithm, it is sometimes more convenient to use the following characterization via an Ehrenfeucht-Fraïssé game, which is known to capture the same information.
Let $k \in \mathbb{N}$.
For graphs $G$ and $H$ on the same number of vertices and with arc colorings $\lambda_G$ and $\lambda_H$, respectively,
we define the \emph{bijective $k$-pebble game} $\BP^{k}(G,H)$ as follows:
\begin{itemize}
 \item The game has two players called \emph{Spoiler} and \emph{Duplicator}.
 \item The game proceeds in rounds, each of which is associated with a pair of positions
 $(\bar v,\bar w)$ with~$\bar v \in \big(V(G)\big)^\ell$ and~$\bar w \in \big(V(H)\big)^\ell$, where $0 \leq \ell \leq k$.
 \item The initial position of the game is a pair of vertex tuples of equal length $\ell$ with $0 \leq \ell \leq k$. If not specified otherwise, the initial position is the pair $\big((),()\big)$ of empty tuples.
 \item Each round consists of the following steps. Suppose the current position of the game is $(\bar v,\bar w) = ((v_1,\ldots,v_\ell),(w_1,\ldots,w_\ell))$.
  First, Spoiler chooses whether to remove a pair of pebbles or to play a new pair of pebbles.
  The first option is only possible if $\ell > 0$, and the second option is only possible if $\ell < k$.
  
  If Spoiler wishes to remove a pair of pebbles, he picks some $i \in [\ell]$ and the game moves to position
  $(\bar v\setminus i,\bar w\setminus i)$ where $\bar v \setminus i \coloneqq (v_1,\dots,v_{i-1},v_{i+1},\dots,v_\ell)$, and the tuple ($\bar w \setminus i)$ is defined in the analogous way.
  Otherwise, the following steps are performed.
  \begin{itemize}
   \item[(D)] Duplicator picks a bijection $f\colon V(G) \rightarrow V(H)$.
   \item[(S)] Spoiler chooses $v \in V(G)$ and sets $w \coloneqq f(v)$.
    Then the game moves to position $\big((v_1,\dots,v_\ell,v),(w_1,\dots,w_\ell,w)\big)$.
  \end{itemize}

  If for the current position~$\big((v_1,\dots,v_\ell),(w_1,\dots,w_\ell)\big)$, the induced ordered subgraphs of $G$ and $H$ are not isomorphic, Spoiler wins the play.
  More precisely, Spoiler wins if there there are $i,j\in [\ell]$ such that $v_i = v_j\notleftright w_i =w_j$, or $v_iv_j \in E(G)\notleftright w_iw_j \in E(H)$, or $\lambda(v_i,v_j) \neq \lambda(w_i,w_j)$ assuming $(v_i,v_j) \in A(G)$.
  If there is no position of the play such that Spoiler wins, then Duplicator wins.
\end{itemize}

Observe that the game naturally extends to pair-colored graphs $G$ and $H$ by adapting the winning condition for Spoiler.

We say that Spoiler (and Duplicator, respectively) \emph{wins the bijective $k$-pebble game $\BP^k(G,H)$} if Spoiler (and Duplicator, respectively) has a winning strategy for the game.

The following theorem describes the correspondence between the Weisfeiler--Leman algorithm and the introduced pebble games.\footnote{The pebble games in \cite{CaiFI92} are defined slightly differently. Still, a player has a winning strategy in the game described there if and only if they have one in our game and thus, Theorem \ref{thm:eq-wl-pebble-tuples} holds for both versions of the game.}

\begin{theorem}[\cite{CaiFI92,Hella96}]
 \label{thm:eq-wl-pebble-tuples}
 Let $G,H$ be two graphs and let $\bar v \in \big(V(G)\big)^{k}$ and $\bar w \in \big(V(H)\big)^{k}$.
 Then $\WL{k}{G}(\bar v) = \WL{k}{H}(\bar w)$ if and only if Duplicator wins the game $\BP^{k+1}(G,H)$ from the initial position $(\bar v,\bar w)$.
\end{theorem}

\begin{corollary}
 \label{cor:eq-wl-pebble}
 Let $G$ and $H$ be two graphs.
 Then $G \simeq_{k} H$ if and only if Duplicator wins the game $\BP^{k+1}(G,H)$.
\end{corollary}

\section{Edge-Transitive Planar Graphs}
\label{sec:edge-transitive}

In this section, we classify planar graphs where all edges receive the same color with respect to $2$-WL.
We call an undirected graph $G$ \emph{edge-transitive} if for all $uv, u'v' \in E(G)$, there is an automorphism $\varphi \colon V(G) \rightarrow V(G)$ with $\varphi(u) = u'$ and $\varphi(v) = v'$.
It is well-known that there are only nine edge-transitive connected planar graphs of minimum degree $3$ \cite{GrunbaumS87}.
Based on this result, one can easily classify all edge-transitive planar graphs.
Clearly, all of these graphs have the property that all edges receive the same color with respect to $2$-WL.
In this section, we show the converse of this statement, i.e., every planar graph in which all edges receive the same color with respect to $2$-WL is edge-transitive.
Towards this goal, we reprove the classification from \cite{GrunbaumS87} relying only on $2$-WL colors.
More precisely, the main result in this section is the following theorem (see also Figure \ref{fig:transitive}).
Since $2$-WL colors directed pairs and it may happen that a pair $(u,v)$ receives a different color than $(v,u)$, it is more convenient to consider directed graphs and demand that all directed edges receive the same color (rather than saying the pair of colors for both orientations is the same for all undirected edges).

\begin{figure}
 \centering
 \begin{subfigure}[b]{.2\linewidth}
  \scalebox{\figscalesmall}{
  \begin{tikzpicture}
   \node[vertex,red!80] (1) at (0,0) {};
   \node[vertex,red!80] (2) at ($(0,0)+(90:1.6)$) {};
   \node[vertex,red!80] (3) at ($(0,0)+(210:1.6)$) {};
   \node[vertex,red!80] (4) at ($(0,0)+(330:1.6)$) {};
   
   \foreach \i in {1,2,3}{
    \foreach \j in {\i,...,4}{
     \draw[line width=1.6pt] (\i) edge (\j);
    }
   }
  \end{tikzpicture}
  }
  \caption{$K_4$}
  \label{fig:graph-k4}
 \end{subfigure}
 \hfill
 \begin{subfigure}[b]{.2\linewidth}
  \scalebox{\figscalesmall}{
  \begin{tikzpicture}[scale = 0.8]
   \node[vertex,red!80] (1) at ($(0,0)+(45:0.8)$) {};
   \node[vertex,red!80] (2) at ($(0,0)+(135:0.8)$) {};
   \node[vertex,red!80] (3) at ($(0,0)+(225:0.8)$) {};
   \node[vertex,red!80] (4) at ($(0,0)+(315:0.8)$) {};
   
   \node[vertex,red!80] (5) at ($(0,0)+(45:2.4)$) {};
   \node[vertex,red!80] (6) at ($(0,0)+(135:2.4)$) {};
   \node[vertex,red!80] (7) at ($(0,0)+(225:2.4)$) {};
   \node[vertex,red!80] (8) at ($(0,0)+(315:2.4)$) {};
   
   \draw[line width=1.6pt] (1) edge (2);
   \draw[line width=1.6pt] (1) edge (4);
   \draw[line width=1.6pt] (2) edge (3);
   \draw[line width=1.6pt] (3) edge (4);
   
   \draw[line width=1.6pt] (5) edge (6);
   \draw[line width=1.6pt] (5) edge (8);
   \draw[line width=1.6pt] (6) edge (7);
   \draw[line width=1.6pt] (7) edge (8);
   
   \draw[line width=1.6pt] (1) edge (5);
   \draw[line width=1.6pt] (2) edge (6);
   \draw[line width=1.6pt] (3) edge (7);
   \draw[line width=1.6pt] (4) edge (8);
   
  \end{tikzpicture}
  }
  \caption{Cube}
  \label{fig:graph-cube}
 \end{subfigure}
 \hfill
 \begin{subfigure}[b]{.2\linewidth}
  \scalebox{\figscalesmall}{
  \begin{tikzpicture}
   \node[vertex,red!80] (1) at ($(0,0)+(30:0.4)$) {};
   \node[vertex,red!80] (2) at ($(0,0)+(150:0.4)$) {};
   \node[vertex,red!80] (3) at ($(0,0)+(270:0.4)$) {};
   
   \node[vertex,red!80] (4) at ($(0,0)+(90:1.6)$) {};
   \node[vertex,red!80] (5) at ($(0,0)+(210:1.6)$) {};
   \node[vertex,red!80] (6) at ($(0,0)+(330:1.6)$) {};
   
   \draw[line width=1.6pt] (1) edge (2);
   \draw[line width=1.6pt] (1) edge (3);
   \draw[line width=1.6pt] (2) edge (3);
   
   \draw[line width=1.6pt] (4) edge (5);
   \draw[line width=1.6pt] (4) edge (6);
   \draw[line width=1.6pt] (5) edge (6);
   
   \draw[line width=1.6pt] (1) edge (4);
   \draw[line width=1.6pt] (1) edge (6);
   \draw[line width=1.6pt] (2) edge (4);
   \draw[line width=1.6pt] (2) edge (5);
   \draw[line width=1.6pt] (3) edge (5);
   \draw[line width=1.6pt] (3) edge (6);
  \end{tikzpicture}
  }
  \caption{Octahedron}
  \label{fig:graph-octahedron}
 \end{subfigure}
 \hfill
 \begin{subfigure}[b]{.3\linewidth}
  \scalebox{\figscalesmall}{
  \begin{tikzpicture}
   \node[vertex,red!80] (1) at ($(0,0)+(54:0.6)$) {};
   \node[vertex,red!80] (2) at ($(0,0)+(126:0.6)$) {};
   \node[vertex,red!80] (3) at ($(0,0)+(198:0.6)$) {};
   \node[vertex,red!80] (4) at ($(0,0)+(270:0.6)$) {};
   \node[vertex,red!80] (5) at ($(0,0)+(342:0.6)$) {};
   
   \node[vertex,red!80] (6) at ($(0,0)+(54:1.2)$) {};
   \node[vertex,red!80] (7) at ($(0,0)+(126:1.2)$) {};
   \node[vertex,red!80] (8) at ($(0,0)+(198:1.2)$) {};
   \node[vertex,red!80] (9) at ($(0,0)+(270:1.2)$) {};
   \node[vertex,red!80] (10) at ($(0,0)+(342:1.2)$) {};
   
   \node[vertex,red!80] (11) at ($(0,0)+(18:1.7)$) {};
   \node[vertex,red!80] (12) at ($(0,0)+(90:1.7)$) {};
   \node[vertex,red!80] (13) at ($(0,0)+(162:1.7)$) {};
   \node[vertex,red!80] (14) at ($(0,0)+(234:1.7)$) {};
   \node[vertex,red!80] (15) at ($(0,0)+(306:1.7)$) {};
   
   \node[vertex,red!80] (16) at ($(0,0)+(18:2.4)$) {};
   \node[vertex,red!80] (17) at ($(0,0)+(90:2.4)$) {};
   \node[vertex,red!80] (18) at ($(0,0)+(162:2.4)$) {};
   \node[vertex,red!80] (19) at ($(0,0)+(234:2.4)$) {};
   \node[vertex,red!80] (20) at ($(0,0)+(306:2.4)$) {};
   
   \foreach \i/\j in {1/2,1/5,2/3,3/4,4/5,1/6,2/7,3/8,4/9,5/10,6/11,6/12,7/12,7/13,8/13,8/14,9/14,9/15,10/11,10/15,11/16,12/17,13/18,14/19,15/20,16/17,16/20,17/18,18/19,19/20}{
    \draw[line width=1.6pt] (\i) edge (\j);
   }
  \end{tikzpicture}
  }
  \caption{Dodecahedron}
  \label{fig:graph-dodecahedron}
 \end{subfigure}
 
 \vspace{0.7cm}
 \begin{subfigure}[b]{.3\linewidth}
  \scalebox{\figscalesmall}{
  \begin{tikzpicture}
   \node[vertex,red!80] (1) at ($(0,0)+(30:0.4)$) {};
   \node[vertex,red!80] (2) at ($(0,0)+(150:0.4)$) {};
   \node[vertex,red!80] (3) at ($(0,0)+(270:0.4)$) {};
   
   \node[vertex,red!80] (4) at ($(0,0)+(90:0.9)$) {};
   \node[vertex,red!80] (5) at ($(0,0)+(210:0.9)$) {};
   \node[vertex,red!80] (6) at ($(0,0)+(330:0.9)$) {};
   
   \node[vertex,red!80] (7) at ($(0,0)+(30:1.1)$) {};
   \node[vertex,red!80] (8) at ($(0,0)+(150:1.1)$) {};
   \node[vertex,red!80] (9) at ($(0,0)+(270:1.1)$) {};
   
   \node[vertex,red!80] (10) at ($(0,0)+(90:2.8)$) {};
   \node[vertex,red!80] (11) at ($(0,0)+(210:2.8)$) {};
   \node[vertex,red!80] (12) at ($(0,0)+(330:2.8)$) {};
   
   \foreach \i/\j in {1/2,1/3,2/3,1/4,1/6,2/4,2/5,3/5,3/6,1/7,4/7,6/7,2/8,4/8,5/8,3/9,5/9,6/9,4/10,7/10,8/10,5/11,8/11,9/11,6/12,7/12,9/12,10/11,10/12,11/12}{
    \draw[line width=1.6pt] (\i) edge (\j);
   }
  \end{tikzpicture}
  }
  \caption{Icosahedron}
  \label{fig:graph-icosahedron}
 \end{subfigure}
 \hfill
 \begin{subfigure}[b]{.23\linewidth}
  \scalebox{\figscalesmall}{
  \begin{tikzpicture}
   \node[vertex,red!80] (1) at ($(0,0)+(45:0.9)$) {};
   \node[vertex,red!80] (2) at ($(0,0)+(135:0.9)$) {};
   \node[vertex,red!80] (3) at ($(0,0)+(225:0.9)$) {};
   \node[vertex,red!80] (4) at ($(0,0)+(315:0.9)$) {};
   
   \node[vertex,red!80] (5) at ($(0,0)+(0:1.2)$) {};
   \node[vertex,red!80] (6) at ($(0,0)+(90:1.2)$) {};
   \node[vertex,red!80] (7) at ($(0,0)+(180:1.2)$) {};
   \node[vertex,red!80] (8) at ($(0,0)+(270:1.2)$) {};
   
   \node[vertex,red!80] (9) at ($(0,0)+(45:2.4)$) {};
   \node[vertex,red!80] (10) at ($(0,0)+(135:2.4)$) {};
   \node[vertex,red!80] (11) at ($(0,0)+(225:2.4)$) {};
   \node[vertex,red!80] (12) at ($(0,0)+(315:2.4)$) {};
   
   \foreach \i/\j in {1/2,1/4,2/3,3/4,1/5,1/6,2/6,2/7,3/7,3/8,4/5,4/8,5/9,5/12,6/9,6/10,7/10,7/11,8/11,8/12,9/10,9/12,10/11,11/12}{
    \draw[line width=1.6pt] (\i) edge (\j);
   }
   
  \end{tikzpicture}
  }
  \caption{Cuboctahedron}
  \label{fig:graph-cuboctahedron}
 \end{subfigure}
 \hfill
 \begin{subfigure}[b]{.37\linewidth}
  \scalebox{\figscalesmall}{
  \begin{tikzpicture}
   \node[vertex,red!80] (1) at ($(0,0)+(54:0.6)$) {};
   \node[vertex,red!80] (2) at ($(0,0)+(126:0.6)$) {};
   \node[vertex,red!80] (3) at ($(0,0)+(198:0.6)$) {};
   \node[vertex,red!80] (4) at ($(0,0)+(270:0.6)$) {};
   \node[vertex,red!80] (5) at ($(0,0)+(342:0.6)$) {};
   
   \node[vertex,red!80] (6) at ($(0,0)+(90:1.0)$) {};
   \node[vertex,red!80] (7) at ($(0,0)+(162:1.0)$) {};
   \node[vertex,red!80] (8) at ($(0,0)+(234:1.0)$) {};
   \node[vertex,red!80] (9) at ($(0,0)+(306:1.0)$) {};
   \node[vertex,red!80] (10) at ($(0,0)+(18:1.0)$) {};
   
   \node[vertex,red!80] (11) at ($(0,0)+(36:1.6)$) {};
   \node[vertex,red!80] (12) at ($(0,0)+(72:1.6)$) {};
   \node[vertex,red!80] (13) at ($(0,0)+(108:1.6)$) {};
   \node[vertex,red!80] (14) at ($(0,0)+(144:1.6)$) {};
   \node[vertex,red!80] (15) at ($(0,0)+(180:1.6)$) {};
   \node[vertex,red!80] (16) at ($(0,0)+(216:1.6)$) {};
   \node[vertex,red!80] (17) at ($(0,0)+(252:1.6)$) {};
   \node[vertex,red!80] (18) at ($(0,0)+(288:1.6)$) {};
   \node[vertex,red!80] (19) at ($(0,0)+(324:1.6)$) {};
   \node[vertex,red!80] (20) at ($(0,0)+(0:1.6)$) {};
   
   \node[vertex,red!80] (21) at ($(0,0)+(54:2.0)$) {};
   \node[vertex,red!80] (22) at ($(0,0)+(126:2.0)$) {};
   \node[vertex,red!80] (23) at ($(0,0)+(198:2.0)$) {};
   \node[vertex,red!80] (24) at ($(0,0)+(270:2.0)$) {};
   \node[vertex,red!80] (25) at ($(0,0)+(342:2.0)$) {};
   
   \node[vertex,red!80] (26) at ($(0,0)+(90:3.2)$) {};
   \node[vertex,red!80] (27) at ($(0,0)+(162:3.2)$) {};
   \node[vertex,red!80] (28) at ($(0,0)+(234:3.2)$) {};
   \node[vertex,red!80] (29) at ($(0,0)+(306:3.2)$) {};
   \node[vertex,red!80] (30) at ($(0,0)+(18:3.2)$) {};
   
   \foreach \i/\j in {1/2,1/5,2/3,3/4,4/5,1/6,1/10,2/6,2/7,3/7,3/8,4/8,4/9,5/9,5/10,6/12,6/13,7/14,7/15,8/16,8/17,9/18,9/19,10/11,10/20,
                      11/12,11/20,12/13,13/14,14/15,15/16,16/17,17/18,18/19,19/20,11/21,12/21,13/22,14/22,15/23,16/23,17/24,18/24,19/25,20/25,
                      21/26,21/30,22/26,22/27,23/27,23/28,24/28,24/29,25/29,25/30,26/27,26/30,27/28,28/29,29/30}{
    \draw[line width=1.6pt] (\i) edge (\j);
   }
  \end{tikzpicture}
  }
  \caption{Icosidodecahedron}
  \label{fig:graph-icosidodecahedron}
 \end{subfigure}
 
 \vspace{0.7cm}
 \begin{subfigure}[b]{.2\linewidth}
  \scalebox{\figscalelarge}{
  \begin{tikzpicture}[scale = 0.8]
   \node[vertex,red!80] (1) at ($(0,0)+(45:0.8)$) {};
   \node[vertex,blue!80] (2) at ($(0,0)+(135:0.8)$) {};
   \node[vertex,red!80] (3) at ($(0,0)+(225:0.8)$) {};
   \node[vertex,blue!80] (4) at ($(0,0)+(315:0.8)$) {};
   
   \node[vertex,blue!80] (5) at ($(0,0)+(45:2.4)$) {};
   \node[vertex,red!80] (6) at ($(0,0)+(135:2.4)$) {};
   \node[vertex,blue!80] (7) at ($(0,0)+(225:2.4)$) {};
   \node[vertex,red!80] (8) at ($(0,0)+(315:2.4)$) {};
   
   \draw[line width=1.6pt] (1) edge (2);
   \draw[line width=1.6pt] (1) edge (4);
   \draw[line width=1.6pt] (2) edge (3);
   \draw[line width=1.6pt] (3) edge (4);
   
   \draw[line width=1.6pt] (5) edge (6);
   \draw[line width=1.6pt] (5) edge (8);
   \draw[line width=1.6pt] (6) edge (7);
   \draw[line width=1.6pt] (7) edge (8);
   
   \draw[line width=1.6pt] (1) edge (5);
   \draw[line width=1.6pt] (2) edge (6);
   \draw[line width=1.6pt] (3) edge (7);
   \draw[line width=1.6pt] (4) edge (8);
   
  \end{tikzpicture}
  }
  \caption{Bicolored cube}
  \label{fig:graph-bicol-cube}
 \end{subfigure}
 \hfill
 \begin{subfigure}[b]{0.3\linewidth}
  \scalebox{\figscalesmall}{
  \begin{tikzpicture}
   \node[vertex,red!80] (1) at (0,0) {};
   
   \node[vertex,blue!80] (2) at ($(0,0)+(30:0.6)$) {};
   \node[vertex,blue!80] (3) at ($(0,0)+(150:0.6)$) {};
   \node[vertex,blue!80] (4) at ($(0,0)+(270:0.6)$) {};
   
   \node[vertex,red!80] (5) at ($(0,0)+(90:0.9)$) {};
   \node[vertex,red!80] (6) at ($(0,0)+(210:0.9)$) {};
   \node[vertex,red!80] (7) at ($(0,0)+(330:0.9)$) {};
   
   \node[vertex,red!80] (8) at ($(0,0)+(30:1.2)$) {};
   \node[vertex,red!80] (9) at ($(0,0)+(150:1.2)$) {};
   \node[vertex,red!80] (10) at ($(0,0)+(270:1.2)$) {};
   
   \node[vertex,blue!80] (11) at ($(0,0)+(90:2.4)$) {};
   \node[vertex,blue!80] (12) at ($(0,0)+(210:2.4)$) {};
   \node[vertex,blue!80] (13) at ($(0,0)+(330:2.4)$) {};
   
   \node[vertex,red!80] (14) at (0,3.2) {};
   
   \foreach \i/\j in {1/2,1/3,1/4,2/5,2/7,3/5,3/6,4/6,4/7,2/8,3/9,4/10,5/11,6/12,7/13,8/11,8/13,9/11,9/12,10/12,10/13,11/14}{
    \draw[line width=1.6pt] (\i) edge (\j);
   }
   \draw[line width=1.6pt,bend left] (12) edge (14);
   \draw[line width=1.6pt,bend right] (13) edge (14);
  \end{tikzpicture}
  }
  \caption{Rhombic dodecahedron}
  \label{fig:graph-rhombic-dodecahedron}
 \end{subfigure}
 \hfill
 \begin{subfigure}[b]{.4\linewidth}
  \scalebox{\figscalesmall}{
  \begin{tikzpicture}
   
   \node[vertex,red!80] (13) at (0,0) {};
   
   \node[vertex,blue!80] (1) at ($(0,0)+(30:0.6)$) {};
   \node[vertex,blue!80] (2) at ($(0,0)+(150:0.6)$) {};
   \node[vertex,blue!80] (3) at ($(0,0)+(270:0.6)$) {};
   
   \node[vertex,red!80] (14) at ($(0,0)+(90:0.6)$) {};
   \node[vertex,red!80] (15) at ($(0,0)+(210:0.6)$) {};
   \node[vertex,red!80] (16) at ($(0,0)+(330:0.6)$) {};
   
   \node[vertex,blue!80] (4) at ($(0,0)+(90:1.2)$) {};
   \node[vertex,blue!80] (5) at ($(0,0)+(210:1.2)$) {};
   \node[vertex,blue!80] (6) at ($(0,0)+(330:1.2)$) {};
   
   \node[vertex,red!80] (17) at ($(0,0)+(0:1.0)$) {};
   \node[vertex,red!80] (18) at ($(0,0)+(60:1.0)$) {};
   \node[vertex,red!80] (19) at ($(0,0)+(120:1.0)$) {};
   \node[vertex,red!80] (20) at ($(0,0)+(180:1.0)$) {};
   \node[vertex,red!80] (21) at ($(0,0)+(240:1.0)$) {};
   \node[vertex,red!80] (22) at ($(0,0)+(300:1.0)$) {};
   
   \node[vertex,blue!80] (7) at ($(0,0)+(30:1.5)$) {};
   \node[vertex,blue!80] (8) at ($(0,0)+(150:1.5)$) {};
   \node[vertex,blue!80] (9) at ($(0,0)+(270:1.5)$) {};
   
   \node[vertex,red!80] (23) at ($(0,0)+(0:1.8)$) {};
   \node[vertex,red!80] (24) at ($(0,0)+(60:1.8)$) {};
   \node[vertex,red!80] (25) at ($(0,0)+(120:1.8)$) {};
   \node[vertex,red!80] (26) at ($(0,0)+(180:1.8)$) {};
   \node[vertex,red!80] (27) at ($(0,0)+(240:1.8)$) {};
   \node[vertex,red!80] (28) at ($(0,0)+(300:1.8)$) {};
   
   \node[vertex,blue!80] (10) at ($(0,0)+(90:3.6)$) {};
   \node[vertex,blue!80] (11) at ($(0,0)+(210:3.6)$) {};
   \node[vertex,blue!80] (12) at ($(0,0)+(330:3.6)$) {};
   
   \node[vertex,red!80] (29) at ($(0,0)+(30:2.4)$) {};
   \node[vertex,red!80] (30) at ($(0,0)+(150:2.4)$) {};
   \node[vertex,red!80] (31) at ($(0,0)+(270:2.4)$) {};
   
   \node[vertex,red!80] (32) at (0,4.4) {};
   
   \foreach \i/\j in {1/13,2/13,3/13,1/14,2/14,2/15,3/15,1/16,3/16,4/14,5/15,6/16,1/17,1/18,2/19,2/20,3/21,3/22,4/18,4/19,5/20,5/21,6/17,6/22,7/17,7/18,8/19,8/20,9/21,9/22,
                      4/24,4/25,5/26,5/27,6/28,6/23,7/23,7/24,8/25,8/26,9/27,9/28,10/24,10/25,11/26,11/27,12/28,12/23,7/29,8/30,9/31,10/29,10/30,11/30,11/31,12/29,12/31,10/32}{
    \draw[line width=1.6pt] (\i) edge (\j);
   }
   \draw[line width=1.6pt,bend left] (11) edge (32);
   \draw[line width=1.6pt,bend right] (12) edge (32);
  \end{tikzpicture}
  }
  \caption{Rhombic triacontahedron}
  \label{fig:graph-rhombic-triacontahedron}
 \end{subfigure}
 \caption{All edge-transitive connected planar graphs of minimum degree $3$.}
 \label{fig:transitive}
\end{figure}
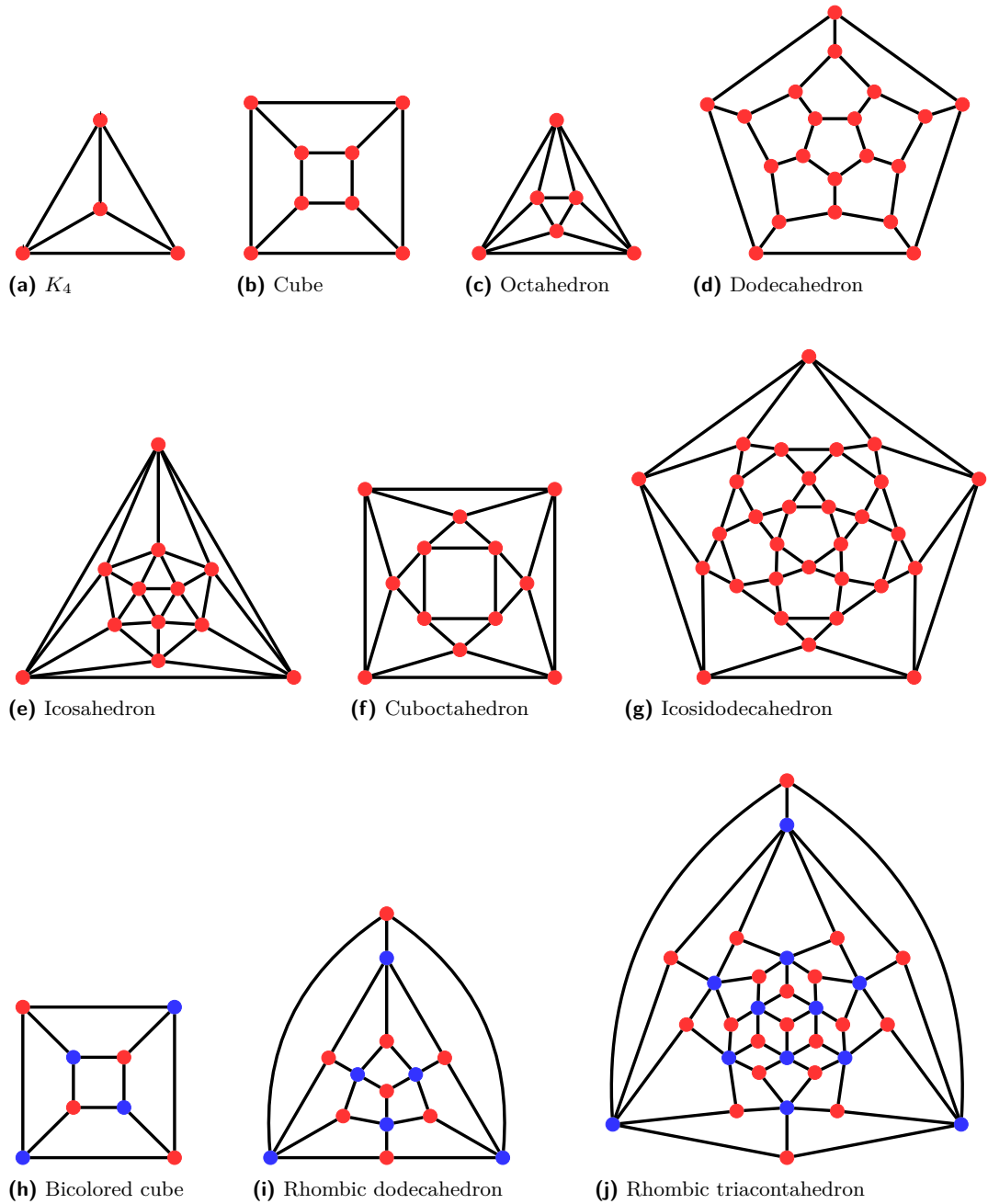

\begin{theorem}
 \label{thm:classification-edge-transitive}
 Let $G$ be a connected planar (directed or undirected) graph of minimum degree at least $3$ such that $\WL{2}{G}(v_1,w_1) = \WL{2}{G}(v_2,w_2)$ for all $(v_1,w_1),(v_2,w_2) \in E(G)$.
 Then one of the following holds:
 \begin{enumerate}[label=(\Alph*)]
  \item\label{item:classification-edge-transitive-1} $G$ is isomorphic to a tetrahedron (Figure \ref{fig:graph-k4}), a cube (Figure \ref{fig:graph-cube}), a dodecahedron (Figure \ref{fig:graph-dodecahedron}), or an icosahedron (Figure \ref{fig:graph-icosahedron}), or
  \item\label{item:classification-edge-transitive-2} the undirected version $\undir{G}$ is isomorphic to an octahedron (Figure \ref{fig:graph-octahedron}), a cuboctahedron (Figure \ref{fig:graph-cuboctahedron}), or an icosidodecahedron (Figure \ref{fig:graph-icosidodecahedron}), or
  \item\label{item:classification-edge-transitive-3} the undirected version $\undir{G}$ is isomorphic to a cube (Figure \ref{fig:graph-bicol-cube}), a rhombic dodecahedron (Figure \ref{fig:graph-rhombic-dodecahedron}), or a rhombic triacontahedron (Figure \ref{fig:graph-rhombic-triacontahedron}).
 \end{enumerate}
\end{theorem}

Note that the classification includes the graphs of all Platonic solids. To prove the theorem, we distinguish two cases.
Let $\chi \coloneqq \WL{2}{G}$ and let $C_V(G,\chi) \coloneqq \{\chi(v,v) \mid v \in V(G)\}$ denote the set of \emph{vertex colors}.
Since $\chi(u,u) = \chi(u',u')$ and $\chi(v,v) = \chi(v',v')$ whenever $\chi(u,v) = \chi(u',v')$, we conclude that $1 \leq |C_V(G,\chi)| \leq 2$.
We first cover the case where $|C_V(G,\chi)| = 1$.
Then $\undir{G}$ is $d$-regular for some $d \geq 3$.
Since $G$ is planar, by Corollary \ref{cor:planar-degree-bound}, we have $d \leq 5$ and thus, $d \in \{3,4,5\}$.
In the following, we analyze these possible cases separately to obtain the graphs listed in Parts \ref{item:classification-edge-transitive-1} and \ref{item:classification-edge-transitive-2}.
Let us remark at this point that obtaining such a classification is more challenging than for edge-transitive graphs.
Indeed, the proofs for edge-transitive graphs highly exploit that the multiset of sizes of faces incident to an edge (and a vertex, respectively) is always the same.
However, we cannot immediately deduce information about the size of faces from considering WL-colors and hence, we cannot directly rely on this type of argument.
Instead, our arguments exploit the fact that $2$-WL can detect $2$-separators \cite{KieferN22} as well as the existence of certain short cycles \cite{FuhlbruckKV21a}.

\begin{theorem}[{\cite[Theorem 3.15]{KieferN22}}]
 \label{thm:3-connected-wl}
 Let $G$ be a connected graph of minimum degree $3$ such that $\WL{2}{G}(v_1,v_1) = \WL{2}{G}(v_2,v_2)$ for all $v_1,v_2 \in V(G)$.
 Then $G$ is $3$-connected.
\end{theorem}

\begin{lemma}[{see, e.g., \cite[Example 2.1 \& Lemma 3.3]{FuhlbruckKV21a}}]
 \label{la:wl-detects-cycles}
 Let $G$ be a graph and let $\ell \in \{3,4,5\}$.
 Also let $v_1v_2,w_1w_2 \in E(G)$ such that $\WL{2}{G}(v_1,v_2) \in \{\WL{2}{G}(w_1,w_2), \WL{2}{G}(w_2,w_1)\}$.
 Finally, suppose there are vertices $v_3,\dots,v_\ell \in V(G)$ such that $v_1,\dots,v_\ell$ are pairwise distinct and $v_2v_3,\dots,v_{\ell-1}v_\ell, v_1v_\ell \in E(G)$ (i.e., $v_1,\dots,v_\ell$ forms an $\ell$-cycle).
 Then there are vertices $w_3,\dots,w_\ell \in V(G)$ such that $w_1,\dots,w_\ell$ are pairwise distinct and $w_2w_3,\dots,w_{\ell-1}v_\ell, w_1w_\ell \in E(G)$ (i.e., $w_1,\dots,w_\ell$ forms an $\ell$-cycle).
\end{lemma}

Also, the following simple observation turns out to be useful in several cases.

\begin{observation}
 \label{obs:neighborhood-degree}
 Let $G$ be a (directed or undirected) graph.
 Also let $v \in V(G)$ and $w_1,w_2 \in N_G(v)$ such that $\WL{2}{G}(v,w_1) \in \{\WL{2}{G}(v,w_2),\WL{2}{G}(w_2,v)\}$.
 Then $|N(v) \cap N(w_1)| = |N(v) \cap N(w_2)|$.
\end{observation}

The next three lemmas cover the three possible values for $d$.
We start with the case $d = 5$.

\begin{lemma}
 \label{la:edge-transitive-degree-5}
 Let $G$ be a connected, directed, planar graph such that
 \begin{enumerate}
  \item\label{item:edge-transitive-degree-5-1} $\deg(v) = 5$ for all $v \in V(G)$,
  \item\label{item:edge-transitive-degree-5-2} $\WL{2}{G}(v_1,v_1) = \WL{2}{G}(v_2,v_2)$ for all $v_1,v_2 \in V(G)$, and
  \item\label{item:edge-transitive-degree-5-3} $\WL{2}{G}(v_1,w_1) = \WL{2}{G}(v_2,w_2)$ for all $(v_1,w_1),(v_2,w_2) \in E(G)$.
 \end{enumerate}
 Then $G$ is isomorphic to an icosahedron (Figure \ref{fig:graph-icosahedron}).
\end{lemma}

\begin{proof}
 First observe that $G$ is undirected (i.e., for all $v,w \in V(G)$, it holds that $(v,w) \in E(G)$ if and only if $(w,v) \in E(G)$), since $\deg(v)$ is odd for all $v \in V(G)$ and Condition~\ref{item:edge-transitive-degree-5-2} holds. Let $v \in V(G)$.
 Then $G[N(v)]$ is regular by Condition \ref{item:edge-transitive-degree-5-3} and Observation~\ref{obs:neighborhood-degree}.
 Let $d \coloneqq |N(v) \cap N(w)|$ for some (and thus all) $w \in N(v)$.
 Then $d$ is even because $\deg(v)$ is odd.
 So $d \in \{0,2,4\}$.
 Moreover, $d \neq 4$, since $G$ is planar (otherwise, $G[N(v)]$ forms a clique of size $5$, which contradicts Theorem \ref{thm:wagner}).
 Also, $d \neq 0$, since $G[N(v)]$ contains at least one edge because, by Euler's formula, $G$ contains a triangle, and Condition \ref{item:edge-transitive-degree-5-2} holds.
 
 So overall, $G[N(v)]$ is a cycle of length $5$ for all $v \in V(G)$.
 This implies that every face of $G$ is a triangle since $G$ is $3$-connected by Theorem \ref{thm:3-connected-wl}.
 By Euler's formula, $|V(G)| = 12$ and $|E(G)| = 20$.
 Consequently, $G$ is isomorphic to an icosahedron.
\end{proof}

Next, we consider the case $d = 3$.

\begin{lemma}
 \label{la:edge-transitive-degree-3}
 Let $G$ be a connected, directed, planar graph such that
 \begin{enumerate}
  \item\label{item:edge-transitive-degree-3-1} $\deg(v) = 3$ for all $v \in V(G)$,
  \item\label{item:edge-transitive-degree-3-2} $\WL{2}{G}(v_1,v_1) = \WL{2}{G}(v_2,v_2)$ for all $v_1,v_2 \in V(G)$, and
  \item\label{item:edge-transitive-degree-3-3} $\WL{2}{G}(v_1,w_1) = \WL{2}{G}(v_2,w_2)$ for all $(v_1,w_1),(v_2,w_2) \in E(G)$.
 \end{enumerate}
 Then $G$ is isomorphic to a tetrahedron (Figure \ref{fig:graph-k4}), a cube (Figure \ref{fig:graph-cube}), or a dodecahedron (Figure \ref{fig:graph-dodecahedron}).
\end{lemma}

\begin{proof}
 As in the proof of the previous lemma, first observe that $G$ is undirected, since $\deg(v)$ is odd for all $v \in V(G)$ and Condition \ref{item:edge-transitive-degree-5-2} holds.
 Also note that $G$ is $3$-connected by Theorem~\ref{thm:3-connected-wl}.
 
 Now first suppose that $G$ contains a triangle and let $v \in V(G)$ be a vertex of a triangle, i.e., $G[N(v)]$ contains at least one edge.
 Also, $G[N(v)]$ is regular by Condition \ref{item:edge-transitive-degree-3-3} and Observation \ref{obs:neighborhood-degree}.
 Since $\deg(v) = 3$, it follows that $G[N(v)]$ is a complete graph on three vertices.
 Hence, $G[N[v]] \cong K_4$.
 Since every vertex in the graph $G[N[v]]$ already has three neighbors in $G[N[v]]$ and $G$ is connected, it holds that $G$ is isomorphic to $K_4$.
 
 Next, assume that $G$ does not contain a triangle, but a cycle of length $4$, i.e., there are vertices $v_1,v_2,v_3,v_4 \in V(G)$ such that $v_1v_2,v_2v_3,v_3v_4,v_1v_4 \in E(G)$.
 Note that $v_1v_3,v_2v_4 \notin E(G)$ since $G$ does not contain a triangle.
 Let $v \coloneqq v_1$.
 Since $\WL{2}{G}(v,w_1) = \WL{2}{G}(v,w_2)$ for all $w_1,w_2 \in N(v)$, it follows that
 \[\{\WL{2}{G}(w_1,w_2),\WL{2}{G}(w_2,w_1)\} = \{\WL{2}{G}(w_3,w_4),\WL{2}{G}(w_4,w_3)\}\]
 for all $w_1,w_2,w_3,w_4 \in N(v)$ such that $w_1 \neq w_2$ and $w_3 \neq w_4$.
 But this means that $|N(w_1) \cap N(w_2)| \geq 2$ for all distinct $w_1,w_2 \in N(v)$.
 
 Since $G$ is $3$-connected and $3$-regular, one can easily verify that there is no vertex $w$ such that $|N(v) \cap N(w)| = 3$.
 Let $N(v) \eqqcolon \{u_1,u_2,u_3\}$.
 Hence, there are $w_{1,2},w_{1,3},w_{2,3} \in V(G) \setminus \{v,u_1,u_2,u_3\}$ such that $w_{i,j}u_i, w_{i,j}u_j \in E(G)$ for all $i,j \in \{1,2,3\}$.
 Since $G$ is $3$-connected, it follows that all faces incident to $v$ have size $4$.
 Actually, by Condition \ref{item:edge-transitive-degree-3-2} and Lemma \ref{la:wl-detects-cycles}, this implies that all faces have size $4$.
 By Euler's formula, $|V(G)| = 8$ and thus, $G$ is isomorphic to the cube.
 
 Finally, suppose that $G$ neither contains a triangle nor a cycle of length $4$.
 Then $G$ contains a cycle of length $5$ by Euler's formula.
 Let $v \in V(G)$ be a vertex contained in a cycle of length $5$.
 Also let $N(v) \eqqcolon \{u_1,u_2,u_3\}$.
 Without loss of generality, assume that there are $w_{1,2},x_{1,2} \in V(G)$ such that $(v,u_1,w_{1,2},x_{1,2},u_2,v)$ forms a cycle of length $5$.
 
 Since $\WL{2}{G}(v,u_i) = \WL{2}{G}(v,u_j)$ holds for all $i,j \in [3]$, it follows that
 \[\{\WL{2}{G}(w_1,w_2),\WL{2}{G}(w_2,w_1)\} = \{\WL{2}{G}(w_3,w_4),\WL{2}{G}(w_4,w_3)\}\]
 for all $w_1,w_2,w_3,w_4 \in N(v)$ such that $w_1 \neq w_2$ and $w_3 \neq w_4$.
 But this implies that, for all $i,j \in [3]$ with $i < j \in [3]$, there are $w_{i,j},x_{i,j} \in V(G)$ such that $(v,u_i,w_{i,j},x_{i,j},u_j,v)$ forms a cycle of length $5$.
 Now, for each of the three cycles of length $5$, three of its vertices already have three neighbors.
 Since $G$ is $3$-connected, this implies that all three faces incident to $v$ have size $5$.
 Actually, by Condition \ref{item:edge-transitive-degree-3-2} and Lemma \ref{la:wl-detects-cycles}, this implies that all faces have size $5$.
 By Euler's formula, $|V(G)| = 20$.
 Moreover, $G$ being $3$-connected also implies that $|\{v,u_1,u_2,u_3,w_{1,2},x_{1,2},w_{1,3},x_{1,3},w_{2,3},x_{2,3}\}| = 10$, i.e., all the listed vertices are distinct.
 Now, it is easy to check that $G$ is isomorphic to the dodecahedron.
\end{proof}

Observe that Lemmas \ref{la:edge-transitive-degree-5} and \ref{la:edge-transitive-degree-3} together give the graphs listed in Part \ref{item:classification-edge-transitive-1}.
The next lemma covers the most involved case $d = 4$, which leads to the graphs listed in Part \ref{item:classification-edge-transitive-2}.

\begin{lemma}
 \label{la:edge-transitive-degree-4}
 Let $G$ be a connected, directed, planar graph such that
 \begin{enumerate}
  \item\label{item:edge-transitive-degree-4-1} $\deg(v) = 4$ for all $v \in V(G)$,
  \item\label{item:edge-transitive-degree-4-2} $\WL{2}{G}(v_1,v_1) = \WL{2}{G}(v_2,v_2)$ for all $v_1,v_2 \in V(G)$, and 
  \item\label{item:edge-transitive-degree-4-3} $\WL{2}{G}(v_1,w_1) = \WL{2}{G}(v_2,w_2)$ for all $(v_1,w_1),(v_2,w_2) \in E(G)$.
 \end{enumerate}
 Then $\undir{G}$ is isomorphic to an octahedron (Figure \ref{fig:graph-octahedron}), a cuboctahedron (Figure \ref{fig:graph-cuboctahedron}), or an icosidodecahedron (Figure \ref{fig:graph-icosidodecahedron}) (where $\undir{G}$ denotes the undirected version of $G$).
\end{lemma}

Before diving into the proof, let us remark that the restriction to $\undir{G}$ is necessary and there are in total four directed versions of the graphs listed in the lemma that satisfy all the requirements.
For the sake of completeness, these graphs are listed in Figure \ref{fig:as}.
This list can for example be obtained from existing databases (see, e.g., \cite{BambergHL19}) that list all graphs $G$ (up to a certain number of vertices) for which $\WL{2}{G}(v_1,v_1) = \WL{2}{G}(v_2,v_2)$ for all $v_1,v_2 \in V(G)$.

\begin{proof}
 The graph $\undir{G}$ is $3$-connected by Theorem \ref{thm:3-connected-wl}. By Corollary \ref{cor:planar-degree-bound}, $\undir{G}$ contains a triangle. 
 Let $v \in V(G)$ be a vertex of such a triangle. Then $\undir{G}[N(v)]$ contains at least one edge.
 Moreover, $\undir{G}[N(v)]$ is regular by Property~\ref{item:edge-transitive-degree-4-3} and Observation~\ref{obs:neighborhood-degree}.
 Let $d \in \{1,2,3\}$ such that $\undir{G}[N(v)]$ is $d$-regular.
 Since $G$ is planar, it must hold that $d \leq 2$, otherwise $\undir{G}[N[v]]$ forms a clique of size $5$, which contradicts Theorem \ref{thm:wagner}.

 Suppose that $d = 2$.
 Then, being $3$-connected, $\undir{G}$ is a triangulation, i.e., every face of $\undir{G}$ is a triangle.
 Thus, $3f = 2e = 4n$ and hence, by Euler's formula, $n = 6$.
 The octahedron is the only connected graph with those parameters.

 Now suppose $d = 1$, i.e., $\undir{G}[N(v)]$ is a matching.
 This means that $v$ is contained in exactly two $3$-cycles of $\undir{G}$.
 Thus, every vertex is contained in exactly two $3$-cycles and every edge is contained in exactly one $3$-cycle of $\undir{G}$ (since an edge contained in two $3$-cycles would yield a vertex $v'$ with a neighbor of degree at least $2$ in $N(v')$).

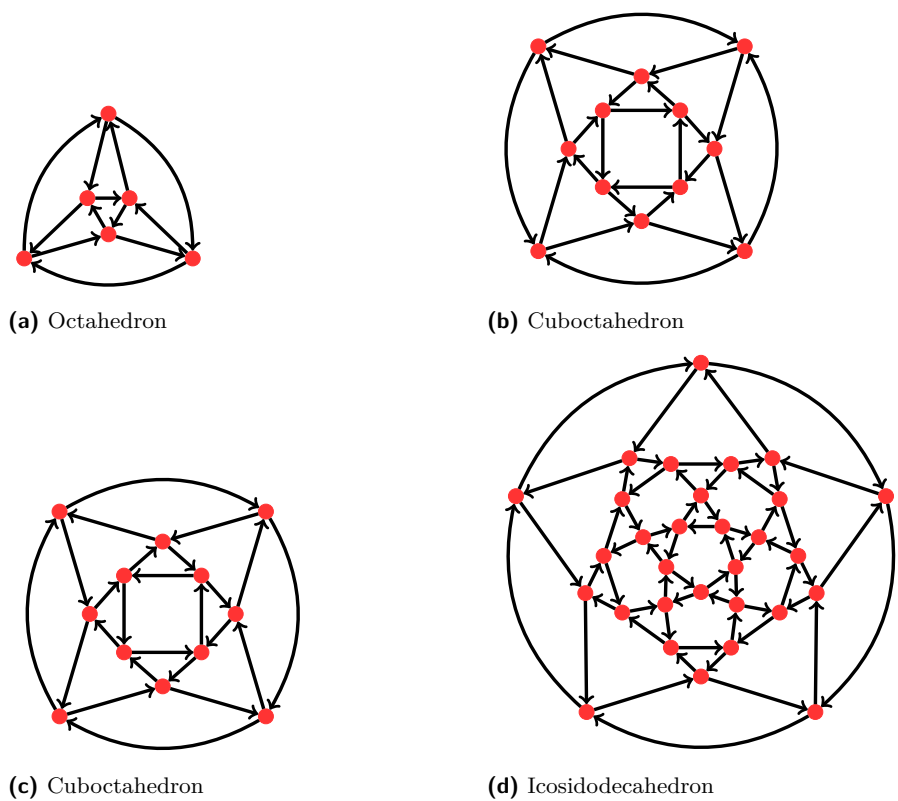
\begin{figure}
 \centering
 \begin{subfigure}[b]{.4\linewidth}
  \scalebox{\figscalesmall}{
  \begin{tikzpicture}
   \node[vertex,red!80] (1) at ($(0,0)+(30:0.4)$) {};
   \node[vertex,red!80] (2) at ($(0,0)+(150:0.4)$) {};
   \node[vertex,red!80] (3) at ($(0,0)+(270:0.4)$) {};
   
   \node[vertex,red!80] (4) at ($(0,0)+(90:1.6)$) {};
   \node[vertex,red!80] (5) at ($(0,0)+(210:1.6)$) {};
   \node[vertex,red!80] (6) at ($(0,0)+(330:1.6)$) {};
   
   \draw[line width=1.6pt,->] (2) edge (1);
   \draw[line width=1.6pt,->] (1) edge (3);
   \draw[line width=1.6pt,->] (3) edge (2);
   
   \draw[line width=1.6pt,->,bend left] (5) edge (4);
   \draw[line width=1.6pt,->,bend left] (4) edge (6);
   \draw[line width=1.6pt,->,bend left] (6) edge (5);
   
   \draw[line width=1.6pt,->] (1) edge (4);
   \draw[line width=1.6pt,->] (6) edge (1);
   \draw[line width=1.6pt,->] (4) edge (2);
   \draw[line width=1.6pt,->] (2) edge (5);
   \draw[line width=1.6pt,->] (5) edge (3);
   \draw[line width=1.6pt,->] (3) edge (6);
  \end{tikzpicture}
  }
  \caption{Octahedron}
  \label{fig:graph-octahedron-dir}
 \end{subfigure}
 \hfill
 \begin{subfigure}[b]{.55\linewidth}
  \scalebox{\figscalesmall}{
  \begin{tikzpicture}
   \node[vertex,red!80] (1) at ($(0,0)+(45:0.9)$) {};
   \node[vertex,red!80] (2) at ($(0,0)+(135:0.9)$) {};
   \node[vertex,red!80] (3) at ($(0,0)+(225:0.9)$) {};
   \node[vertex,red!80] (4) at ($(0,0)+(315:0.9)$) {};
   
   \node[vertex,red!80] (5) at ($(0,0)+(0:1.2)$) {};
   \node[vertex,red!80] (6) at ($(0,0)+(90:1.2)$) {};
   \node[vertex,red!80] (7) at ($(0,0)+(180:1.2)$) {};
   \node[vertex,red!80] (8) at ($(0,0)+(270:1.2)$) {};
   
   \node[vertex,red!80] (9) at ($(0,0)+(45:2.4)$) {};
   \node[vertex,red!80] (10) at ($(0,0)+(135:2.4)$) {};
   \node[vertex,red!80] (11) at ($(0,0)+(225:2.4)$) {};
   \node[vertex,red!80] (12) at ($(0,0)+(315:2.4)$) {};
   
   \foreach \i/\j in {2/1,4/1,2/3,4/3,1/5,1/6,6/2,7/2,3/7,3/8,5/4,8/4,9/5,5/12,9/6,6/10,7/10,11/7,11/8,8/12}{
    \draw[line width=1.6pt,->] (\i) edge (\j);
   }
   \foreach \i/\j in {10/9,12/11}{
    \draw[line width=1.6pt,->,bend left] (\i) edge (\j);
   }
   \foreach \i/\j in {12/9,10/11}{
    \draw[line width=1.6pt,->,bend right] (\i) edge (\j);
   }
   
  \end{tikzpicture}
  }
  \caption{Cuboctahedron}
  \label{fig:graph-cuboctahedron-dir-1}
 \end{subfigure}
 
 \vspace{1.5ex}
 \begin{subfigure}[b]{.4\linewidth}
  \scalebox{\figscalesmall}{
  \begin{tikzpicture}
   \node[vertex,red!80] (1) at ($(0,0)+(45:0.9)$) {};
   \node[vertex,red!80] (2) at ($(0,0)+(135:0.9)$) {};
   \node[vertex,red!80] (3) at ($(0,0)+(225:0.9)$) {};
   \node[vertex,red!80] (4) at ($(0,0)+(315:0.9)$) {};
   
   \node[vertex,red!80] (5) at ($(0,0)+(0:1.2)$) {};
   \node[vertex,red!80] (6) at ($(0,0)+(90:1.2)$) {};
   \node[vertex,red!80] (7) at ($(0,0)+(180:1.2)$) {};
   \node[vertex,red!80] (8) at ($(0,0)+(270:1.2)$) {};
   
   \node[vertex,red!80] (9) at ($(0,0)+(45:2.4)$) {};
   \node[vertex,red!80] (10) at ($(0,0)+(135:2.4)$) {};
   \node[vertex,red!80] (11) at ($(0,0)+(225:2.4)$) {};
   \node[vertex,red!80] (12) at ($(0,0)+(315:2.4)$) {};
   
   \foreach \i/\j in {1/2,4/1,2/3,3/4,1/5,6/1,2/6,7/2,3/7,8/3,5/4,4/8,5/9,12/5,9/6,6/10,10/7,7/11,11/8,8/12}{
    \draw[line width=1.6pt,->] (\i) edge (\j);
   }
   \foreach \i/\j in {10/9,9/12,11/10,12/11}{
    \draw[line width=1.6pt,->,bend left] (\i) edge (\j);
   }
   
  \end{tikzpicture}
  }
  \caption{Cuboctahedron}
  \label{fig:graph-cuboctahedron-dir-2}
 \end{subfigure}
 \hfill
 \begin{subfigure}[b]{.55\linewidth}
  \scalebox{\figscalesmall}{
  \begin{tikzpicture}
   \node[vertex,red!80] (1) at ($(0,0)+(54:0.6)$) {};
   \node[vertex,red!80] (2) at ($(0,0)+(126:0.6)$) {};
   \node[vertex,red!80] (3) at ($(0,0)+(198:0.6)$) {};
   \node[vertex,red!80] (4) at ($(0,0)+(270:0.6)$) {};
   \node[vertex,red!80] (5) at ($(0,0)+(342:0.6)$) {};
   
   \node[vertex,red!80] (6) at ($(0,0)+(90:1.0)$) {};
   \node[vertex,red!80] (7) at ($(0,0)+(162:1.0)$) {};
   \node[vertex,red!80] (8) at ($(0,0)+(234:1.0)$) {};
   \node[vertex,red!80] (9) at ($(0,0)+(306:1.0)$) {};
   \node[vertex,red!80] (10) at ($(0,0)+(18:1.0)$) {};
   
   \node[vertex,red!80] (11) at ($(0,0)+(36:1.6)$) {};
   \node[vertex,red!80] (12) at ($(0,0)+(72:1.6)$) {};
   \node[vertex,red!80] (13) at ($(0,0)+(108:1.6)$) {};
   \node[vertex,red!80] (14) at ($(0,0)+(144:1.6)$) {};
   \node[vertex,red!80] (15) at ($(0,0)+(180:1.6)$) {};
   \node[vertex,red!80] (16) at ($(0,0)+(216:1.6)$) {};
   \node[vertex,red!80] (17) at ($(0,0)+(252:1.6)$) {};
   \node[vertex,red!80] (18) at ($(0,0)+(288:1.6)$) {};
   \node[vertex,red!80] (19) at ($(0,0)+(324:1.6)$) {};
   \node[vertex,red!80] (20) at ($(0,0)+(0:1.6)$) {};
   
   \node[vertex,red!80] (21) at ($(0,0)+(54:2.0)$) {};
   \node[vertex,red!80] (22) at ($(0,0)+(126:2.0)$) {};
   \node[vertex,red!80] (23) at ($(0,0)+(198:2.0)$) {};
   \node[vertex,red!80] (24) at ($(0,0)+(270:2.0)$) {};
   \node[vertex,red!80] (25) at ($(0,0)+(342:2.0)$) {};
   
   \node[vertex,red!80] (26) at ($(0,0)+(90:3.2)$) {};
   \node[vertex,red!80] (27) at ($(0,0)+(162:3.2)$) {};
   \node[vertex,red!80] (28) at ($(0,0)+(234:3.2)$) {};
   \node[vertex,red!80] (29) at ($(0,0)+(306:3.2)$) {};
   \node[vertex,red!80] (30) at ($(0,0)+(18:3.2)$) {};
   
   \foreach \i/\j in {1/2,5/1,2/3,3/4,4/5,6/1,1/10,2/6,7/2,3/7,8/3,4/8,9/4,5/9,10/5,12/6,6/13,14/7,7/15,16/8,8/17,18/9,9/19,10/11,20/10,
                      11/12,11/20,13/12,13/14,15/14,15/16,17/16,17/18,19/18,19/20,21/11,12/21,22/13,14/22,23/15,16/23,24/17,18/24,25/19,20/25,
                      21/26,30/21,26/22,22/27,27/23,23/28,28/24,24/29,29/25,25/30}{
    \draw[line width=1.6pt,->] (\i) edge (\j);
   }
   \foreach \i/\j in {27/26,26/30,28/27,29/28,30/29}{
    \draw[line width=1.6pt,->,bend left] (\i) edge (\j);
   }
  \end{tikzpicture}
  }
  \caption{Icosidodecahedron}
  \label{fig:graph-icosidodecahedron-dir}
 \end{subfigure}
 
 \caption{Directed edge-transitive planar graphs.}
 \label{fig:as}
\end{figure}

 \begin{claim}
  \label{cl:face-cycles-3-4}
  In $\undir{G}$, for $\ell \in \{3,4\}$, every cycle of length $\ell$ is facial.
 \end{claim}

 \begin{claimproof}
  Consider a $3$-cycle in $\undir{G}$ with vertices $v_1,v_2,v_3$ and suppose that it is not facial.
  The cycle bounds two areas $A_1$ and $A_2$ and it suffices to show that one of them is empty.
  Suppose this is not the case. Every $v_i$ has two neighbors $v'_i, v''_i \notin \{v_1,v_2,v_3\}$ and all of these six vertices are pairwise distinct, because otherwise, an edge $v_iv_{i+1}$ of $\undir{G}$ would be contained in two $3$-cycles. Since $\undir{G}[v_i]$ induces a matching, it holds for all $i \in [3]$ that either both $v'_i$ and $v''_i$ lie in $A_1$ or both lie in $A_2$. Let $A \in \{A_1,A_2\}$ be the area that contains only $v'_j$ and $v''_j$ for a single $j$.
  Then $v_j$ separates $v'_j$ from $v'_{j+1}$, contradicting the $3$-connectedness of $\undir{G}$.
  Thus, the $3$-cycle must be facial.

  Consider now a $4$-cycle in $\undir{G}$ with vertices $v_1,v_2,v_3,v_4$ and suppose it is not facial.
  For $i \in [4]$, let $v'_i$ be the unique common neighbor of $v_i$ and $v_{i+1}$ (indices taken modulo~$4$).
  Because no edge is contained in two $3$-cycles, the $4$-cycle is induced and all the $v'_i$ are pairwise distinct.
  There must be an area $A \in \{A_1,A_2\}$ that contains at least one and at most two of the $v'_i$.
  In the first case, we obtain a contradiction just as for the $3$-cycle.
  In the second case, let the two vertices in $A$ be $w_1$ and $w_2$.
  Then $w_1$ must have a neighbor $w'$ not contained in $\{v_1,v_2,v_3,v_4,w_2\}$, which is separated from $\{v'_1,v'_2,v'_3,v'_4\} \setminus \{w_1,w_2\}$ by $\{w_1,w_2\}$, again contradicting the $3$-connectedness of $\undir{G}$.
  Thus, the cycle is facial.
 \end{claimproof}

 We now show that $\undir{G}$ contains a cycle of length at most $5$.
 Let $n \coloneqq |V(G)|$. For $i \in [n]$, let $f_i$ be the number of facial cycles of length $i$.
 By Euler's formula, it holds that $n - 2n + \sum_{i=3}^n f_i = 2$.
 Since every vertex participates in exactly two facial $3$-cycles, we have $f_3 = \frac{2n}{3}$.
 Furthermore, since every vertex participates in exactly two facial cycles of length larger than $3$, we obtain that $\sum_{i=4}^n f_i = 2 + \frac{n}{3}$ holds.
 If $\undir{G}$ contained no cycle of length $4$ and no cycle of length $5$, this would imply that $\sum_{i=6}^n f_i = 2 + \frac{n}{3}$.
 However, $\sum_{i=6}^n 6 \cdot f_i \leq \sum_{i=6}^n i \cdot f_i \leq 2n$.
 Hence, $\sum_{i=6}^n f_i \leq \frac{2n}{6}$, which gives $2 + \frac{n}{3} \leq \frac{n}{3}$, a contradiction.
 Thus, $\undir{G}$ contains a cycle of length $4$ or $5$.

 Suppose that $\undir{G}$ contains a $4$-cycle.
 Due to Lemma \ref{la:wl-detects-cycles}, Claim \ref{cl:face-cycles-3-4}, and Property \ref{item:edge-transitive-degree-4-3}, this implies that, in $\undir{G}$, every edge is incident to one facial cycle of length $3$ and one facial cycle of length $4$.
 Therefore, every vertex of $\undir{G}$ is incident to exactly two facial cycles of length $3$ and two facial cycles of length $4$.
 Let $f_3$ and $f_4$ be the number of facial cycles of length $3$ and $4$, respectively.
 Then $f = f_3 + f_4$ and $3f_3 = 2n = 4f_4$.
 Also $2e = 4n$.
 Thus, by Euler's formula $n-2n+2n/3+n/2 = 2$.
 Hence $n = 12$ and $f_3 = 8$.
 The cuboctahedron is the only $4$-regular $3$-connected planar graph with these parameters.

 Now suppose that $\undir{G}$ contains a $5$-cycle and no $4$-cycle.
 In this case, we can extend the previous claim also to $5$-cycles.

 \begin{claim}
  \label{cl:face-cycles-5}
  In $\undir{G}$, every cycle of length $5$ is facial.
 \end{claim}

 \begin{claimproof}
  Consider a $5$-cycle in $\undir{G}$ with vertices $v_1,v_2,v_3,v_4,v_5$.
  For $i \in [5]$, let $v'_i$ be the unique common neighbor of $v_i$ and $v_{i+1}$ (indices taken modulo~$5$).
  Since there are no $4$-cycles in $\undir{G}$, it holds that $\{v_i \mid i \in [5]\} \cap \{v'_i \mid i \in [5]\} = \emptyset$.
  Also, all the $v'_i$ are pairwise distinct, since otherwise, there would be a $j$ with $v'_j = v'_{j+1}$ or $v'_j = v'_{j+2}$ (indices modulo $5$), which would yield an edge contained in two $3$-cycles.
  
  As before, if the $5$-cycle is not facial, there must be an area $A$ that contains at least one and at most two of the $v'_i$.
  In both cases, we obtain a contradiction just as in the proof of Claim \ref{cl:face-cycles-3-4}.
  Thus, the $5$-cycle is facial.
 \end{claimproof}

 Due to Lemma \ref{la:wl-detects-cycles}, Claim \ref{cl:face-cycles-5}, and Property \ref{item:edge-transitive-degree-4-3}, this implies that in $\undir{G}$, every edge is incident to one facial cycle of length $3$ and one facial cycle of length $5$.
 Therefore, every vertex of $\undir{G}$ is incident to exactly two facial cycles of length $3$ and two facial cycles of length $5$.
 Let $f_3$ and $f_5$ be the number of facial cycles of length $3$ and $5$, respectively.
 Then $f = f_3 + f_5$ and $3f_3 = 2n = 5f_5$. Also $2e = 4n$.
 Thus, by Euler's formula $n-2n+2n/3+2n/5 = 2$. Hence $n = 30$ and $f_3 = 20$.
 The icosidodecahedron is the only $4$-regular $3$-connected planar graph with these parameters.
\end{proof}

Overall, this completes the case $|C_V(G,\chi)| = 1$.
Next, we turn to the case that $|C_V(G,\chi)| = 2$.
This case can be reduced to the previous case by defining an auxiliary graph on one of the two vertex color classes.

\begin{lemma}
 \label{la:edge-transitive-two-vertex-colors}
 Let $G$ be a connected, directed, planar graph of minimum degree at least $3$ such that
 \begin{enumerate}
  \item\label{item:edge-transitive-two-vertex-colors-1} $|\{\WL{2}{G}(v_1,v_1) \mid v_1 \in V(G)\}| \geq 2$, and
  \item\label{item:edge-transitive-two-vertex-colors-2} $\WL{2}{G}(v_1,w_1) = \WL{2}{G}(v_2,w_2)$ for all $(v_1,w_1),(v_2,w_2) \in E(G)$.
 \end{enumerate}
 Then $\undir{G}$ is isomorphic to a cube (Figure \ref{fig:graph-bicol-cube}), a rhombic dodecahedron (Figure \ref{fig:graph-rhombic-dodecahedron}), or a rhombic triacontahedron (Figure \ref{fig:graph-rhombic-triacontahedron}).
\end{lemma}

\begin{proof}
 From Conditions \ref{item:edge-transitive-two-vertex-colors-1} and \ref{item:edge-transitive-two-vertex-colors-2}, it follows that there is a bipartition $(V,W)$ of $V(G)$ such that $v \in V$ and $w \in W$ holds for each edge $(v,w) \in E(G)$.
 Moreover, $\WL{2}{G}(v_1,v_1) = \WL{2}{G}(v_2,v_2)$ for all $v_1,v_2 \in V$ and $\WL{2}{G}(w_1,w_1) = \WL{2}{G}(w_2,w_2)$ for all $w_1,w_2 \in W$.

Since $G$ is planar and bipartite, there exists a $v \in V(G)$ such that $\deg_G(v) = 3$ by Corollary \ref{cor:planar-degree-bound}.
 Without loss of generality, suppose that $v \in V$.
 This means that $\deg_G(v) = 3$ holds for all $v \in V$.
 
 Define $H$ to be the graph with $V(H) \coloneqq W$ and
 \[E(H) \coloneqq \{w_1w_2 \mid \exists v \in V \colon (v,w_1),(v,w_2) \in E(G)\}.\]
 Since $\deg_G(v) = 3$ holds for all $v \in V$, the graph $H$ is planar (a plane embedding of $H$ can directly be obtained from an embedding of $G$). Moreover,
 \[\{\WL{2}{G}(w_1,w_2),\WL{2}{G}(w_2,w_1)\} = \{\WL{2}{G}(w_3,w_4),\WL{2}{G}(w_4,w_3)\}\]
 for all edges $(w_1,w_2),(w_3,w_4) \in E(H)$.
 Thus, by Lemmas \ref{la:edge-transitive-degree-5}, \ref{la:edge-transitive-degree-3} and \ref{la:edge-transitive-degree-4}, $H$ is one of the edge-transitive graphs depicted in Figure \ref{fig:graph-k4} -- \ref{fig:graph-icosidodecahedron}.
 From the definition of $H$, it follows that for each $w_1w_2 \in E(H)$, there is a $w_3 \in W$ such that $w_1w_3,w_2w_3 \in E(H)$.
 Hence, $H$ is neither isomorphic to a cube nor a dodecahedron.
 
 Now $G$ can be recovered from $H$ by placing in each triangular face of $H$ a vertex of degree~$3$ adjacent to the three vertices of the face.
 Via a case-by-case analysis, it can be checked that $G$ is one of the graphs listed in the statement of the lemma (recall that $G$ is supposed to have minimum degree at least $3$).
\end{proof}

Observe that Lemma \ref{la:edge-transitive-two-vertex-colors} lists the graphs from Part \ref{item:classification-edge-transitive-3}.
Here, it is notable that the cube appears for a second time because it is bipartite and directing all edges from one bipartition class to the other one also leads to an edge-transitive graph (see Figure \ref{fig:graph-bicol-cube}).

Combining all the lemmas proved above now allows us to conclude Theorem \ref{thm:classification-edge-transitive}.

\begin{proof}[Proof of Theorem \ref{thm:classification-edge-transitive}]
 Since $\WL{2}{G}(v_1,w_1) = \WL{2}{G}(v_2,w_2)$ for all $(v_1,w_1),(v_2,w_2) \in E(G)$ (and there are no isolated vertices), it follows that
 \[|\{\WL{2}{G}(v,v) \mid v \in V(G)\}| \leq 2.\]
 First suppose that $|\{\WL{2}{G}(v,v) \mid v \in V(G)\}| = 1$.
 Then $G$ is $d$-regular for some number $d \geq 3$.

 Since $G$ is planar, we know from Corollary \ref{cor:planar-degree-bound} that $d \leq 5$.
 So $G$ is isomorphic to one of the graphs listed in Parts \ref{item:classification-edge-transitive-1} and \ref{item:classification-edge-transitive-2} by Lemmas \ref{la:edge-transitive-degree-5}, \ref{la:edge-transitive-degree-3}, and \ref{la:edge-transitive-degree-4}.
 
 Otherwise, $|\{\WL{2}{G}(v,v) \mid v \in V(G)\}| = 2$.
 Then $G$ is isomorphic to one of the graphs listed in Part \ref{item:classification-edge-transitive-3} by Lemma \ref{la:edge-transitive-two-vertex-colors}.
\end{proof}

In Theorem \ref{thm:classification-edge-transitive}, we restrict ourselves to graphs that are connected and have minimum degree at least $3$.
Both of these restrictions can easily be lifted as follows.
Let us first consider the restriction on the degree and let $G$ be a connected planar graph such that $\WL{2}{G}(v_1,w_1) = \WL{2}{G}(v_2,w_2)$ holds for all $(v_1,w_1),(v_2,w_2) \in E(G)$.
If $G$ has maximum degree~$2$ or contains a vertex of degree at most~$1$, then it is easy to see that $G$ is either a cycle or isomorphic to a star $K_{1,h}$ for some $h \geq 0$ ($h = 0$ covers the special case that $G$ consists of a single vertex).

\begin{definition}
 Let $H$ be a graph and $s \geq 1$.
 The \emph{$s$-subdivision of $H$} is the graph $H^{(s)}$ obtained from $H$ by replacing each edge with $s$ parallel paths of length $2$.
 Formally, $H^{(s)}$ is the graph with vertex set $V(H^{(s)}) \coloneqq V(H) \uplus (E(H) \times [s])$ and edge set
 \[E(H^{(s)}) \coloneqq \Big\{v(e,i) \;\Big\vert\; e \in E(H), v \in e, i \in [s]\Big\}.\]
\end{definition}

In the remaining case, $G$ has maximum degree at least $3$ and minimum degree at least $2$.
Then it is easy to see that $G$ is one of the graphs from Theorem \ref{thm:classification-edge-transitive}, or an $s$-subdivision of
\begin{itemize}
 \item one of the graphs listed in Parts \ref{item:classification-edge-transitive-1} and \ref{item:classification-edge-transitive-2},
 \item a cycle $C_\ell$ for some $\ell \geq 3$, or
 \item the complete graph on two vertices $K_2$,
\end{itemize}
for some $s \geq 1$.

Finally, if $G$ is not connected, then it is isomorphic to the disjoint union of $\ell$ copies of one of its connected components for some $\ell \geq 2$, because all graphs listed above can be distinguished from each other by $2$-WL.
Actually, it can be checked that all of the graphs are even identified by $2$-WL.
Overall, this gives the following corollary.

\begin{corollary}
 \label{cor:edge-orbits-for-edge-transitive}
 Let $G$ be a directed planar graph such that $\{\!\{\WL{2}{G}(v,w),\WL{2}{G}(w,v)\}\!\} = \{\!\{\WL{2}{G}(v',w'),\WL{2}{G}(w',v')\}\!\}$ holds for all $(v,w),(v',w') \in E(G)$.
 Then $2$-WL determines arc orbits on $G$.
 In particular, $G$ is edge-transitive.
\end{corollary}

\section{Graphs Induced by a Single Edge Color}
\label{sec:types}

After considering planar graphs with a single edge color with respect to $2$-WL, we now wish to analyze the $2$-WL coloring of arbitrary planar graphs.
Since, by \cite{KieferN22}, the algorithm $2$-WL implicitly computes the decomposition of a graph into $3$-connected components\footnote{For the formal and quite technical definition of this notion, we refer to \cite{KieferN22}.}, understanding $2$-WL on planar graphs essentially amounts to a study of $3$-connected planar graphs.
Hence, we restrict our attention to those graphs.

\subsection{Notation and Tools}

We start by giving further notation as well as providing basic tools for analyzing the $2$-WL colorings.

Let $G$ be a graph and let $\chi \coloneqq \WL{2}{G}$ denote the coloring computed by $2$-WL on $G$. We define $C_V = C_V(G,\chi) \coloneqq \{\chi(v,v) \mid v \in V(G)\}$ to be the set of \emph{vertex colors}.
Similarly, we define $C_E = C_E(G,\chi) \coloneqq \{\chi(v,w) \mid vw \in E(G)\}$ to be the set of \emph{edge colors}.
For a vertex $v \in V(G)$, we denote by $[v]_\chi$ the color class of $v$, i.e., $[v]_\chi \coloneqq \{w \in V(G) \mid \chi(v,v) = \chi(w,w)\}$.
Also, the same notation is used for any vertex coloring $\lambda \colon V(G) \rightarrow C$, i.e., $[v]_\lambda \coloneqq \{w \in V(G) \mid \lambda(v) = \lambda(w)\}$.
For a vertex color $d \in C_V$, we define $V_d \coloneqq \{v \in V(G) \mid \chi(v,v) = d\}$.
We say a set $W \subseteq V(G)$ is \emph{$\chi$-invariant} if there is a set of vertex colors $D \subseteq C_V$ such that $W = \bigcup_{d \in D}V_d$.
For an edge color $c \in C_E(G,\chi)$ and $v \in V(G)$, we define $N_c^+(v) \coloneqq \{w \in V(G) \mid \chi(v,w) = c\} \subseteq N(v)$ and $N_c^-(v) \coloneqq \{w \in V(G) \mid \chi(w,v) = c\} \subseteq N(v)$.
Also, we define $N_c(v) \coloneqq N_c^+(v) \cup N_c^-(v)$.

Next, let $C \subseteq C_E$ be a set of edge colors.
We define the graph $G[C]$ with 
\[ E(G[C]) \coloneqq \{v_1v_2 \mid v_1, v_2 \in V(G), \WL{2}{G}(v_1,v_2) \in C\} \qquad \text{and} \qquad V(G[C]) \coloneqq \bigcup\limits_{e \in E(G[C])} e \ .\]
In case $C = \{c_1,\dots,c_\ell\}$, we also write $G[c_1,\dots,c_\ell]$ instead of $G[\{c_1,\dots,c_\ell\}]$.
Observe that $G[C]$ is defined as an undirected graph.
However, it may be that $\chi(v_1,v_2) \neq \chi(v_2,v_1)$ for some edges $v_1v_2 \in E(G[C])$.
Since this information turns out to be relevant in some cases, we shall always assume that $G[C]$ is equipped with an arc coloring where colors are inherited from $\chi$.
Slightly abusing notation, we will mostly ignore the additional information provided by the arc coloring, but in some cases rely on it to facilitate the analysis of certain graphs~$G[C]$.

As indicated, we are particularly interested in the case $C = \{c\}$ for a single color~$c$.
Observe that the ends of $c$-colored edges have the same vertex color, i.e., if $\chi(v_1,w_1) = \chi(v_2,w_2) = c$, then $\chi(v_1,v_1) = \chi(v_2,v_2)$ and $\chi(w_1,w_1) = \chi(w_2,w_2)$.
This implies that $1 \leq |C_V(G[c],\chi)| \leq 2$.
We say that $G[c]$ is \emph{unicolored} if $|C_V(G[c],\chi)| = 1$.
Otherwise, we say that $G[c]$ is \emph{bicolored}.

Let $A_1,\dots,A_\ell$ be the vertex sets of the connected components of $G[C]$.
We also define the graph $G/C$ as the one obtained from contracting every set $A_i$ to a single vertex.
Formally, $V(G/C) \coloneqq \{\{v\} \mid v \in V(G) \setminus V(G[C])\} \cup \{A_1,\dots,A_\ell\}$ and $E(G/C) \coloneqq \{X_1X_2 \mid \exists v_1 \in X_1,v_2 \in X_2\colon v_1v_2 \in E(G)\}$. If $C$ consists of a single edge color $\{c\}$, we write $G/c$ instead of $G/C$.

\begin{lemma}[see {\cite[Theorem 3.1.11]{ChenP19}}]
 \label{la:factor-graph-2-wl}
 Let $G$ be a graph and $C$ be a set of edge colors of $\chi \coloneqq \WL{2}{G}$.
 Define $(\chi/C)(X_1,X_2) \coloneqq \{\!\{\chi(v_1,v_2) \mid v_1 \in X_1, v_2 \in X_2\}\!\}$ for all $X_1,X_2 \in V(G/C)$.
 Then $\chi/C$ is a $2$-stable coloring of the graph $G/C$ with respect to $2$-WL.
 
 Moreover, for all $X_1,X_2,X_1',X_2' \in V(G/C)$, it holds that either $(\chi/C)(X_1,X_2) = (\chi/C)(X_1',X_2')$ or $(\chi/C)(X_1,X_2) \cap (\chi/C)(X_1',X_2') = \emptyset$.
\end{lemma}

Next, let $D \subseteq C_V(G,\chi)$ be a set of vertex colors.
A path $v_0,\dots,v_m$ \emph{avoids $D$} if $\chi(v_i,v_i) \notin D$ for all $i \in [m-1]$.
Observe that the endpoints of the path may have a vertex color contained in the set $D$.

\begin{observation}
 \label{obs:wl-knows-paths-avoiding-colors}
 Let $G$ be a graph and let $\chi$ be $2$-stable with respect to $G$.
 Also let $v_1,w_1,v_2,w_2 \in V(G)$ such that $\chi(v_1,w_1) = \chi(v_2,w_2)$, and let $D \subseteq C_V(G,\chi)$.
 Moreover, suppose there is a path from $v_1$ to $w_1$ that avoids $D$.
 Then there is a path from $v_2$ to $w_2$ that avoids $D$.
\end{observation}

We say a set of edge colors $C \subseteq C_E(G,\chi)$ is \emph{locally determined} if $C \subseteq C_E(G[A],\chi)$ for every vertex set $A$ of a connected component of $G[C]$, i.e., all edge colors from $C$ appear in every connected component of $G[C]$.
Observe that all singleton sets are locally determined, i.e., if $|C| = 1$, then $C$ is locally determined.

\begin{lemma}
 \label{la:path-between-color-components}
 Let $G$ be a connected graph and let $\chi$ be $2$-stable with respect to $G$.
 Also, let $C \subseteq C_E(G,\chi)$ be locally determined and suppose $A_1$, $A_2$, $A_3$ are the vertex sets of distinct connected components of $G[C]$.
 Then there is a path $u_0,\dots,u_m$ in $G$ such that $u_0 \in A_1$, $u_m \in A_3$, and $u_i \notin A_2$ for all $i \in [0,m]$.
\end{lemma}

\begin{proof}
 Let $D \coloneqq \{\chi(v,v),\chi(w,w) \mid v,w \in V(G), \chi(v,w) \in C\}$ be the set of vertex colors that appear on endpoints of $C$-colored edges.
 Also, suppose that $A_1,\dots,A_\ell$ are the vertex sets of all connected components of $G[C]$.
 Consider the graph $H$ with vertex set $V(H) \coloneqq \{A_1,\dots,A_\ell\}$ and edge set
 $E(H) \coloneqq \{A_iA_j \mid \exists v \in A_i, w \in A_j\colon \text{there is a path from $v$ to $w$ that avoids } D\}$.
 Using Lemma \ref{la:factor-graph-2-wl} and the fact that $C$ is locally determined as well as Observation \ref{obs:wl-knows-paths-avoiding-colors}, we get that $\WL{2}{H}(A_i,A_i) = \WL{2}{H}(A_j,A_j)$ for all $i,j \in [\ell]$.
 Moreover, $H$ is connected because $G$ is connected.
 Together with \cite[Theorem 3.15]{KieferN22}, this means that $H$ is $2$-connected.
 So $A_3$ is reachable from $A_1$ in the graph $H - \{A_2\}$.
 This implies the statement of the lemma since every path in $H$ that does not visit $A_2$ can be lifted to a walk in $G$ that does not visit a vertex from $A_2$.
 Omitting all cycles from the walk, we obtain the desired path.
\end{proof}

\subsection{Edge Types}
\label{sec:types-defs}

To analyze $2$-WL on $3$-connected planar graphs, we consider the graphs $G[c]$ for suitable edge colors $c \in C_E$.
Towards this end, it turns out to be useful to group the graphs $G[c]$ according to the number of faces of each connected component of $G[c]$.
Note that $2$-WL detects connected components of graphs.
More precisely, it holds that
\begin{itemize}
 \item all connected components in $G[c]$ have the same size because 2-WL detects, for every $w \in V(G[c])$, the set of vertices reachable in the arc-colored graph $(G,\chi)$ from $w$ via edges $uv$ of color $c$ (i.e., $\chi(u,v) = c$ or $\chi(v,u) = c$), and
 \item all connected components in $G[c]$ have the same multiset of vertex degrees in $G[c]$ since otherwise the vertex colors and, thus, also the arc colors would be different.
\end{itemize}
In combination, all connected components of $G[c]$ have the same number of vertices and edges and hence, by Euler's formula, they also have the same number of faces. 
We distinguish between three types in $C_E$.

\subparagraph{Type I.}
For the first category, we consider those graphs $G[c]$ that have only one face.
To be more precise, we say that $c \in C_E$ has \emph{Type I} if $(G[c])[A]$ has a single face for every vertex set $A$ of a connected component of $G[c]$. It is not difficult to see that $G[c]$ is isomorphic to a disjoint union of stars $K_{1,h}$ for some $h \in [n]$.

\subparagraph{Type II.}
For the second category, we consider those graphs $G[c]$ where every connected component has exactly two faces.
Formally, we say that $c \in C_E$ has \emph{Type II} if $(G[c])[A]$ has exactly two faces for every vertex set $A$ of a connected component of $G[c]$.
In this case, $G[c]$ is a disjoint union of cycles of the same length.
Also, it is not difficult to see that every connected component of $G[c]$ is either a directed cycle (i.e., $\chi(v_1,v_2) \neq \chi(v_2,v_1)$ holds for every edge $v_1v_2 \in E(G[c])$), or an undirected cycle in which all vertices have the same color with respect to $2$-WL, or an undirected cycle with two vertex colors that alternate along the cycle.

\subparagraph{Type III.}
Finally, for the last category, we consider those graphs $G[c]$ where each connected component has at least three faces.
Again, to be precise, we say that $c \in C_E$ has \emph{Type III} if $(G[c])[A]$ has at least three faces for every vertex set $A$ of a connected component of $G[c]$.

\medskip

Also, we define the \emph{type} of an edge $v_1v_2 \in E(G)$ as the type of its color $\chi(v_1,v_2)$ (note that the type of $\chi(v_1,v_2)$ is equal to the type of $\chi(v_2,v_1)$).

In the following, we derive several properties of the graphs $G[c]$ depending on the type of~$c$, as well as properties of $G$ depending on which types of edge colors occur.
Towards this end, we also define the \emph{type of $G$} as the maximal type of any edge color $c \in C_E$.
So we say that $G$ has Type III if there is some $c \in C_E$ of Type III.
The graph $G$ has Type II if there is some $c \in C_E$ of Type II but there is no $c' \in C_E$ of Type III.
Lastly, $G$ has Type I if every $c \in C_E$ has Type I.
Two example graphs are displayed in Figure \ref{fig:type-examples}.

\subsection{Graphs of Fixing Number One}

We start by investigating $3$-connected planar graphs of Type I, see Figure \ref{fig:type-i-example} for an example.
It turns out that such graphs have fixing number $1$ with respect to $1$-WL (after coloring all edges with their $2$-WL colors), which in particular implies that $2$-WL identifies all graphs of Type I.
The proof is based on the following result.

\begin{figure}
 \centering
 \begin{subfigure}[b]{.48\linewidth}
  \centering
  \scalebox{\figscalesmall}{
  \begin{tikzpicture}
   \node[vertex,red!80] (1) at (0,0) {};
   
   \node[vertex,blue!80] (2) at ($(0,0)+(00:1.0)$) {};
   \node[vertex,blue!80] (3) at ($(0,0)+(60:1.0)$) {};
   \node[vertex,blue!80] (4) at ($(0,0)+(120:1.0)$) {};
   \node[vertex,blue!80] (5) at ($(0,0)+(180:1.0)$) {};
   \node[vertex,blue!80] (6) at ($(0,0)+(240:1.0)$) {};
   \node[vertex,blue!80] (7) at ($(0,0)+(300:1.0)$) {};
   
   \node[vertex,blue!80] (8) at ($(0,0)+(00:1.0)+(330:0.8)$) {};
   \node[vertex,blue!80] (9) at ($(0,0)+(60:1.0)+(90:0.8)$) {};
   \node[vertex,blue!80] (10) at ($(0,0)+(120:1.0)+(90:0.8)$) {};
   \node[vertex,blue!80] (11) at ($(0,0)+(180:1.0)+(210:0.8)$) {};
   \node[vertex,blue!80] (12) at ($(0,0)+(240:1.0)+(210:0.8)$) {};
   \node[vertex,blue!80] (13) at ($(0,0)+(300:1.0)+(330:0.8)$) {};
   
   \node[vertex,red!80] (14) at ($(0,0)+(90:3.0)$) {};
   \node[vertex,blue!80] (15) at ($(0,0)+(130:3.0)$) {};
   \node[vertex,blue!80] (16) at ($(0,0)+(170:3.0)$) {};
   \node[vertex,red!80] (17) at ($(0,0)+(210:3.0)$) {};
   \node[vertex,blue!80] (18) at ($(0,0)+(250:3.0)$) {};
   \node[vertex,blue!80] (19) at ($(0,0)+(290:3.0)$) {};
   \node[vertex,red!80] (20) at ($(0,0)+(330:3.0)$) {};
   \node[vertex,blue!80] (21) at ($(0,0)+(10:3.0)$) {};
   \node[vertex,blue!80] (22) at ($(0,0)+(50:3.0)$) {};
   
   \node[vertex,blue!80] (23) at ($(0,0)+(00:4.0)$) {};
   \node[vertex,blue!80] (24) at ($(0,0)+(60:4.0)$) {};
   \node[vertex,blue!80] (25) at ($(0,0)+(120:4.0)$) {};
   \node[vertex,blue!80] (26) at ($(0,0)+(180:4.0)$) {};
   \node[vertex,blue!80] (27) at ($(0,0)+(240:4.0)$) {};
   \node[vertex,blue!80] (28) at ($(0,0)+(300:4.0)$) {};
   
   \foreach \i/\j in {1/2,1/3,1/4,1/5,1/6,1/7,14/9,14/10,14/15,14/22,14/24,14/25,17/11,17/12,17/16,17/18,17/26,17/27,20/8,20/13,20/19,20/21,20/23,20/28}{
    \draw[line width=1.6pt] (\i) edge (\j);
   }
   \foreach \i/\j in {2/8,3/9,4/10,5/11,6/12,7/13,15/16,18/19,21/22,23/24,25/26,27/28}{
    \draw[line width=1.6pt,darkpastelgreen] (\i) edge (\j);
   }
   \foreach \i/\j in {2/3,4/5,6/7,9/22,10/15,11/16,12/18,8/21,13/19,23/28,24/25,26/27}{
    \draw[line width=1.6pt,lipicsYellow] (\i) edge (\j);
   }
   \foreach \i/\j in {2/7,3/4,5/6,8/13,9/10,11/12,15/25,16/26,18/27,19/28,21/23,22/24}{
    \draw[line width=1.6pt,magenta] (\i) edge (\j);
   }
   
  \end{tikzpicture}
  }
  \caption{A graph $G$ of Type I. Each edge color $c \in C_E$ defines a graph $G[c]$ that is isomorphic to a disjoint union of stars.
   Individualizing an arbitrary {\sf blue} vertex and performing $1$-WL results in a discrete coloring.
   Hence, the graph is identified by $2$-WL.}
  \label{fig:type-i-example}
 \end{subfigure}
 \hfill
 \begin{subfigure}[b]{.48\linewidth}
  \centering
  \scalebox{\figscalesmall}{
  \begin{tikzpicture}
   
   \node[vertex,red!80] (1) at (0,0) {};
   
   \node[vertex,blue!80] (2) at ($(0,0)+(00:1.0)$) {};
   \node[vertex,blue!80] (3) at ($(0,0)+(60:1.0)$) {};
   \node[vertex,blue!80] (4) at ($(0,0)+(120:1.0)$) {};
   \node[vertex,blue!80] (5) at ($(0,0)+(180:1.0)$) {};
   \node[vertex,blue!80] (6) at ($(0,0)+(240:1.0)$) {};
   \node[vertex,blue!80] (7) at ($(0,0)+(300:1.0)$) {};
   
   \node[vertex,springgreen] (8) at ($(0,0)+(30:1.2)$) {};
   \node[vertex,springgreen] (9) at ($(0,0)+(150:1.2)$) {};
   \node[vertex,springgreen] (10) at ($(0,0)+(270:1.2)$) {};
   
   \node[vertex,blue!80] (11) at ($(0,0)+(30:1.8)$) {};
   \node[vertex,blue!80] (12) at ($(0,0)+(150:1.8)$) {};
   \node[vertex,blue!80] (13) at ($(0,0)+(270:1.8)$) {};
   
   \node[vertex,red!80] (14) at ($(0,0)+(90:3.0)$) {};
   \node[vertex,red!80] (15) at ($(0,0)+(210:3.0)$) {};
   \node[vertex,red!80] (16) at ($(0,0)+(330:3.0)$) {};
   
   \node[vertex,blue!80] (17) at ($(0,0)+(30:4.0)$) {};
   \node[vertex,blue!80] (18) at ($(0,0)+(150:4.0)$) {};
   \node[vertex,blue!80] (19) at ($(0,0)+(270:4.0)$) {};
   
   \node[vertex,springgreen] (20) at ($(0,0)+(270:4.8)$) {};
   
   \foreach \i/\j in {1/14,1/15,1/16}{
    \draw[line width=1.6pt] (\i) edge (\j);
   }
   \draw[line width=1.6pt, bend right] (14) edge (15);
   \draw[line width=1.6pt, bend left] (14) edge (16);
   \draw[line width=1.6pt, bend right] (15) edge (16);
   
   \foreach \i/\j in {1/2,1/3,1/4,1/5,1/6,1/7,2/16,3/14,4/14,5/15,6/15,7/16,11/14,11/16,12/14,12/15,13/15,13/16,17/14,17/16,18/14,18/15,19/15,19/16}{
    \draw[line width=1.6pt, darkpastelgreen] (\i) edge (\j);
   }
   
   \foreach \i/\j in {2/3,4/5,6/7,2/11,3/11,4/12,5/12,6/13,7/13}{
    \draw[line width=1.6pt, lipicsYellow] (\i) edge (\j);
   }
   \draw[line width=1.6pt, lipicsYellow, bend right = 40] (17) edge (18);
   \draw[line width=1.6pt, lipicsYellow, bend left = 40] (17) edge (19);
   \draw[line width=1.6pt, lipicsYellow, bend right = 40] (18) edge (19);
   
   \foreach \i/\j in {2/8,3/8,4/9,5/9,6/10,7/10,8/11,9/12,10/13,19/20}{
    \draw[line width=1.6pt, magenta] (\i) edge (\j);
   }
   \draw[line width=1.6pt, magenta, bend left = 40] (17) edge (20);
   \draw[line width=1.6pt, magenta, bend right = 40] (18) edge (20);
  \end{tikzpicture}
  }
  \caption{A graph $G$ of Type III. The edge colors {\sf black} and {\sf green} have Type III, {\sf yellow} has Type II, and {\sf pink} has Type I.
   Note that $G[c]$ is connected for every edge color $c$ of Type III whereas the other edge colors induce non-connected subgraphs.}
  \label{fig:type-iii-example}
 \end{subfigure}
 \caption{Two $3$-connected planar graphs where all vertices and edges are colored by their $2$-WL color. For visualization purposes, we only color edges and do not distinguish between potentially different colors of two arcs $(v,w)$ and $(w,v)$.}
 \label{fig:type-examples}
\end{figure}
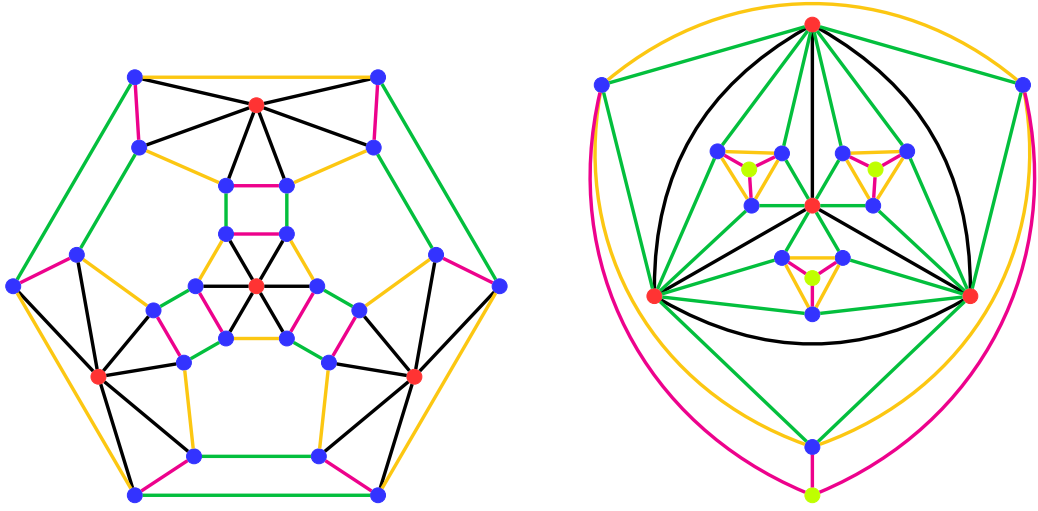

\begin{theorem}[{\cite[Lemma 23]{KieferPS19}}]
 \label{thm:splitting-with-tutte}
 Let $G$ be a $3$-connected planar graph and suppose $v_1,v_2,v_3 \in V(G)$ are pairwise distinct vertices lying on a common face of $G$.
 Then $\WL{1}{G,v_1,v_2,v_3}$ is discrete.
\end{theorem}

Here, $\WL{1}{G,v_1,v_2,v_3}$ denotes the coloring computed by $1$-WL after individualizing $v_1$, $v_2$, and $v_3$.
Recall that for a vertex coloring $\lambda\colon V \rightarrow C$ and $v \in V$, we write $[v]_\lambda \coloneqq \{w \in V \mid \lambda(v) = \lambda(w)\}$ for the color class of $v$.
We also define $\Fix(\lambda) \coloneqq \{v \in V \mid |[v]_\lambda| = 1\}$.
For a graph $G$ and vertices $v_1,\dots,v_\ell \in V(G)$, we define
\[\Disc_G(v_1,\dots,v_\ell) \coloneqq \Fix(\WL{1}{G,\WL{2}{G},v_1,\dots,v_\ell}).\]
In other words, $\Disc_G(v_1,\dots,v_\ell)$ is the set of all vertices appearing in a singleton color class after performing $1$-WL on $G$ where every pair is colored with its $2$-WL-color, and where $v_1,\dots,v_\ell$ are individualized.

\begin{lemma}
 \label{la:wl-pair-orbits-from-fixing-number}
 Let $G$ be a graph and let $v_1,\dots,v_\ell \in V(G)$ such that $\Disc_G(v_1,\dots,v_\ell) = V(G)$.
 Also define $k \coloneqq \max\{2,\ell+1\}$.
 Then $k$-WL determines pair orbits in $G$.
\end{lemma}

\begin{proof}
 Let $H$ be a second graph such that $|V(G)| = |V(H)|$ (otherwise $k$-WL clearly distinguishes between $G$ and $H$).
 Also, let $u_1,u_2 \in V(G)$ and $w_1,w_2 \in V(H)$ such that
 \[\WL{k}{G}(u_1,u_2,\dots,u_2) = \WL{k}{H}(w_1,w_2,\dots,w_2).\]
 This means that Duplicator has a winning strategy in the game $\BP^{k+1}(G,H)$ from the initial position $((u_1,u_2),(w_1,w_2))$.
 Now, suppose that in the next $\ell$ rounds Spoiler places pebbles on $v_1,\dots,v_\ell$.
 This means, the game moves to a position
 \[\big((u_1,u_2,v_1,\dots,v_\ell),(w_1,w_2,v_1',\dots,v_\ell')\big).\]
 Observe that $k+1 \geq \ell + 2$, so Spoiler can indeed have pebbles placed on all these vertices at the same time.
 
 Now, $\lambda_G \coloneqq \WL{1}{G,\WL{2}{G},v_1,\dots,v_\ell}$ is discrete.
 Let $\lambda_H \coloneqq \WL{1}{H,\WL{2}{H},v_1',\dots,v_\ell'}$.
 First suppose that for every $u \in V(G)$, there is some $w \in V(H)$ such that $\lambda_G(u) = \lambda_H(w)$.
 Also suppose that $\lambda_G(u_j) = \lambda_H(w_j)$ for both $j \in \{1,2\}$.
 Since $|V(G)| = |V(H)|$, this implies that $\lambda_H$ is also discrete.
 Let $\varphi\colon V(G) \rightarrow V(H)$ be the bijection defined via $\varphi(u) = w$ for the unique $w \in V(H)$ such that $\lambda_G(u) = \lambda_H(w)$.
 It is easy to see that $\varphi$ constitutes an isomorphism between $G$ and $H$, and $\varphi(u_j) = w_j$ for both $j \in \{1,2\}$.
 
 In the other case, we give a winning strategy for Spoiler from the current position, which leads to a contradiction.
 To win the game, Spoiler's first goal is to move to a position $\big((u,v_1,\dots,v_\ell),(w,v_1',\dots,v_\ell')\big)$ such that $\lambda_G(u) \neq \lambda_H(w)$.
 If $\lambda_G(u_j) \neq \lambda_H(w_j)$ for some $j \in \{1,2\}$, then Spoiler can simply remove the other pair $(u_{3-j},w_{3-j})$.
 In the other case, there is some $u \in V(G)$ such that $\lambda_G(u) \notin \im(\lambda_H)$ (here, $\im(\lambda_H) \coloneqq \{\lambda_H(v) \mid v \in V(H)\}$ denotes the image of $\lambda_H$).
 Spoiler first removes the pebbles from $(u_1,w_1)$ and $(u_2,w_2)$.
 Then he plays a new pair of pebbles $(u,w)$ such that $\lambda_G(u) \notin \im(\lambda_H)$.
 
 Now, by Theorem \ref{thm:eq-wl-pebble-tuples}, Spoiler can reach a position $\big((u_1',u_2',v_1,\dots,v_\ell),(w_1',w_2',v_1',\dots,v_\ell')\big)$ such that one of the following holds.
 \begin{enumerate}
  \item $\WL{2}{G}(u_1',u_2') \neq \WL{2}{G}(w_1',w_2')$,
  \item $u_1' = u_2' \not\Leftrightarrow w_1' = w_2'$, or
  \item $u_j' = v_i \not\Leftrightarrow w_j' = v_i'$ for some $j \in \{1,2\}$ and $i \in [\ell]$.
 \end{enumerate}
 In the latter two cases, Spoiler wins immediately.
 Otherwise, Spoiler eventually wins by first removing all pebbles from $v_1,\dots,v_\ell$ and $v_1',\dots,v_\ell'$, and then using $3 \geq k+1$ pebble pairs in total to simulate $2$-WL building on Theorem \ref{thm:eq-wl-pebble-tuples}.
\end{proof}

The following lemma provides a sufficient condition for a $3$-connected planar graph to have fixing number $1$.

\begin{lemma}
 \label{la:face-cycle-with-star-edges}
 Let $G$ be a $3$-connected planar graph and suppose there is a face $F$ such that every edge $e \in E(F)$ has Type I.
 Then there is a vertex $v \in V(G)$ such that $\Disc_G(v) = V(G)$.
\end{lemma}

\begin{proof}
 Let $H$ be a directed graph with vertex set $V(H) \coloneqq V(G)$ and edge set
 \[E(H) \coloneqq \left\{(v,w) \, \Big\vert \, vw \in E(G) \wedge \deg_{G[\WL{2}{G}(v,w)]}(v) = 1\right\}.\]
 Intuitively speaking, we add a directed edge $(v,w)$ to the graph $H$ if $w$ is the only neighbor of $v$ reachable via an edge of color $\WL{2}{G}(v,w)$.
 In particular, if $v$ is individualized, then $w$ is also fixed after performing $1$-WL.
 
 For every edge $vw \in E(G)$ of Type I, it holds that $(v,w) \in E(H)$ or $(w,v) \in E(H)$.
 Hence, there are three vertices $v_1,v_2,v_3 \in V(G)$ lying on the face $F$ of $G$ such that $(v_1,v_2),(v_2,v_3) \in E(H)$ or $(v_1,v_2),(v_1,v_3) \in E(H)$.
 
 Now consider the coloring $\lambda \coloneqq \WL{1}{G,\WL{2}{G},v_1}$.
 Let $c \coloneqq \WL{2}{G}(v_1,v_2)$.
 By definition of the edge set of $H$, it holds that $v_2$ is the only neighbor of $v_1$ which is adjacent via an edge of color $c$.
 Hence, $|[v_2]_\lambda| = 1$.
 By the same argument, $|[v_3]_\lambda| = 1$.
 Since $v_1,v_2,v_3$ all lie on the face $F$, it follows from Theorem \ref{thm:splitting-with-tutte} that $\lambda$ is discrete.
 In other words, $\Disc_G(v_1) = V(G)$.
\end{proof}

\begin{corollary}
 \label{cor:wl-pair-orbits-type-i}
 Let $G$ be a $3$-connected planar graph of Type I.
 Then there is a vertex $v \in V(G)$ such that $\Disc_G(v) = V(G)$.
 In particular, $2$-WL determines pair orbits of $G$.
\end{corollary}

We also record the following useful lemma, which is another consequence of Theorem \ref{thm:splitting-with-tutte}.

\begin{lemma}
 \label{la:fixing-cycle-and-extra-vertex}
 Let $G$ be a $3$-connected planar graph.
 Suppose $v_1,\dots,v_\ell \in V(G)$ form a cycle in $G$, i.e., $v_1,\dots,v_\ell$ are pairwise distinct, $v_1v_\ell \in E(G)$ and $v_iv_{i+1} \in E(G)$ holds for all $i \in [\ell-1]$.
 Let $w \in V(G) \setminus \{v_1,\dots,v_\ell\}$.
 Let $\lambda$ be a vertex coloring of $G$ such that $\{v_1,\dots,v_\ell,w\} \subseteq \Disc(\lambda)$ and $\lambda$ is $1$-stable with respect to $G$.
 Then $\lambda$ is discrete, i.e., $\Disc(\lambda) = V(G)$.
\end{lemma}

The lemma says that in a $3$-connected planar graph $G$, it suffices to fix a cycle and one additional vertex in order to fix the entire graph.
For example, this allows us to extract from the presence of certain $2$-WL-detectable subgraphs bounds on the fixing number of the entire $3$-connected planar graph $G$.
Note that in the case that the fixing number in the subgraph is $1$, Lemma \ref{la:wl-pair-orbits-from-fixing-number} yields that $2$-WL determines pair orbits in $G$.

\begin{proof}[Proof of Lemma \ref{la:fixing-cycle-and-extra-vertex}]
 Let $A$ denote the vertex set of the connected component of $G - \{v_1,\dots,v_\ell\}$ that contains $w$.
 Let $H$ be the graph defined via $V(H) \coloneqq A \cup \{v_1,\dots,v_\ell\} \cup \{z\}$ (where $z$ denotes some fresh vertex) and
 \[E(H) \coloneqq E_G(A, A \cup \{v_1,\dots,v_\ell\}) \cup \{v_iv_{i+1} \mid i \in [\ell-1]\} \cup \{v_1v_\ell\} \cup \{zv_i \mid i \in [\ell]\}.\]
 Also, let $\lambda_H \coloneqq V(H) \rightarrow \{1,2\}$ be the vertex coloring defined via $\lambda_H(z) = 1$ and $\lambda_H(v) = 2$ for all $v \in V(H) \setminus \{z\}$.
 It is easy to see that $H$ is planar and $3$-connected.
 Thus, $\WL{1}{H,\lambda_H,v_1,\dots,v_\ell}$ is discrete by Theorem \ref{thm:splitting-with-tutte}.
 
 Now, consider the coloring $\lambda$.
 Since $w$ is individualized, it holds that $A$ is $\lambda$-invariant.
 But this means that
 \[\lambda|_A \preceq \big(\WL{1}{H,\lambda_H,v_1,\dots,v_\ell}\big)|_A,\]
 since all refinement steps that are performed in $H$ can also be done in $G$.
 In particular, $A \subseteq \Disc(\lambda)$.
 But this means $\Disc(\lambda)$ contains three distinct vertices $u_1,u_2,u_3$ that lie on a common face of $G$.
 So, $\Disc(\lambda) = V(G)$ by Theorem \ref{thm:splitting-with-tutte}.
\end{proof}

\subsection{Three Faces}
We now turn to edge colors of Types II and III.
For both types, it is not difficult to see that it is in general impossible to bound the fixing number by $1$ (see, e.g., the graphs displayed in Figures \ref{fig:type-iii-example} and \ref{fig:type-ii-examples}).
Instead, our focus here lies on investigating how edge colors of the corresponding type can appear within a $3$-connected planar graph.

We first focus on edge colors of Type III. Ssee Figure \ref{fig:type-iii-example} for an example.
Let $G$ be a $3$-connected planar graph and let $c \in C_E(G,\chi)$ be an edge color of Type III, where $\chi \coloneqq \WL{2}{G}$.
By Corollary \ref{cor:edge-orbits-for-edge-transitive}, the graph $G[c]$ is edge-transitive, which already puts severe restrictions on $G[c]$.
However, as it turns out, due to the planarity and $3$-connectedness of $G$, many edge-transitive graphs can in fact not appear as subgraphs $G[c]$.
In the following, we classify the graphs $G[c]$ induced by an edge color $c$ of Type III.
The following lemma is a useful tool for the proof of our classification in Theorem \ref{thm:type-iii-classification}, but also an interesting insight by itself, since it also yields restrictions on how different colors $c,c'$ of Type III can appear together in one graph $G$.

\begin{lemma}
 \label{la:3-faces-connected}
 Let $G$ be a $3$-connected planar graph and let $c \in C_E(G,\WL{2}{G})$ be an edge color of Type III.
 Then $G[c]$ is connected.
 Moreover, for every edge color $c' \in C_E(G,\WL{2}{G})$ of Type III, it holds that $V(G[c]) \cap V(G[c']) \neq \emptyset$.
\end{lemma}

\begin{proof}
 For the first part of the lemma, suppose that $G[c]$ is not connected and let $A_1,\dots,A_\ell$ be the vertex sets of the connected components of $G[c]$.
 Also let $D \coloneqq \{\chi(v,v),\chi(w,w) \mid v,w \in V(G), \chi(v,w) = c\}$.
 Observe that $|D| \in \{1,2\}$.
 By Lemma \ref{la:path-between-color-components}, there is a vertex set $X$ of a connected component of the graph $G - A_1$ such that $A_i \subseteq X$ for all $i \in \{2,\dots,\ell\}$.
 
 \begin{claim}
  \label{claim:neighborhood-color-closed}
  Let $d \in D$ and suppose that $N_G(X) \cap V_d \neq \emptyset$.
  Then $V_d \cap A_1 \subseteq N_G(X)$.
 \end{claim}
 \begin{claimproof}
  Consider the set
  \begin{align*}
   M \coloneqq \Big\{(v,w) \in (V(G))^2 \;\Big|\; &v \in A_i \text{ and } w \in A_j \text{ for some } i \neq j \in [\ell] \text{ and }\\
                                                  &\text{there is a path from } v \text{ to } w \text{ that avoids } D\Big\}
  \end{align*}
  This set is detected by $2$-WL, i.e., there is a set of colors $C_M$ such that $(v,w) \in M$ if and only if $\WL{2}{G}(v,w) \in C_M$ for all $v,w \in V(G)$ (cf.\ Observation \ref{obs:wl-knows-paths-avoiding-colors}).
  Now let $u \in N_G(X) \cap V_d$ and let $v \in X$ such that $uv \in E(G)$.
  Since $G[X]$ is connected and $A_2,\dots,A_\ell \subseteq X$, there is some $w \in A_2 \cup \dots \cup A_\ell$ such that there is a path from $v$ to $w$ that avoids $D$.
  Extending this path with the edge $uv$ gives a path from $u$ to $w$ that avoids $D$.
  In particular, $(u,w) \in M$ and $\WL{2}{G}(u,w) \in C_M$.
  
  Now let $u' \in V_d \cap A_1$ be a second vertex.
  Since $\WL{2}{G}(u,u) = \WL{2}{G}(u',u')$ there is some vertex $w' \in V(G)$ such that $\WL{2}{G}(u',w') \in C_M$.
  So $(u',w') \in M$.
  Let $j \in \{2,\dots,\ell\}$ be such that $w' \in A_j$.
  Note that $w' \in X$.
  Since there is a path from $u'$ to $w'$ that avoids $D$, it follows that $u' \in N_G(X)$.
 \end{claimproof}
 
 On the other hand, by planarity, there is some face $F$ of $G[A_1]$ such that all vertices contained in $N_G(X)$ must be incident to $F$ (the connected subgraph $G[X]$ is contained in the face $F$).
 
 Let $v \in N_G(X)$ and let $d \coloneqq \WL{2}{G}(v,v)$.
 By the last claim and the comment above, there is some face $F$ of $G[A_1]$ such that all vertices from $V_d \cap A_1$ are incident to $F$.
 
 Now, let us analyze in which cases this is possible based on Theorem \ref{thm:classification-edge-transitive} and the comments below the theorem.
 First, $G[A_1]$ has minimum degree at least $2$ and maximum degree at least $3$ since $c$ has Type III.
 If $G$ has minimum degree at least $3$, then $G[A_1]$ is isomorphic to one of the graphs from Figure \ref{fig:transitive} by Theorem \ref{thm:classification-edge-transitive}.
 However, for none of these graphs, there is a color $d \in D$ and a face $F$ of $G[A_1]$ such that all vertices from $V_d \cap A_1$ are incident to $F$.
 So this case cannot occur.
 By the same argument, $G[A_1]$ cannot be an $s$-subdivision of one of the graphs from Figure \ref{fig:transitive}.
 
 Next, suppose that $G[A_1]$ is an $s$-subdivision of the graph $K_2$, i.e., $G[A_1]$ is isomorphic to $K_{2,s}$ for some $s \geq 3$ (observe that $s \leq 2$ is not possible since $c$ as Type III).
 In this case, $|D| = 2$.
 Suppose $D = \{d_1,d_2\}$ where $V_{d_1}$ contains all degree-$2$-vertices of $G[c]$ and $V_{d_2}$ contains the remaining vertices of $G[c]$.
 If $N_G(X) \cap V_{d_1} \neq \emptyset$, then $V_{d_1} \cap A_1 \subseteq N_G(X)$ by Claim \ref{claim:neighborhood-color-closed}.
 But this contradicts $G$ being planar since contracting $X$ results in a subgraph isomorphic to $K_{3,s}$ (recall that $s \geq 3$).
 So $N_G(X) \subseteq V_{d_2} \cap A_1$.
 But $|V_{d_2} \cap A_1| = 2$, which contradicts the $3$-connectedness of $G$.
 
 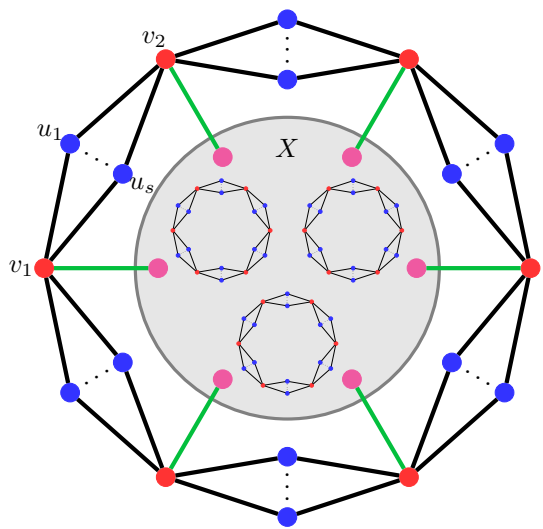
\begin{figure}
  \centering
  \begin{tikzpicture}
   \draw[line width=1.2pt, gray, fill = gray!20] (0,0) circle (2.0cm);
   
   \foreach \i in {0,...,5}{
    \node[vertex,red!80] (u\i) at ($(0,0)+(60*\i:3.2)$) {};
    \node[vertex,blue!80] (v\i-1) at ($(0,0)+(30+60*\i:2.5)$) {};
    \node[vertex,blue!80] (v\i-2) at ($(0,0)+(30+60*\i:3.3)$) {};
    \node[vertex,magenta!80] (w\i) at ($(0,0)+(60*\i:1.7)$) {};
    
    \node[rotate = 30+60*\i] at ($(0,0)+(30+60*\i:2.925)$) {$\dots$};
   }
   \foreach \i/\j in {0/1,1/2,2/3,3/4,4/5,5/0}{
    \draw[line width=1.6pt] (u\i) edge (v\i-1);
    \draw[line width=1.6pt] (u\i) edge (v\i-2);
    \draw[line width=1.6pt] (u\j) edge (v\i-1);
    \draw[line width=1.6pt] (u\j) edge (v\i-2);
    \draw[line width=1.6pt,darkpastelgreen] (u\i) edge (w\i);
   }
   
   \node at (0,1.6) {$X$};
   \node at (180:3.5) {$v_1$};
   \node at (120:3.5) {$v_2$};
   \node at (150:3.6) {$u_1$};
   \node at (150:2.2) {$u_s$};
   
   \node at (30:1){
    \scalebox{0.2}{
    \begin{tikzpicture}
     \foreach \i in {0,...,5}{
      \node[vertex,red!80] (u\i) at ($(0,0)+(60*\i:3.2)$) {};
      \node[vertex,blue!80] (v\i-1) at ($(0,0)+(30+60*\i:2.5)$) {};
      \node[vertex,blue!80] (v\i-2) at ($(0,0)+(30+60*\i:3.3)$) {};
      
      \node[rotate = 30+60*\i] at ($(0,0)+(30+60*\i:2.925)$) {$\dots$};
     }
     \foreach \i/\j in {0/1,1/2,2/3,3/4,4/5,5/0}{
      \draw[line width=1.6pt] (u\i) edge (v\i-1);
      \draw[line width=1.6pt] (u\i) edge (v\i-2);
      \draw[line width=1.6pt] (u\j) edge (v\i-1);
      \draw[line width=1.6pt] (u\j) edge (v\i-2);
     }
    \end{tikzpicture}
    }
   };
   
   \node at (150:1){
    \scalebox{0.2}{
    \begin{tikzpicture}
     \foreach \i in {0,...,5}{
      \node[vertex,red!80] (u\i) at ($(0,0)+(60*\i:3.2)$) {};
      \node[vertex,blue!80] (v\i-1) at ($(0,0)+(30+60*\i:2.5)$) {};
      \node[vertex,blue!80] (v\i-2) at ($(0,0)+(30+60*\i:3.3)$) {};
      
      \node[rotate = 30+60*\i] at ($(0,0)+(30+60*\i:2.925)$) {$\dots$};
     }
     \foreach \i/\j in {0/1,1/2,2/3,3/4,4/5,5/0}{
      \draw[line width=1.6pt] (u\i) edge (v\i-1);
      \draw[line width=1.6pt] (u\i) edge (v\i-2);
      \draw[line width=1.6pt] (u\j) edge (v\i-1);
      \draw[line width=1.6pt] (u\j) edge (v\i-2);
     }
    \end{tikzpicture}
    }
   };
   
   \node at (270:1){
    \scalebox{0.2}{
    \begin{tikzpicture}
     \foreach \i in {0,...,5}{
      \node[vertex,red!80] (u\i) at ($(0,0)+(60*\i:3.2)$) {};
      \node[vertex,blue!80] (v\i-1) at ($(0,0)+(30+60*\i:2.5)$) {};
      \node[vertex,blue!80] (v\i-2) at ($(0,0)+(30+60*\i:3.3)$) {};
      
      \node[rotate = 30+60*\i] at ($(0,0)+(30+60*\i:2.925)$) {$\dots$};
     }
     \foreach \i/\j in {0/1,1/2,2/3,3/4,4/5,5/0}{
      \draw[line width=1.6pt] (u\i) edge (v\i-1);
      \draw[line width=1.6pt] (u\i) edge (v\i-2);
      \draw[line width=1.6pt] (u\j) edge (v\i-1);
      \draw[line width=1.6pt] (u\j) edge (v\i-2);
     }
    \end{tikzpicture}
    }
   };
  \end{tikzpicture}
  \caption{Visualization for the proof of Lemma \ref{la:3-faces-connected}. The vertices from $V_{d_1}$ are colored {\sf blue} and the vertices from $V_{d_2}$ are colored {\sf red}. Also, edges of the graph $G[c]$ are {\sf black}. Since $G$ is planar, there is no path from $u_s$ to a vertex $x_s \in A_1 \setminus \{u_1,\dots,u_s,v_1,v_2\}$ that avoids $D$.}
  \label{fig:type-iii-connected}
 \end{figure}

 So $G[A_1]$ is isomorphic to an $s$-subdivision of a cycle $C_\ell$ of length $\ell$ for some $\ell \geq 3$ and $s \geq 2$.
 Again suppose $D = \{d_1,d_2\}$, where $V_{d_1}$ contains all degree-$2$-vertices of $G[c]$ and $V_{d_2}$ contains the remaining vertices of $G[c]$ (which have degree $2s$ in $G[c]$).
 We obtain $N_G(X) = V_{d_2} \cap A_1$ by the same arguments as above.
 Now let $u_1,\dots,u_s \in A_1 \cap V_{d_1}$ be a set of $s$ distinct vertices and $v_1,v_2 \in A_1 \cap V_{d_2}$  such that $N_{G[c]}(u_i) = N_{G[c]}(u_j) = \{v_1,v_2\}$ for all $i,j \in [s]$ (see Figure \ref{fig:type-iii-connected}).
 Also, let $v_3 \in A_1 \cap V_{d_2}$ be another vertex that is distinct from $v_1$ and $v_2$ (such a vertex exists since $|A_1 \cap V_{d_2}| = \ell \geq 3$).
 Since $G$ is $3$-connected, there is a path $u_1 = w_0,\dots,w_\ell = v_3$ in $G$ that visits neither $v_1$ nor $v_2$.
 Now, consider the last occurrence $w_i$ of a vertex from $u_1,\dots,u_s$ on the path.
 Also, consider the first occurrence $w_j$, $j > i$, of a vertex from $A_1 \setminus \{u_1,\dots,u_s\}$ on the path.
 Consider the subpath $P = w_i,\dots,w_j$.
 Since $N_G(X) = V_{d_2} \cap A_1$, we conclude that $P$ does not visit any vertex from $X$.
 In particular, $P$ avoids $D$.
 Using Observation \ref{obs:wl-knows-paths-avoiding-colors}, we conclude that such a path exists for all vertices $u_1,\dots,u_s$ (by the same arguments as in Claim \ref{claim:neighborhood-color-closed}), i.e., for every $i \in [s]$, there is a path $P_i$ from $u_i$ to some vertex $x_i \in A_1 \setminus \{u_1,\dots,u_s,v_1,v_2\}$ that avoids $D$.
 By the same argument as above, the path $P_i$ does not visit any vertex from $X$.
 However, by Theorem \ref{thm:wagner}, this contradicts the planarity of $G$, as it allows us to construct a minor $K_{3,s+1}$ with vertices $v_1,v_2,A_1 \setminus \{u_1,\dots,u_s,v_1,v_2\}$ on the left side and $u_1,\dots,u_s,X$ on the right side (the paths $P_i$ are used to connect $u_i$ to $A_1 \setminus \{u_1,\dots,u_s,v_1,v_2\}$ for all $i \in [s]$).
 
 So overall, $G[c]$ is connected, which completes the first part of the proof.
 
 \medskip
 
 Now, let us turn to the second statement.
 Let $c' \in C_E(G,\WL{2}{G})$ be a second edge color of Type III.
 By the first part of the lemma, both $G[c]$ and $G[c']$ are connected.
 Suppose that $V(G[c]) \cap V(G[c']) = \emptyset$.
 Let $X$ denote the connected component of $G - V(G[c])$ that contains $V(G[c'])$.
 Now, we can derive a contradiction by employing exactly the same arguments as in the first case.
\end{proof}

\begin{theorem}
 \label{thm:type-iii-classification}
 Let $G$ be a $3$-connected planar graph and let $c \in C_E(G,\WL{2}{G})$ be of Type III.
 Then one of the following holds.
 \begin{enumerate}
  \item\label{item:type-iii-classification-1} $G[c]$ is bicolored and isomorphic to $K_{2,\ell}$ for some $\ell \geq 3$,
  \item\label{item:type-iii-classification-2} $G[c]$ is bicolored and isomorphic to a $2$-subdivision of a cycle $C_\ell$ for some $\ell \geq 3$,
  \item\label{item:type-iii-classification-3} $G[c]$ is bicolored and isomorphic to a graph from Fig.\ \ref{fig:graph-bicol-cube} -- \ref{fig:graph-rhombic-triacontahedron},
  \item\label{item:type-iii-classification-4} $G[c]$ is unicolored and isomorphic to a graph from Fig.\ \ref{fig:graph-k4} -- \ref{fig:graph-icosidodecahedron},
  \item\label{item:type-iii-classification-5} $G[c]$ is bicolored and isomorphic to a $1$-subdivision of a graph from Fig.\ \ref{fig:graph-k4} -- \ref{fig:graph-icosidodecahedron}, or
  \item\label{item:type-iii-classification-6} $G[c]$ is bicolored and isomorphic to a $2$-subdivision of a graph from Fig.\ \ref{fig:graph-k4} -- \ref{fig:graph-icosahedron}.
 \end{enumerate}
\end{theorem}

\begin{proof}
 By Lemma \ref{la:3-faces-connected}, the graph $G[c]$ is connected.
 Also, since $c$ has Type III, $G[c]$ has minimum degree at least $2$ and maximum degree at least $3$.
 If $G[c]$ has minimum degree at least $3$, then $G[c]$ is isomorphic to one of the graphs from Figure \ref{fig:transitive} by Theorem \ref{thm:classification-edge-transitive}.
 All these graphs are listed in the theorem (see Parts \ref{item:type-iii-classification-3} and \ref{item:type-iii-classification-4}).
 
 So suppose $G[c]$ has minimum degree $2$.
 Since $G[c]$ has maximum degree at least three, $G[c]$ is bicolored.
 By the comments below Theorem \ref{thm:classification-edge-transitive}, $G[c]$ is isomorphic to one of the following graphs:
 \begin{enumerate}[label=(\alph*)]
  \item\label{item:type-iii-bicolored-1} an $s$-subdivision of one of the graphs from Figure \ref{fig:graph-k4} -- \ref{fig:graph-icosidodecahedron} for some $s \geq 1$,
  \item\label{item:type-iii-bicolored-2} an $s$-subdivision of the cycle $C_\ell$ for some $\ell \geq 3$ and $s \geq 1$, or
  \item\label{item:type-iii-bicolored-3} an $s$-subdivision of the complete graph $K_2$ for some $s \geq 1$.
 \end{enumerate}
 
 First consider Option \ref{item:type-iii-bicolored-3}.
 In this case, $G[c]$ is isomorphic to $K_{2,s}$ for some $s \geq 1$.
 Also, since $G[c]$ has Type III, it holds that $s \geq 3$.
 This case is covered by Part \ref{item:type-iii-classification-1}.
 
 For the other two options, suppose $D \coloneqq \{\chi(v,v),\chi(w,w) \mid v,w \in V(G), \chi(v,w) = c\}$ is the set of vertex colors of $G[c]$.
 Recall that $|D| = 2$, i.e., $G[c]$ is bicolored.
 Suppose $D = \{d_1,d_2\}$, where $V_{d_1}$ contains all degree-$2$-vertices of $G[c]$ and $V_{d_2}$ contains the remaining vertices of $G[c]$ (which have degree $r \geq 3$).
 
 We first argue that $s \leq 2$ for Options \ref{item:type-iii-bicolored-1} and \ref{item:type-iii-bicolored-2}.
 For ease of notation, let $A \coloneqq V(G[c])$.
 Observe that $A = V_{d_1} \cup V_{d_2}$ holds.
 
 \begin{claim}
  \label{claim:no-3-subdivision}
  Suppose $G[c]$ is isomorphic to an $s$-subdivision of a graph $H$ such that $|V(H)| \geq 3$.
  Also suppose that $H$ as well as all $H - \{v_1,v_2\}$ for $v_1v_2 \in E(H)$ are connected.
  Then $s \leq 2$.
 \end{claim}
 \begin{claimproof}
  Let $u_1,\dots,u_s \in V_{d_1}$ be a set of $s$ distinct vertices and $v_1,v_2 \in V_{d_2}$ such that $N_{G[c]}(u_i) = N_{G[c]}(u_j) = \{v_1,v_2\}$ for all $i,j \in [s]$.
  Also, let $v_3 \in V_{d_2}$ be another vertex that is distinct from $v_1,v_2$ (which exists since $|V(H)| \geq 3$).
  Since $G$ is $3$-connected, there is a path $u_1 = w_0,\dots,w_\ell = v_3$ in $G$ that visits neither $v_1$ nor $v_2$.
  Now, consider the last occurrence $w_i$ of a vertex from $u_1,\dots,u_s$ on the path.
  Also, consider the first occurrence $w_j$, $j > i$, of a vertex from $A \setminus \{u_1,\dots,u_s\}$ on the path.
  Consider the subpath $P = w_i,\dots,w_j$.
  We have that $P$ avoids $D$ and it connects $w_i \in \{u_1,\dots,u_s\}$ to $w_j \in A \setminus \{u_1,\dots,u_s,v_1,v_2\}$.
  Using Observation \ref{obs:wl-knows-paths-avoiding-colors}, we conclude that such a path $P$ exists for all vertices $u_1,\dots,u_s$, i.e., for every $i \in [s]$, there is a path $P_i$ from $u_i$ to some vertex $x_i \in A \setminus \{u_1,\dots,u_s,v_1,v_2\}$ that avoids $D$.
  
  Now, we use these paths $P_i$ to construct a minor $K_{3,s}$ as follows.
  Since $H - \{v_1,v_2\}$ is connected for every $v_1v_2 \in E(H)$, it follows that $G[c] - \{u_1,\dots,u_s,v_1,v_2\}$ is connected.
  Let $X$ be the connected component of $G - \{u_1,\dots,u_s,v_1,v_2\}$ that contains $A \setminus \{u_1,\dots,u_s,v_1,v_2\}$.
  We contract $X$ to a single vertex.
  We have that $v_1,v_2 \in N_G(X)$ since $A \setminus \{u_1,\dots,u_s,v_1,v_2\} \subseteq X$ and $H$ is connected.
  Also $u_i \in N_G(X)$, since the path $P_i$ avoids $D$ and its endpoint $x_i$ is contained in $X$.
  In total, this gives a minor $K_{3,s}$ with vertices $v_1,v_2,X$ on the left side and vertices $u_1,\dots,u_s$ on the right side.
  Since $G$ is planar, it follows with Theorem \ref{thm:wagner} that $s \leq 2$.
 \end{claimproof}

 Now, consider Option \ref{item:type-iii-bicolored-2}, i.e., $G[c]$ is isomorphic to an $s$-subdivision of a cycle $C_\ell$ for $\ell \geq 3$ and $s \geq 1$.
 Then $s \leq 2$ by Claim \ref{claim:no-3-subdivision}.
 Also, $s \geq 2$ since $G[c]$ has Type III.
 So $s = 2$, which is covered by Part \ref{item:type-iii-classification-2}.
 
 So it remains to consider Option \ref{item:type-iii-bicolored-1}, i.e., $G[c]$ is isomorphic to an $s$-subdivision of one of the graphs from Figure \ref{fig:graph-k4} -- \ref{fig:graph-icosidodecahedron} for some $s \geq 1$.
 Again, $s \leq 2$ by Claim \ref{claim:no-3-subdivision} (observe that Claim \ref{claim:no-3-subdivision} is applicable since the graphs from Figure \ref{fig:graph-k4} -- \ref{fig:graph-icosidodecahedron} are all $3$-connected).
 The case $s = 1$ is covered by Part \ref{item:type-iii-classification-5}.
 Also, for the graphs from Figure \ref{fig:graph-k4} -- \ref{fig:graph-icosahedron}, the case $s = 2$ is covered by Part \ref{item:type-iii-classification-6}.
 Hence, it remains to prove that $G[c]$ is not isomorphic to a $2$-subdivision of a cuboctahedron or an icosidodecahedron.
 
 Suppose towards a contradiction that $G[c]$ is isomorphic to a $2$-subdivision of a cuboctahedron.
 Let $H$ denote the cuboctahedron with vertex set $V_{d_2}$ (i.e., $G[c]$ is isomorphic to the $2$-subdivision of $H$).
 As before, let $u_1,u_2 \in V_{d_1}$ be distinct vertices and $v_1,v_2 \in V_{d_2}$ such that $N_{G[c]}(u_1) = N_{G[c]}(u_2) = \{v_1,v_2\}$.
 Let us fix a planar embedding of $G$ and consider the induced embedding of $G[c]$.
 Let $F,F_1,F_2$ denote the three faces so that $u_1$ and $u_2$ are incident to $F$, $u_1$ is incident to $F_1$ but not $F_2$, and $u_2$ is incident to $F_2$ but not $F_1$.
 We have that $|F| = 4$ and $\{|F_1|,|F_2|\} = \{6,8\}$.
 Without loss of generality, suppose that $|F_1| = 6$.
 A visualization is given in Figure \ref{fig:type-iii-classification}.
 
 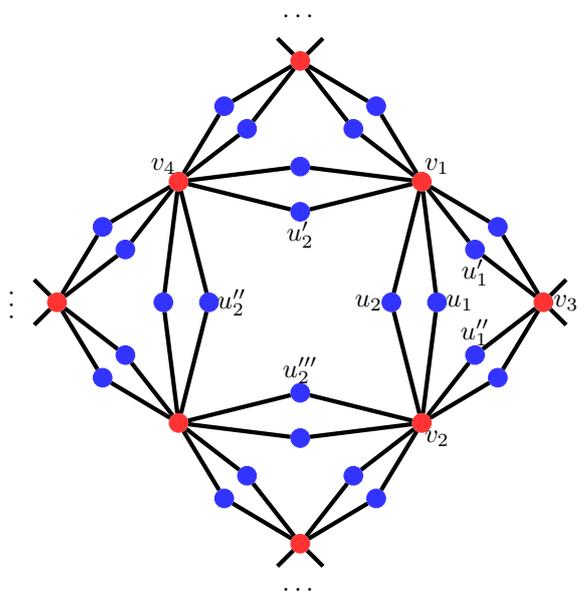
\begin{figure}
  \centering
  \begin{tikzpicture}
   \node[vertex,red!80] (1) at (1.6,1.6) {};
   \node[vertex,red!80] (2) at (-1.6,1.6) {};
   \node[vertex,red!80] (3) at (-1.6,-1.6) {};
   \node[vertex,red!80] (4) at (1.6,-1.6) {};
   
   \node[vertex,red!80] (5) at (3.2,0) {};
   \node[vertex,red!80] (6) at (0,3.2) {};
   \node[vertex,red!80] (7) at (-3.2,0) {};
   \node[vertex,red!80] (8) at (0,-3.2) {};
   
   \node[vertex,blue!80] (u1-1) at (1.2,0) {};
   \node[vertex,blue!80] (u2-1) at (0,1.2) {};
   \node[vertex,blue!80] (u3-1) at (-1.2,0) {};
   \node[vertex,blue!80] (u4-1) at (0,-1.2) {};
   \node[vertex,blue!80] (u1-2) at (1.8,0) {};
   \node[vertex,blue!80] (u2-2) at (0,1.8) {};
   \node[vertex,blue!80] (u3-2) at (-1.8,0) {};
   \node[vertex,blue!80] (u4-2) at (0,-1.8) {};
   
   \node[vertex,blue!80] (u5-1) at (2.3,0.7) {};
   \node[vertex,blue!80] (u6-1) at (0.7,2.3) {};
   \node[vertex,blue!80] (u7-1) at (-2.3,0.7) {};
   \node[vertex,blue!80] (u8-1) at (-0.7,2.3) {};
   \node[vertex,blue!80] (u9-1) at (-2.3,-0.7) {};
   \node[vertex,blue!80] (u10-1) at (-0.7,-2.3) {};
   \node[vertex,blue!80] (u11-1) at (2.3,-0.7) {};
   \node[vertex,blue!80] (u12-1) at (0.7,-2.3) {};
   
   \node[vertex,blue!80] (u5-2) at (2.6,1.0) {};
   \node[vertex,blue!80] (u6-2) at (1.0,2.6) {};
   \node[vertex,blue!80] (u7-2) at (-2.6,1.0) {};
   \node[vertex,blue!80] (u8-2) at (-1.0,2.6) {};
   \node[vertex,blue!80] (u9-2) at (-2.6,-1.0) {};
   \node[vertex,blue!80] (u10-2) at (-1.0,-2.6) {};
   \node[vertex,blue!80] (u11-2) at (2.6,-1.0) {};
   \node[vertex,blue!80] (u12-2) at (1.0,-2.6) {};
   
   \foreach \i/\j in {1/1,1/2,2/2,2/3,3/3,3/4,4/4,4/1,5/5,1/5,1/6,6/6,6/8,2/8,2/7,7/7,7/9,3/9,3/10,8/10,8/12,4/12,4/11,5/11}{
    \foreach \c in {1,2}{
     \draw[line width=1.6pt] (\i) edge (u\j-\c);
    }
   }
   
   \draw[line width=1.6pt] (5) edge (3.5,0.3);
   \draw[line width=1.6pt] (5) edge (3.5,-0.3);
   \node[rotate=90] at (3.8,0) {$\dots$};
   \draw[line width=1.6pt] (6) edge (0.3,3.5);
   \draw[line width=1.6pt] (6) edge (-0.3,3.5);
   \node at (0,3.8) {$\dots$};
   \draw[line width=1.6pt] (7) edge (-3.5,0.3);
   \draw[line width=1.6pt] (7) edge (-3.5,-0.3);
   \node[rotate=90] at (-3.8,0) {$\dots$};
   \draw[line width=1.6pt] (8) edge (0.3,-3.5);
   \draw[line width=1.6pt] (8) edge (-0.3,-3.5);
   \node at (0,-3.8) {$\dots$};
    
   \node at (1.8,1.8) {$v_1$};
   \node at (1.8,-1.8) {$v_2$};
   \node at (2.1,0) {$u_1$};
   \node at (0.9,0) {$u_2$};
   
   \node at (3.5,0) {$v_3$};
   \node at (-1.8,1.8) {$v_4$};
   
   \node at (2.3,0.4) {$u_1'$};
   \node at (2.3,-0.4) {$u_1''$};
   
   \node at (0,0.9) {$u_2'$};
   \node at (-0.9,0) {$u_2''$};
   \node at (0,-0.9) {$u_2'''$};
  \end{tikzpicture}
  \caption{Visualization for the proof of Theorem \ref{thm:type-iii-classification}. The vertices from $V_{d_1}$ are colored {\sf blue} and the vertices from $V_{d_2}$ are colored {\sf red}. Also, edges of the graph $G[c]$ are {\sf black}.}
  \label{fig:type-iii-classification}
 \end{figure}
 
 As in the proof of Claim \ref{claim:no-3-subdivision}, there is a path $P$ from $u_1$ to some vertex $x_1 \in A \setminus \{u_1,u_2,v_1,v_2\}$ such that $P$ avoids $D$.
 For $u \in A$, let
 \[R(u) \coloneqq \{v \in A \mid \text{there is a path from } u \text{ to } v \text{ that avoids } D\}.\]
 Observe that $R(u_1) \subseteq V(F) \cup V(F_1)$ by planarity (recall that $V(F)$ denotes the set of all vertices incident to the face $F$).
 Similarly, $R(u_2) \subseteq V(F) \cup V(F_2)$.
 
 Now, we distinguish several cases depending on which elements are contained in $R(u_1)$.
 Clearly, $v_1,v_2 \in R(u_1)$ (via trivial paths of length $1$).
 We have that $|V(F_1) \cap V_{d_2}| = 3$ and $v_1,v_2 \in V(F_1) \cap V_{d_2}$.
 Let $v_3$ denote the third vertex from $V(F_1) \cap V_{d_2}$.
 First suppose that $v_3 \in R(u_1)$.
 Since $R(u_1) \subseteq V(F) \cup V(F_1)$ and $(V(F) \cup V(F_1)) \cap V_{d_2} = \{v_1,v_2,v_3\}$, it follows that $|R(u_1) \cap V_{d_2}| = 3$.
 Since $\WL{2}{G}(u_1,u_1) = \WL{2}{G}(u_2,u_2)$ and using Observation \ref{obs:wl-knows-paths-avoiding-colors}, we get that $|R(u_2) \cap V_{d_2}| = 3$.
 Clearly $v_1,v_2 \in R(u_2)$.
 Let $v_4 \in R(u_2) \cap V_{d_2}$ denote the third vertex in the set.
 We have that $v_4 \in V(F_2)$ and either $v_1v_4 \in E(H)$ or $v_2v_4 \in E(H)$.
 Without loss of generality, suppose that $v_1v_4 \in E(H)$.
 Now, we claim that $\WL{2}{G}(u_2,v_1) \neq \WL{2}{G}(u_2,v_2)$.
 Towards this end, for all $u \in V_{d_1}$ and $v \in V_{d_2}$ such that $uv \in E(G[c])$, consider the set
 \[M(u,v) \coloneqq \{w \in V_{d_2} \mid vw \in E(H) \text{ and there is a path from } u \text{ to } w \text{ that avoids } D\}.\]
 We have that $M(u_2,v_1) = \{v_2,v_4\}$ and $M(u_2,v_2) = \{v_1\}$.
 So $|M(u_2,v_1)| \neq |M(u_2,v_2)|$.
 Using Observation \ref{obs:wl-knows-paths-avoiding-colors}, it follows that $\WL{2}{G}(u_2,v_1) \neq \WL{2}{G}(u_2,v_2)$ since $2$-WL ``knows the cardinality of $M(u,v)$''.
 But this is a contradiction since $\WL{2}{G}(u_2,v_1) = c =  \WL{2}{G}(u_2,v_2)$.
 
 So $v_3 \notin R(u_1)$.
 Observe that $|V(F_1) \cap V_{d_1}| = 3$ and $u_1 \in V(F_1) \cap V_{d_1}$.
 Let $u_1',u_1''$ denote the other two vertices contained in $V(F_1) \cap V_{d_1}$ so that $v_1u_1',v_2u_1'' \in E(G[c])$.
 Since $R(u_1) \setminus \{u_1,u_2,v_1,v_2\} \neq \emptyset$, $R(u_1) \subseteq V(F) \cup V(F_1)$ and $V(F) = \{u_1,u_2,v_1,v_2\}$, we conclude that $u_1' \in R(u_1)$ or $u_1'' \in R(u_1)$.
 Without loss of generality, suppose that $u_1' \in R(u_1)$.
 Using similar arguments as above, it follows that $u_1'' \in R(u_1)$ since otherwise $\WL{2}{G}(u_1,v_1) \neq \WL{2}{G}(u_1,v_2)$.
 
 Now, we say that $u,u' \in V_{d_1}$ are \emph{$c$-twins} if $N_{G[c]}(u) = N_{G[c]}(u')$.
 Observe that $u_1$ and $u_2$ are $c$-twins.
 Consider the auxiliary graph $F$ with $V(F) \coloneqq V_{d_1}$ and
 \begin{align*}
  E(F) \coloneqq \big\{uu' \;\big|\; &u \text{ and } u' \text{ are not } c \text{-twins and } \dist_{G[c]}(u,u') = 2 \text{ and}\\
                                     &\text{there is a path from } u \text{ to } u' \text{ that avoids } D \big\}
 \end{align*}
 By the comments above, we have that $N_F(u_1) = \{u_1',u_1''\}$.
 Also, using Observation \ref{obs:wl-knows-paths-avoiding-colors}, we conclude that $\WL{2}{G}|_{(V(F))^2} \preceq \WL{2}{F}$ and hence, $F$ is a regular graph.
 It follows that $F$ is $2$-regular.
 Let $\{u_2,u_2',u_2'',u_2'''\} = V(F_2) \cap V_{d_1}$.
 Now, $F[\{u_2,u_2',u_2'',u_2'''\}]$ has to be a $4$-cycle, but $F[\{u_1,u_1',u_1''\}]$ has to be a $3$-cycle.
 But this means that $\WL{2}{F}(u_1,u_1) \neq \WL{2}{F}(u_2,u_2)$.
 It follows that $\WL{2}{G}(u_1,u_1) \neq \WL{2}{G}(u_2,u_2)$ which is a contradiction.
 
 So $G[c]$ is not isomorphic to a $2$-subdivision of a cuboctahedron.
 The same argument can be used to show that $G[c]$ is also not isomorphic to a $2$-subdivision of an icosidodecahedron.
\end{proof}

Observe that the classification provided in Theorem \ref{thm:type-iii-classification} is optimal in the sense that every graph listed in the theorem can actually appear as a graph $G[c]$ for some edge color $c$ within a $3$-connected planar graph.
An easy way to see this is to take one of the graphs listed in the theorem, embed this graph $H$ in the plane, place a fresh vertex $v_F$ into every face $F$ and connect it to all vertices on the boundary of $F$. 
The resulting graph $G$ is $3$-connected and planar, and $H = G[c]$ for some edge color $c \in C_E(G,\WL{2}{G})$.

Also note that for a $3$-connected planar graph $G$ and an edge color $c \in C_E(G,\WL{2}{G})$ of Type III, the automorphism group $\Aut(G)$ is isomorphic to a subgroup of $\Aut(G[c])$.
Indeed, every automorphism $\gamma \in \Aut(G)$ naturally restricts to an automorphism $\gamma|_{V(G[c])}$ of $G[c]$, since the coloring computed by $2$-WL is invariant.
This gives rise to a homomorphism $\varphi\colon \Aut(G) \rightarrow \Aut(G[c])\colon \gamma \mapsto \gamma|_{V(G[c])}$.
By Theorem \ref{thm:type-iii-classification} and Lemma \ref{la:fixing-cycle-and-extra-vertex}, we obtain that $\Disc_G(w_1,\dots,w_\ell) = V(G)$ where $\{w_1,\dots,w_\ell\} = V(G[c])$.
This implies that the kernel of $\varphi$ is trivial, which implies that $\Aut(G)$ is isomorphic to a subgroup of $\Aut(G[c])$.

\section{Disjoint Unions of Cycles}
\label{sec:cycles}

\begin{figure}
 \centering
 \begin{subfigure}[b]{.48\linewidth}
  \centering
  \scalebox{\figscalesmall}{
  \begin{tikzpicture}
   
   \node[vertex,red!80] (1) at ($(0,0)+(330:0.6)$) {};
   \node[vertex,red!80] (2) at ($(0,0)+(90:0.6)$) {};
   \node[vertex,red!80] (3) at ($(0,0)+(210:0.6)$) {};
   
   \node[vertex,blue!80] (4) at ($(0,0)+(30:1.1)$) {};
   \node[vertex,blue!80] (5) at ($(0,0)+(150:1.1)$) {};
   \node[vertex,blue!80] (6) at ($(0,0)+(270:1.1)$) {};
   
   \node[vertex,springgreen] (7) at ($(0,0)+(315:1.5)$) {};
   \node[vertex,springgreen] (8) at ($(0,0)+(345:1.5)$) {};
   \node[vertex,springgreen] (9) at ($(0,0)+(75:1.5)$) {};
   \node[vertex,springgreen] (10) at ($(0,0)+(105:1.5)$) {};
   \node[vertex,springgreen] (11) at ($(0,0)+(195:1.5)$) {};
   \node[vertex,springgreen] (12) at ($(0,0)+(225:1.5)$) {};
   
   \node[vertex,springgreen] (13) at ($(0,0)+(30:3.0)$) {};
   \node[vertex,springgreen] (14) at ($(0,0)+(150:3.0)$) {};
   \node[vertex,springgreen] (15) at ($(0,0)+(270:3.0)$) {};
   
   \node[vertex,red!80] (16) at ($(0,0)+(330:2.5)$) {};
   \node[vertex,red!80] (17) at ($(0,0)+(90:2.5)$) {};
   \node[vertex,red!80] (18) at ($(0,0)+(210:2.5)$) {};
   
   \node[vertex,blue!80] (19) at ($(0,0)+(305:2.4)$) {};
   \node[vertex,blue!80] (20) at ($(0,0)+(355:2.4)$) {};
   \node[vertex,blue!80] (21) at ($(0,0)+(65:2.4)$) {};
   \node[vertex,blue!80] (22) at ($(0,0)+(115:2.4)$) {};
   \node[vertex,blue!80] (23) at ($(0,0)+(185:2.4)$) {};
   \node[vertex,blue!80] (24) at ($(0,0)+(235:2.4)$) {};
   
   \node[vertex,red!80] (25) at ($(0,0)+(315:3.2)$) {};
   \node[vertex,red!80] (26) at ($(0,0)+(345:3.2)$) {};
   \node[vertex,red!80] (27) at ($(0,0)+(75:3.2)$) {};
   \node[vertex,red!80] (28) at ($(0,0)+(105:3.2)$) {};
   \node[vertex,red!80] (29) at ($(0,0)+(195:3.2)$) {};
   \node[vertex,red!80] (30) at ($(0,0)+(225:3.2)$) {};
   
   \node[vertex,blue!80] (31) at ($(0,0)+(330:4.0)$) {};
   \node[vertex,springgreen] (32) at ($(0,0)+(30:4.0)$) {};
   \node[vertex,blue!80] (33) at ($(0,0)+(90:4.0)$) {};
   \node[vertex,springgreen] (34) at ($(0,0)+(150:4.0)$) {};
   \node[vertex,blue!80] (35) at ($(0,0)+(210:4.0)$) {};
   \node[vertex,springgreen] (36) at ($(0,0)+(270:4.0)$) {};
   
   \foreach \i/\j in {1/7,1/8,2/9,2/10,3/11,3/12,7/16,8/16,9/17,10/17,11/18,12/18,15/25,25/36,13/26,26/32,13/27,27/32,14/28,28/34,14/29,29/34,15/30,30/36}{
    \draw[line width=1.6pt] (\i) edge (\j);
   }
   \foreach \i/\j in {1/2,1/3,2/3,16/25,16/26,25/26,17/27,17/28,27/28,18/29,18/30,29/30}{
    \draw[line width=1.6pt,darkpastelgreen] (\i) edge (\j);
   }
   \foreach \i/\j in {1/4,1/6,2/4,2/5,3/5,3/6,16/19,16/20,17/21,17/22,18/23,18/24,19/25,20/26,21/27,22/28,23/29,24/30,25/31,26/31,27/33,28/33,29/35,30/35}{
    \draw[line width=1.6pt,lipicsYellow] (\i) edge (\j);
   }
   \foreach \i/\j in {4/8,4/9,5/10,5/11,6/7,6/12,7/19,8/20,9/21,10/22,11/23,12/24,13/20,13/21,14/22,14/23,15/19,15/24,31/32,31/36,32/33,33/34,34/35,35/36}{
    \draw[line width=1.6pt,magenta] (\i) edge (\j);
   }

  \end{tikzpicture}
  }
  \caption{A graph $G$ of Type IIa where each edge color has Type IIa.
   The {\sf black} and {\sf green} edges induce a connected subgraph that is isomorphic to a parallel subdivision of a truncated tetrahedron.}
  \label{fig:type-iia-example}
 \end{subfigure}
 \hfill
 \begin{subfigure}[b]{.48\linewidth}
  \centering
  \scalebox{\figscalesmall}{
  \begin{tikzpicture}
   \foreach \i in {0,...,7}{
    \node[vertex,springgreen] (a\i) at ($(0,0)+(22.5+\i*45:0.8)$) {};
   }
   \foreach \i in {0,...,7}{
    \node[vertex,blue!80] (b\i) at ($(0,0)+(\i*45:1.6)$) {};
   }
   \foreach \i in {0,...,7}{
    \node[vertex,red!80] (c\i) at ($(0,0)+(\i*45:2.4)$) {};
   }
   \foreach \i in {0,...,7}{
    \node[vertex,blue!80] (d\i) at ($(0,0)+(\i*45:3.2)$) {};
   }
   \foreach \i in {0,...,7}{
    \node[vertex,springgreen] (e\i) at ($(0,0)+(22.5+\i*45:4.0)$) {};
   }
   
   \foreach \i/\j in {0/1,1/2,2/3,3/4,4/5,5/6,6/7,7/0}{
    \draw[line width=1.6pt] (c\i) edge (c\j);
   }
   \foreach \i in {0,...,7}{
    \draw[line width=1.6pt,darkpastelgreen] (b\i) edge (c\i);
    \draw[line width=1.6pt,darkpastelgreen] (c\i) edge (d\i);
   }
   \foreach \i/\j in {0/1,1/2,2/3,3/4,4/5,5/6,6/7,7/0}{
    \draw[line width=1.6pt,lipicsYellow] (b\i) edge (b\j);
    \draw[line width=1.6pt,lipicsYellow] (d\i) edge (d\j);
   }
   \foreach \i/\j in {0/1,1/2,2/3,3/4,4/5,5/6,6/7,7/0}{
    \draw[line width=1.6pt,magenta] (a\i) edge (a\j);
    \draw[line width=1.6pt,magenta] (e\i) edge (e\j);
   }
   \foreach \i/\j in {0/0,0/1,1/1,1/2,2/2,2/3,3/3,3/4,4/4,4/5,5/5,5/6,6/6,6/7,7/7,7/0}{
    \draw[line width=1.6pt,violet!80] (a\i) edge (b\j);
    \draw[line width=1.6pt,violet!80] (e\i) edge (d\j);
   }
   
  \end{tikzpicture}
  }
  \caption{A graph $G$ of Type IIb.
   The color {\sf black} has Type IIb, {\sf yellow}, {\sf violet} and {\sf pink} have Type IIa, and {\sf green} has Type I.
   Note that $\Aut(G)$ is isomorphic to $(\Aut(G[{\sf black}])) \times \ZZ_2$ because, after individualizing all {\sf red} vertices, we can only swap the ``interior'' and the ``exterior'' region of the {\sf black} cycle.}
  \label{fig:type-iib-example}
 \end{subfigure}
 \caption{Two $3$-connected planar graphs of Type II. All vertices and edges are colored by their $2$-WL colors. For visualization purposes, we only color edges and do not distinguish between potentially different colors of two arcs $(v,w)$ and $(w,v)$.}
 \label{fig:type-ii-examples}
\end{figure}
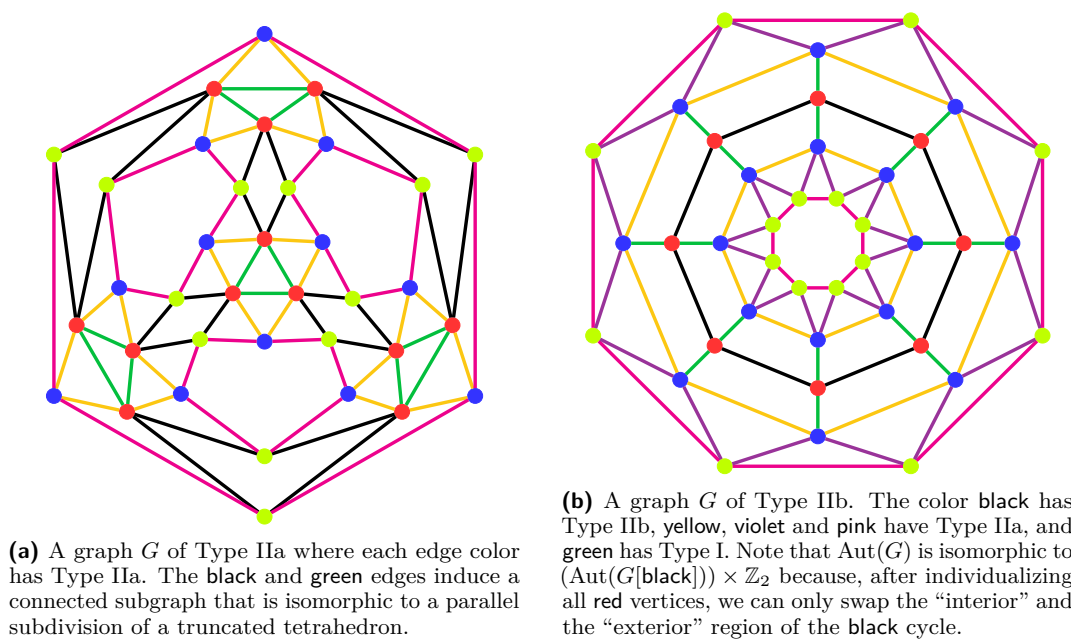

In this section, we consider $3$-connected planar graphs of Type II.
Let $G$ be a $3$-connected planar graph and let $\chi \coloneqq \WL{2}{G}$ be the coloring computed by $2$-WL.
Suppose that $G$ has Type II, i.e., there is an edge color $c \in C_E(G,\chi)$ of Type II, but there is no edge color of Type III.
As before, we wish to understand in which ways edge colors of Type II can occur in $G$.
More precisely, similarly to the case where $G$ has Type III, our goal is to identify and classify connected subgraphs defined by few edge colors. 
Towards this end, we define three subcategories of edge colors of Type II.
Let $c \in C_E(G,\chi)$ be of Type II.
If $G[c]$ is unicolored and $\{(v,w) \mid \chi(v,w) = c\} \neq \{(v,w) \mid \chi(w,v) = c\}$ (i.e., $G[c]$ is a disjoint union of directed cycles), then we say that $c$ has Type IIc.
If $c$ does not have Type IIc, but $G[c]$ is connected, then we say that $c$ has Type IIb.
If $c$ does not have Type IIc, and $G[c]$ is not connected, then we say that $c$ has Type IIa.

Let us remark that the main point of the subtypes is to distinguish between edge colors $c$ of Type II that induce non-connected subgraphs (Type IIa) and those that induce connected subgraphs (Type IIb).
The reason why we additionally single out the directed cycles (Type IIc) is that the existence of an edge color of Type IIc almost always (i.e., with the exception of one graph family) implies that the graph has fixing number $1$, because individualizing a single vertex in a directed cycle fixes all other vertices on the cycle as well.
In particular, we can show that every graph that contains an edge color of Type IIc is identified by $2$-WL.

We also point out that the existence of an edge color of Type IIb immediately puts severe restrictions on the structure of $G$.
For example, if $c \in C_E(G,\chi)$ has Type IIb, it can be shown that $\Aut(G)$ is isomorphic to a subgroup of $(\Aut(G[c])) \times \ZZ_2$ because, after individualizing all vertices of $G[c]$, an automorphism of $G$ can only swap the ``interior'' and the ``exterior'' region of the cycle (see also Figure \ref{fig:type-iib-example}).

Recall that the type of $G$ is defined as the maximal type of any of its edge colors.
We extend this definition to subtypes in the natural way.
For example, $G$ has Type IIb if there is an edge color $c \in C_E(G,\chi)$ of Type IIb, but no edge color of Type III or IIc.
Two example graphs of Type II are given in Figure \ref{fig:type-ii-examples}.

\subsection{Directed Cycles}

We start by analyzing graphs that contain an edge color $c$ of Type IIc.
As indicated above, this is a particularly well-behaved case because we can precisely classify those graphs that do not have fixing number $1$.
The \emph{bipyramid (of order $m \geq 3$)} is the graph $P_m^*$ with vertex set $V(P_m^*) \coloneqq \{u_1,u_2\} \cup \{v_i \mid i \in [m]\}$ and edge set $E(P_m^*) = \{u_iv_j \mid i \in [2], j \in [m]\} \cup \{v_iv_{i+1} \mid i \in [\ell-1]\} \cup \{v_1v_\ell\}$.

\begin{lemma}
 \label{la:directed-cycle}
 Let $G$ be a $3$-connected planar graph.
 Also let $c \in C_E(G,\WL{2}{G})$ be an edge color of Type IIc.
 Then there is a vertex $v \in V(G)$ such that $\Disc_G(v) = V(G)$, or $G$ is isomorphic to a bipyramid.
\end{lemma}

\begin{figure}
 \centering
  \scalebox{\figscalesmall}{
  \begin{tikzpicture}
   \node[vertex,red!80] (1) at ($(0,0)+(0:1.0)$) {};
   \node[vertex,red!80] (2) at ($(0,0)+(90:1.0)$) {};
   \node[vertex,red!80] (3) at ($(0,0)+(180:1.0)$) {};
   \node[vertex,red!80] (4) at ($(0,0)+(270:1.0)$) {};
   
   \node[vertex,red!80] (5) at ($(0,0)+(45:2.0)$) {};
   \node[vertex,red!80] (6) at ($(0,0)+(135:2.0)$) {};
   \node[vertex,red!80] (7) at ($(0,0)+(225:2.0)$) {};
   \node[vertex,red!80] (8) at ($(0,0)+(315:2.0)$) {};
   
   \node[vertex,red!80] (9) at ($(0,0)+(0:2.4)$) {};
   \node[vertex,red!80] (10) at ($(0,0)+(90:2.4)$) {};
   \node[vertex,red!80] (11) at ($(0,0)+(180:2.4)$) {};
   \node[vertex,red!80] (12) at ($(0,0)+(270:2.4)$) {};
   
   \node[vertex,red!80] (13) at ($(0,0)+(45:3.2)$) {};
   \node[vertex,red!80] (14) at ($(0,0)+(135:3.2)$) {};
   \node[vertex,red!80] (15) at ($(0,0)+(225:3.2)$) {};
   \node[vertex,red!80] (16) at ($(0,0)+(315:3.2)$) {};
   
   \node[vertex,red!80] (17) at ($(0,0)+(0:3.6)$) {};
   \node[vertex,red!80] (18) at ($(0,0)+(90:3.6)$) {};
   \node[vertex,red!80] (19) at ($(0,0)+(180:3.6)$) {};
   \node[vertex,red!80] (20) at ($(0,0)+(270:3.6)$) {};
   
   \node[vertex,red!80] (21) at ($(0,0)+(45:6.0)$) {};
   \node[vertex,red!80] (22) at ($(0,0)+(135:6.0)$) {};
   \node[vertex,red!80] (23) at ($(0,0)+(225:6.0)$) {};
   \node[vertex,red!80] (24) at ($(0,0)+(315:6.0)$) {};
   
   \node[vertex,blue!80] (v1) at ($(0,0)+(17:1.6)$) {};
   \node[vertex,blue!80] (v2) at ($(0,0)+(107:1.6)$) {};
   \node[vertex,blue!80] (v3) at ($(0,0)+(197:1.6)$) {};
   \node[vertex,blue!80] (v4) at ($(0,0)+(287:1.6)$) {};
   
   \node[vertex,blue!80] (v5) at ($(0,0)+(35:3.9)$) {};
   \node[vertex,blue!80] (v6) at ($(0,0)+(125:3.9)$) {};
   \node[vertex,blue!80] (v7) at ($(0,0)+(215:3.9)$) {};
   \node[vertex,blue!80] (v8) at ($(0,0)+(305:3.9)$) {};

   \foreach \i/\j in {1/9,2/10,3/11,4/12,9/5,10/6,11/7,12/8,5/1,6/2,7/3,8/4,13/17,14/18,15/19,16/20,17/21,18/22,19/23,20/24,21/13,22/14,23/15,24/16}{
    \draw[->,line width=1.6pt] (\i) edge (\j);
   }
   
   \foreach \i/\j in {1/2,2/3,3/4,4/1,10/5,11/6,12/7,9/8,5/13,6/14,7/15,8/16,13/18,14/19,15/20,16/17,17/9,18/10,19/11,20/12,21/24,24/23,23/22,22/21}{
    \draw[->,line width=1.6pt,darkpastelgreen] (\i) edge (\j);
   }
   
   \foreach \i/\j in {2/5,3/6,4/7,1/8,9/13,10/14,11/15,12/16,18/21,19/22,20/23,17/24}{
    \draw[line width=1.6pt,magenta] (\i) edge (\j);
   }
   
   \foreach \i/\j in {1/1,1/5,1/9,2/2,2/6,2/10,3/3,3/7,3/11,4/4,4/8,4/12,5/13,5/17,5/21,6/14,6/18,6/22,7/15,7/19,7/23,8/16,8/20,8/24}{
    \draw[line width=1.6pt,lipicsYellow] (v\i) edge (\j);
   }
  \end{tikzpicture}
  }
  \caption{A graph $G$ of Type IIc.
   The colors {\sf black} and {\sf green} have Type IIc, and {\sf pink} and {\sf yellow} have Type I.
   Individualizing an arbitrary {\sf red} vertex and performing $1$-WL results in a discrete coloring.
   Hence, the graph is identified by $2$-WL.
   Also, setting $c$ to be {\sf black} in the proof of Lemma \ref{la:directed-cycle}, each set $B_i$ consists of a single {\sf blue} vertex whereas $M$ is the set of {\sf red} vertices.
  }
  \label{fig:type-iic-example}
\end{figure}
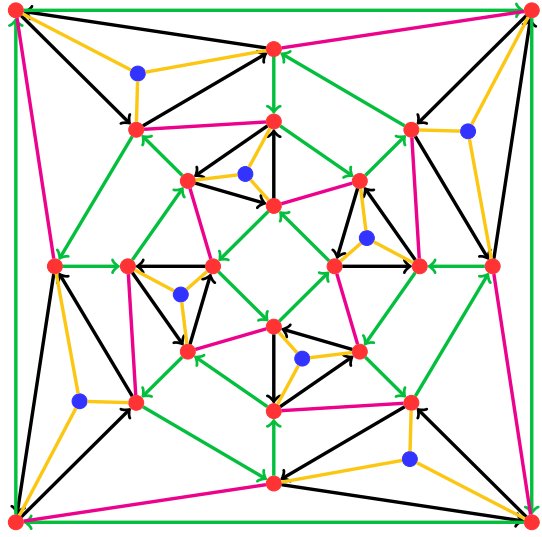

\begin{proof}
 Let $\chi \coloneqq \WL{2}{G}$ and let $A_1,\dots,A_\ell$ be the vertex sets of the connected components of $G[c]$.
 We distinguish several cases.
 An example is also given in Figure \ref{fig:type-iic-example}.
 
 First suppose $\ell \geq 2$.
 Let us fix a planar embedding of $G$.
 For every $i \in [\ell]$, the cycle $(G[c])[A_i]$ splits the plane into two regions, the interior region and the exterior region.
 We denote by $\exte_c(A_i)$ the set of vertices located in the exterior region, and $\inte_c(A_i)$ denotes the set of vertices located in the interior region.
 By Lemma \ref{la:path-between-color-components}, we may assume without loss of generality that $A_j \subseteq \exte_c(A_i)$ for all distinct $i,j \in[\ell]$.
 
 Let $D \coloneqq \{\chi(v,v) \mid v \in V(G[c])\}$.
 Recall that a path $u_0,\dots,u_m$ \emph{avoids $D$} if $\chi(u_i,u_i) \notin D$ for all $i \in [m-1]$, i.e., the color of no internal vertex of the path is contained in the set $D$.
 For a vertex $u \in V(G)$, we define
 \[\Theta_c(u) \coloneqq \{i \in [\ell] \mid \exists w \in A_i \colon \text{there is a path from $u$ to $w$ that avoids $D$}\}.\]
 Moreover, we define
 \[M \coloneqq \{u \in V(G) \mid |\Theta_c(u)| \geq 2\}.\]
 Since $2$-WL detects whether there is a path from $u$ to $w$ that avoids $D$ (see Observation \ref{obs:wl-knows-paths-avoiding-colors}), we conclude that $M$ is $\chi$-invariant.
 Moreover, since $G$ is connected, there is some $i \in [\ell]$ such that $A_i \cap M \neq \emptyset$.
 Since $M$ is $\chi$-invariant and $G[c]$ is unicolored, it follows that $A_i \subseteq M$ for all $i \in [\ell]$.
 Finally, for $i \in [\ell]$, we define
 \[B_i \coloneqq \{u \in V(G) \mid \Theta_c(u) = \{i\}\}.\]
 
 \begin{claim}
  $M \setminus A_i \subseteq \exte_c(A_i)$ for all $i \in [\ell]$.
 \end{claim}
 \begin{claimproof}
  This follows directly from the fact that $A_j \subseteq \exte_c(A_i)$ for all distinct $i,j \in [\ell]$.
  Indeed, suppose there are some $u \in M \setminus A_i$ and $i \in [\ell]$ such that $u \in \inte_c(A_i)$.
  By definition, $|\Theta_c(u)| \geq 2$, and there is some $i \neq j \in [\ell]$ and a vertex $w \in A_j$ such that $w$ is reachable from $u$ via a path that avoids $D$.
  But $w \in \exte_c(A_i)$.
  Hence, any path from $u$ to $w$ must visit some vertex $v \in A_i$ since $G$ is planar.
  This is a contradiction since $\chi(v,v) \in D$.
 \end{claimproof}
 
 Now suppose that $M = V(G)$.
 Then $\inte_c(A_i) = \emptyset$ for all $i \in [\ell]$, i.e., there are $v_1,v_2,v_3 \in A_i$ such that $v_1,v_2,v_3$ lie on a common face of $G$.
 Consider the set $\Disc_G(v_1)$.
 Clearly, $A_i \subseteq \Disc_G(v_1)$ since $A_i$ forms a directed cycle via edge color $c$.
 But this implies that $\Disc_G(v_1) = V(G)$ using Theorem \ref{thm:splitting-with-tutte}.
 
 \begin{claim}
  \label{claim:many-directed-cycles}
  Suppose there is some $i \in [\ell]$ such that $B_i \neq \emptyset$ and let $u \in A_i$.
  Then $\lambda \coloneqq \WL{1}{G,\WL{2}{G},u}$ is discrete.
 \end{claim}
 \begin{claimproof}
  First, since $G[c]$ forms a directed cycle on $A_i$, it holds that $|[v]_\lambda| = 1$ for all $v \in A_i$.
  Also, $B_i$ is $\lambda$-invariant.
  
  Consider the graph $H_i$ defined by $V(H_i) \coloneqq A_i \cup B_i \cup \{z\}$ (where $z \notin V(G)$ is a fresh vertex) and
  \[E(H_i) \coloneqq E_G(A_i \cup B_i,A_i \cup B_i) \cup \{zv \mid v \in A_i\}.\]
  Also, let $\lambda_{H_i} \coloneqq V(H_i) \rightarrow \{1,2\}$ be the vertex coloring defined via $\lambda_{H_i}(z) = 1$ and $\lambda_{H_i}(v) = 2$ for all $v \in V(H_i) \setminus \{z\}$.
  
  We first argue that $H_i$ is planar.
  Indeed, by Lemma \ref{la:path-between-color-components}, we conclude that $G[M \setminus A_i]$ is connected.
  Also, $A_i \subseteq M$, which means that for every $v \in A_i$, there are some $i \neq j \in [\ell]$ and $w \in A_j$ such that there is a path $P$ from $v$ to $w$ that avoids $D$.
  By definition of the set $M$, every vertex from $P$ is contained in $M$.
  This implies that $A_i \subseteq N_G(M \setminus A_i)$.
  So $H_i$ is a minor of $G$, which implies that $H_i$ is planar.
  
  Next, we argue that $H_i$ is $3$-connected.
  Suppose towards a contradiction that $S$ is a separator of $H_i$ of size $|S| \leq 2$.
  First, we may assume $z \notin S$ since $H_i[A_i]$ is $2$-connected.
  Let $w_1,w_2$ be two vertices from distinct connected components of $H_i - S$.
  Since $\deg_{H_i}(z) \geq 3$, we may assume that $w_j \neq z$ for both $j \in \{1,2\}$.
  Because $G$ is $3$-connected, there is a path $w_1 = u_0,\dots,u_p = w_2$ from $w_1$ to $w_2$ in $G - S$.
  Let $u_s$ be the first vertex on the path that is not contained in $A_i \cup B_i$, and $u_t$ be the last vertex on the path that is not contained in $A_i \cup B_i$.
  Since $N_G(B_i) \subseteq A_i$ it follows that $u_{s-1},u_{t+1} \in A_i$.
  But then $w_1 = u_0,\dots,u_{s-1},z,u_{t+1},\dots,u_p = w_2$ is a path from $w_1$ to $w_2$ in $H - S$, which is a contradiction.
  
  So together, $H_i$ is planar and $3$-connected.
  Thus, $\WL{1}{H_i,\lambda_{H_i},v_1,\dots,v_k}$ is discrete by Lemma \ref{la:fixing-cycle-and-extra-vertex}, where $A_i = \{v_1,\dots,v_k\}$ (note that $z$ is individualized as well).
  
  Now, consider the coloring $\lambda$.
  Since $B_i$ is $\lambda$-invariant and $|[v]_\lambda| = 1$ for all $v \in A_i$, we get that
  \[\lambda|_{B_i} \preceq (\WL{1}{H_i,\lambda_{H_i},v_1,v_2})|_{B_i},\]
  since any refinement done in the graph $H_i$ can also be performed in $G$.
  So $|[v]_\lambda| = 1$ for all $v \in A_i \cup B_i$.
  Since $B_i \neq \emptyset$, it follows that $\lambda$ is discrete by Lemma \ref{la:fixing-cycle-and-extra-vertex}.
 \end{claimproof}
 
 Hence, it remains to consider the case $\ell = 1$.
 Let $K_1,\dots,K_p$ be the vertex sets of the connected components of $G - A_1$.
 Note that $N_G(K_i) \subseteq A_1$ for all $i \in [p]$.
 
 We distinguish two further subcases.
 First suppose that $N_G(K_i) \neq N_G(K_j)$ for all distinct $i,j \in [p]$.
 
 Let $\lambda \coloneqq \WL{1}{G,\WL{2}{G},u}$ for some $u \in A_1$.
 Clearly, $\lambda(v_1) \neq \lambda(v_2)$ for all $v_1 \in A_1$ and $v_1 \neq v_2 \in V(G)$.  
 Also,
 \[\{\lambda(w) \mid w \in K_i\} \cap \{\lambda(w) \mid w \in K_j\} = \emptyset\]
 for all distinct $i,j \in [p]$.
 Let $i \in [p]$.
 Consider the graph $H_i$ defined via $V(H_i) \coloneqq K_i \cup A_1 \cup \{z\}$ (where $z \notin V(G)$ is a fresh vertex) and
 \[E(H_i) \coloneqq E_G(A_i \cup A_1,K_i \cup A_1) \cup \{zv \mid v \in A_1\}.\]
 Also, let $\lambda_{H_i} \coloneqq V(H_i) \rightarrow \{1,2\}$ be the vertex coloring defined via $\lambda_{H_i}(z) = 1$ and $\lambda_{H_i}(v) = 2$ for all $v \in V(H_i) \setminus \{z\}$.
 It is easy to see that $H$ is planar and $3$-connected.
 Moreover, $\WL{1}{H,\lambda_H,v_1,v_2}$ is discrete by Theorem \ref{thm:splitting-with-tutte} where $v_1,v_2 \in A_1$ are two arbitrary vertices such that $v_1v_2 \in E(G[c])$.
 
 As before, it holds that $\lambda|_{K_i} \preceq (\WL{1}{H,\lambda_H,v_1,v_2})|_{K_i}$ which implies that $|[v']_\lambda| = 1$ for all $v' \in A_i \cup K_i$.
 But now, there are three vertices $u_1,u_2,u_3 \in A_i \cup K_i$ that lie on a common face of $G$.
 Hence, $\lambda$ is discrete by Theorem \ref{thm:splitting-with-tutte}.
 
 \medskip
 
 Finally, assume there are $i,j \in [p]$ such that $N_G(K_i) = N_G(K_j)$.
 Note that, by $3$-connectedness, it holds that $|N_G(K_i)| = |N_G(K_j)| \geq 3$.
 
 \begin{claim}
  $N_G(K_i) = A_1$ and $p = 2$.
 \end{claim}
 \begin{claimproof}
  Let $v_1,v_2,v_3 \in N_G(K_i)$ be three distinct vertices such that $v_2$ appears on the directed path from $v_1$ to $v_3$ in $G[c]$.
  Also, let $S \subseteq A_1$ denote the segment between $v_1$ and $v_3$ (excluding $v_1$ and $v_3$).
  Observe that $v_2 \in S$.
  The color $\WL{2}{G}(v_1,v_3)$ encodes that $v_1$ and $v_3$ both have a neighbors in a common connected component of $G - A_1$ (e.g., $K_i$).
  Let $d \coloneqq |S|$.
  Then, since all (directed) edges in $G[c]$ have the same $\WL{2}{G}$-color, all vertex pairs $v,v' \in A_1$, for which the segment from $v$ to $v'$ on $G[c]$ has length $d$, must have a common neighbor in a connected component of $G - A_1$.
  In particular, this must hold for every choice of $v \in S$. In that case, $v' \in A_1 \setminus S$.
  Note that $v_1,v_3$ and vertices $w_1 \in N(v_1) \cap K_i, w'_1\in N(v_1) \cap K_j$ and $w_3 \in N(v_3) \cap K_i, w'_3 \in N(v_3) \cap K_i$ bound an area which contains $v$ via their connecting edges.
  Thus, by planarity, the common neighbor of $v$ and $v'$ must also belong to $K_i$ and, since we have assumed $N(K_i) = N(K_j)$, the vertices $v$ and $v'$ also have a common neighbor in $K_j$.
  An iterative application of this argument shows that $N_G(K_i) = N_G(K_j) = A_1$.
  Since $G$ is planar and $3$-connected, this also implies that $p = 2$.
 \end{claimproof}

 So $p = 2$ and $N_G(K_1) = N_G(K_2) = A_1$.
 If there is some $v \in K_1 \cup K_2$ for which $N(v)$ is a segment of $G[c]$ that is distinct from $A_1$, the coloring $\WL{1}{G,\WL{2}{G},v}$ is discrete by Lemma \ref{la:fixing-cycle-and-extra-vertex}.
 So suppose that, for all $v \in K_1 \cup K_2$, the neighborhood $N(v)$ is not a segment of $A_1$.
 In particular, $|N(v)| \geq 2$.
 Suppose that there is a vertex $v' \in K_1$ with $N(v') \subsetneq A_1$.
 Let $v_1,\dots,v_\ell$ denote the neighbors of $v'$ on $A_1$ in cyclic ordering (where $\ell \geq 2$).
 For $i \in [\ell]$, let $S_i$ be the segment of $G[c]$ from $v_i$ to $v_{i+1}$, including $v_i$ and $v_{i+1}$ (where indices are taken modulo $\ell$).
 By planarity, every vertex $v'' \neq v'$ with $v'' \in K_1$ satisfies $N(v'') \subseteq S_i$ for some $i \in [\ell]$.
 Now, pick some $i \in [\ell]$ such that $|S_i| > 2$ and $|S_i|$ is minimal among all sets $S_j$, $j \in [\ell]$ for which $|S_j| > 2$.
 Since $N_G(K_1) = A_1$ there is some vertex $v'' \in K_1$ such that $N(v'') \subseteq S_i$.
 Now, we can repeat the same argument for $v''$ instead of $v'$, leading to a segment $S_{i'}'$ defined in the same way as $S_i$.
 Note that $|S_{i'}'| < |S_i|$.
 However, this is a contraction since repeating this argument must eventually lead to some vertex $v''' \in K_1$ whose neighborhood is a segment.
 
 Thus, every vertex $v \in K_1 \cup K_2$ satisfies $N(v) = A_1$.
 Since $|A_1| \geq 3$, we conclude that $|K_1 \cup K_2| \leq 2$ using Theorem \ref{thm:wagner} (because there is a complete bipartite graph between $A_1$ and $K_1 \cup K_2$).
 It follows that $K_1$ and $K_2$ are singletons with vertices $u_1$ and $u_2$ respectively.
 So $G$ is isomorphic to a bipyramid.
\end{proof}

\begin{corollary}
 Let $G$ be a $3$-connected planar graph and suppose $G$ contains an edge color of Type IIc.
 Then $2$-WL determines pair orbits in $G$.
\end{corollary}

\begin{proof}
 If there is a vertex $v \in V(G)$ such that $\Disc_G(v) = V(G)$, then $2$-WL determines pair orbits in $G$ by Lemma \ref{la:wl-pair-orbits-from-fixing-number}.
 Otherwise, $G$ is a bipyramid and it is easy to check that $2$-WL determines pair orbits in $G$.
\end{proof}

\subsection{Definable Matchings}

In the remainder of this section, we analyze graphs of Type IIa and, similarly as before, aim at finding highly regular connected substructures.
(Observe that, if $G$ has Type IIb, then a witnessing edge color already provides such an object).
Unfortunately, we need to allow for two further possible outcomes. 
First, we are again satisfied with finding a vertex $v \in V(G)$ such that $\Disc_G(v) = V(G)$, which in particular implies that $2$-WL determines pair orbits on $G$ (see Lemma \ref{la:wl-pair-orbits-from-fixing-number}).
As the other potential outcome, we consider \emph{definable matchings}.

\begin{definition}
 Let $G$ be a graph and let $\chi\coloneqq \WL{2}{G}$.
 A color $c \in C_E(G,\chi)$ \emph{defines a matching} if for every $(v,w) \in \chi^{-1}(c)$, it holds that
 \begin{enumerate}[label=(\roman*)]
  \item $\chi(v,v) \neq \chi(w,w)$,
  \item $\{v' \in V(G) \mid \chi(v',w) = c\} = \{v\}$, and
  \item $\{w' \in V(G) \mid \chi(v,w') = c\} = \{w\}$.
 \end{enumerate}
\end{definition}

Suppose there is some $c \in C_E(G,\chi)$ that defines a matching.
Such a situation is generally beneficial since we can simply contract all edges of color $c$ and move to a smaller graph.
This operation neither changes the $2$-WL coloring (see Lemma \ref{la:factor-graph-2-wl}) nor identification of the graph by $2$-WL, as shown in the next lemma (see also \cite{FuhlbruckKV21}).

Recall the definition of the coloring $\chi/c$ given in Lemma \ref{la:factor-graph-2-wl}.

\begin{lemma}
 \label{la:definable-matching}
 Let $G$ be a graph and let $\chi\coloneqq \WL{2}{G}$.
 Also, let $c \in C_E(G,\chi)$ be an edge color that defines a matching.
 Suppose that $2$-WL determines arc orbits (resp.\ pair orbits) of the arc-colored graph $(G/c,\lambda)$, where $\lambda(X_1,X_2) \coloneqq (\chi/c)(X_1,X_2)$ for all $(X_1,X_2) \in A(G/c)$.
 Then $2$-WL determines arc orbits (resp.\ pair orbits) of $G$.
\end{lemma}

\begin{proof}
 We give the proof for \textquoteleft determining arc orbits\textquoteright, the arguments for \textquoteleft determining pair orbits\textquoteright\ are analogous.
 Let $G'$ be a second graph and define $\chi'\coloneqq \WL{2}{G'}$.
 Let $(v,w) \in A(G)$ and $(v',w') \in A(G')$ such that $\chi(v,w) = \chi'(v',w')$.
 Let $X,Y \in V(G/c)$ be the unique elements such that $v \in X$ and $w \in Y$.
 Similarly, pick $X',Y' \in V(G'/c)$ such that $v' \in X'$ and $w' \in Y'$.
 Then $(\chi/c)(X,Y) = (\chi'/c)(X',Y')$ holds by Lemma \ref{la:factor-graph-2-wl}.
 
 Let $\lambda'$ be the arc coloring of $G'/c$ defined by
 \[\lambda'(X_1,X_2) \coloneqq (\chi'/c)(X_1,X_2)\]
 for all $(X_1,X_2) \in A(G'/c)$.
 Then $(\WL{2}{G/c,\lambda})(X,Y) = (\WL{2}{G'/c,\lambda'})(X',Y')$ holds by Lemma~\ref{la:factor-graph-2-wl}.
 Since $2$-WL determines arc orbits of $(G/c,\lambda)$, there is an isomorphism $\psi\colon (G/c,\lambda) \cong (G'/c,\lambda')$ such that $\psi(X) = X'$ and $\psi(Y) = Y'$.
 
 Now define a bijection $\varphi\colon V(G) \rightarrow V(G')$ as follows.
 Let $u \in V(G)$ and let $Z \in V(G/c)$ be the unique set such that $u \in Z$.
 Since $c$ defines a matching, it holds that $|Z| \leq 2$ and $\chi(v_1,v_1) \neq \chi(v_2,v_2)$ for all distinct $v_1,v_2 \in Z$.
 Hence, there is a unique $u' \in \psi(Z)$ such that $\chi(u,u) = \chi'(u',u')$.
 Set $\varphi(u) \coloneqq u'$ and proceed analogously with all other vertices in $V(G)$.
 Clearly, $\varphi$ is a bijection and $\varphi(v) = v'$ and $\varphi(w) = w'$.
 So it remains to argue that $\varphi$ is an isomorphism between $G$ and $G'$.
 Let $u_1,u_2 \in V(G)$ such that $u_1u_2 \in E(G)$.
 Let $Z_1,Z_2 \in V(G/c)$ such that $u_1 \in Z_1$ and $u_2 \in Z_2$.
 Then $(\chi/c)(Z_1,Z_2) = (\chi'/c)(\psi(Z_1),\psi(Z_2))$.
 Since $c$ defines a matching, it follows that $\chi(u_1,u_2) = \chi'(\varphi(u_1),\varphi(u_2))$.
 In particular, $\varphi(u_1)\varphi(u_2) \in E(G')$.
\end{proof}

\subsection{Spring Embeddings}

Towards analyzing graphs of Type IIa, we also need some additional tools that allow us to prove that certain $3$-connected planar graphs $G$ have fixing number $1$, i.e., $\Disc_G(v) = V(G)$ for some $v \in V(G)$.
More precisely, we start by giving an extension of Theorem \ref{thm:splitting-with-tutte} that can be applied to planar graphs that are not necessarily $3$-connected.

Let $G$ be a planar graph and let $(T,\beta)$ denote its (unique) tree decomposition into triconnected components in the sense of Tutte \cite{Tutte84}. 
For $t \in V(T)$, we denote by $G[[\beta(t)]]$ the \emph{torso} of $G$ on the set $\beta(t)$, i.e., $V(G[[\beta(t)]]) \coloneqq \beta(t)$ and
\begin{align*}
 E(G[[\beta(t)]]) \coloneqq\;\;\;\;\; &E(G[\beta(t)])\\
                               \cup\; &\{vw \mid v,w \in N_G(Z) \text{ for some connected component $Z$ of } G - \beta(t)\}.
\end{align*}
The edges of the second type are usually called \emph{virtual edges} of $G[[\beta(t)]]$.

Now pick some $t \in V(T)$ and suppose $v_1,v_2,v_3$ forms a triangular face in $G[[\beta(t)]]$.
We inductively define mappings $\mu_i \colon V(G) \rightarrow \RR^2$ for all $i \geq 0$.
We set $\mu_0(v_1) \coloneqq (0,1)$, $\mu_0(v_2) \coloneqq (1,0)$, $\mu_0(v_3) \coloneqq (1,1)$, and $\mu(v) \coloneqq (0,0)$ for all $v \in V(G) \setminus \{v_1,v_2,v_3\}$.
For $i \geq 0$, we define
\[\mu_{i+1}(v) \coloneqq \begin{cases}
                          \mu_0(v)                                        &\text{if } v \in \{v_1,v_2,v_3\}\\
                          \frac{1}{\deg_G(v)} \sum_{w \in N_G(v)}\mu_i(w) &\text{otherwise}
                         \end{cases}.\]
This sequence converges to a mapping $\mu_\infty\colon V(G) \rightarrow \RR^2$ such that
\[\mu_\infty(v) = \frac{1}{\deg_G(v)} \sum_{w \in N_G(v)}\mu(w)\]
for all $v \in V(G) \setminus \{v_1,v_2,v_3\}$.
Tutte's Spring Embedding Theorem \cite{Tutte63} states that if $G$ is $3$-connected (and hence $T$ consists of a single node), then $\mu_\infty$ provides a straight-line embedding of $G$ in the plane.
In this paper, we exploit the slightly stronger statement that $\mu_\infty|_{\beta(t)}$ provides a straight-line embedding for $G[[\beta(t)]]$, i.e., the unique triconnected component of $G$ that contains $v_1,v_2,v_3$.
This extension can be proved in the same way as the original Spring Embedding Theorem.
A proof of Tutte's Spring Embedding Theorem is for example given in \cite[Chapter 3.2]{Lovasz19}, and this proof can directly be modified for the variant stated here.
More precisely, we shall use the following lemma.

\begin{lemma}
 \label{la:spring-embedding-properties}
 The mapping $\mu_\infty$ satisfies the following conditions:
 \begin{enumerate}[label=(\arabic*)]
  \item\label{item:spring-embedding-item-1} $\mu_\infty(v) \neq \mu_\infty(w)$ for all distinct $v,w \in \beta(t)$,
  \item\label{item:spring-embedding-item-2} $\mu_\infty(v) \neq \mu_\infty(w)$ for all $v \in Z$ and $w \in \beta(t)$ where $Z$ denotes a connected component of $G - \beta(t)$ for which $|N_G(Z)| = 2$, and
  \item\label{item:spring-embedding-item-3} $\mu_\infty(v) = \mu_\infty(w)$ for all $v \in Z$ and $w \in N_G(Z)$ where $Z$ denotes a connected component of $G - \beta(t)$ for which $|N_G(Z)| = 1$.
 \end{enumerate}
\end{lemma}

\begin{lemma}
 \label{la:coloring-from-spring-embedding}
 Let $G$ be a planar graph and let $(T,\beta)$ denote its decomposition into triconnected components.
 Let $t \in V(T)$ and suppose $v_1,v_2,v_3$ forms a triangular face in $G[[\beta(t)]]$.
 Also let $\chi\colon V(G) \rightarrow C$ be a coloring that is $1$-stable with respect to $G$, and $|[v_i]_{\chi}| = 1$ for all $i \in \{1,2,3\}$.
 Then $|[v]_\chi| = 1$ for all $v \in \beta(t)$.
\end{lemma}

\begin{proof}
 We first show by induction on $i \geq 0$ that
 \[\mu_i(v) \neq \mu_i(w) \;\;\;\Rightarrow\;\;\; \chi(v) \neq \chi(w)\]
 for all $v,w \in V(G)$.
 
 The base case $i = 0$ follows directly from the definition of $\mu_0$ and the fact that $|[v_i]_{\chi}| = 1$ for all $i \in \{1,2,3\}$.
 Suppose $i \geq 0$ and pick $v,w \in V(G)$ such that $\mu_{i+1}(v) \neq \mu_{i+1}(w)$.
 If $\{v,w\} \cap \{v_1,v_2,v_3\} \neq \emptyset$, then $\chi(v) \neq \chi(w)$, since $|[v_i]_{\chi}| = 1$ for all $i \in \{1,2,3\}$.
 So suppose that $v,w \notin \{v_1,v_2,v_3\}$.
 Then
 \[\{\!\!\{\mu_i(v') \mid v' \in N_G(v)\}\!\!\} \neq \{\!\!\{\mu_i(w') \mid w' \in N_G(w)\}\!\!\}.\]

 By the induction hypothesis, we conclude that $\chi(v) \neq \chi(w)$, since $\chi$ is $1$-stable.
 
 Since the sequence $(\mu_i)_{i \geq 0}$ converges to the function $\mu_\infty$, it follows that
 \[\mu_\infty(v) \neq \mu_\infty(w) \;\;\;\Rightarrow\;\;\; \chi(v) \neq \chi(w)\]
 for all $v,w \in V(G)$.
 Now let $v \in \beta(t)$ and $w \in V(G)$ be distinct vertices.
 To prove the lemma, it suffices to show that $\chi(v) \neq \chi(w)$.
 If $w \in \beta(t)$, this follows from Lemma \ref{la:spring-embedding-properties}\ref{item:spring-embedding-item-1}.
 So assume that $w \notin \beta(t)$ and let $Z$ be the vertex set of the connected component of $G - \beta(t)$ that contains $w$.
 Then $|N_G(Z)| \leq 2$.
 If $|N_G(Z)| = 2$, we conclude that $\chi(v) \neq \chi(w)$ by Lemma \ref{la:spring-embedding-properties}\ref{item:spring-embedding-item-2}.
 If $|N_G(Z)| = 0$, then $\chi(v) \neq \chi(w)$ holds because $\chi$ is $1$-stable, $v_1$ is individualized, and $v$ is reachable from $v_1$, but $w$ is not reachable from $v_1$.
 Finally, suppose that $|N_G(Z)| = 1$ and let $w' \in \beta(t)$ be the unique vertex contained in $N_G(Z)$.
 If $w' \neq v$, then $\chi(v) \neq \chi(w)$ by Lemma \ref{la:spring-embedding-properties}, Parts \ref{item:spring-embedding-item-1} and \ref{item:spring-embedding-item-3}.
 So consider the case that $v = w'$.
 Then every path from $v_1$ to $w$ passes through $v$.
 In particular, $\dist_G(v,v_1) < \dist_G(w,v_1)$.
 Since $|[v_1]_{\chi}| = 1$ and $\chi$ is $1$-stable, we conclude that $\chi(v) \neq \chi(w)$.
\end{proof}

\subsection{Obtaining a Connected Subgraph}

Now, towards understanding graphs of Type IIa, we argue under which conditions we can identify connections between different components of $G[c]$ using a single second edge color $d$.
More precisely, let $c \in C_E(G,\WL{2}{G})$ be an edge color of Type IIa, i.e., $G[c]$ is a disjoint union of $\ell \geq 2$ cycles.
As a first step, we aim at identifying a second edge color $d \in C_E(G,\WL{2}{G})$ such that edges of color $d$ connect different components of $G[c]$ (thus, bringing us one step closer to finding a connected subgraph).
This is essentially achieved by the lemma below which either finds such a color $d$ or concludes that one of the other desirable options discussed above is satisfied.

\begin{definition}
 Let $C \subseteq C_E(G,\chi)$ be a set of edge colors.
 For $v,w \in V(G)$, we define the \emph{$C$-distance} between $v$ and $w$, denoted $\dist^C(v,w)$, to be the length of the shortest path
 $v = u_0,u_1,\dots,u_\ell=w$ such that $\chi(u_{i-1},u_i) \in C$ for all $i \in [\ell]$.
\end{definition}

Let $c \in C_E(G,\chi)$ be of Type IIa, and let $d \in C_E(G,\chi)$ be a second edge color.
We say that $c$ \emph{admits short $d$-connections in $G$} if there are vertex sets $A_1,A_2$ of distinct connected components of $G[c]$ and vertices $v_1 \in A_1$ and $v_2 \in A_2$ such that $\dist^D(v_1,v_2) \leq 2$ where $D \coloneqq \{d,d^{-1}\}$ and $d^{-1}$ denotes the ``reverse color'' of $d$, i.e., if $\chi(v,w) = d$ then $d^{-1} \coloneqq \chi(w,v)$.

\begin{lemma}
 \label{la:cycle-connections}
 Let $G$ be a $3$-connected planar graph and let $\chi \coloneqq \WL{2}{G}$ be the coloring computed by $2$-WL.
 Suppose that $G$ has Type IIa.
 Then one of the following holds.
 \begin{enumerate}
  \item there are colors $c,d \in C_E(G,\chi)$ such that $c$ has Type IIa and admits short $d$-connections in $G$,
  \item there is a color $c \in C_E(G,\chi)$ that defines a matching, or
  \item there is a vertex $v \in V(G)$ such that $\Disc_G(v) = V(G)$.
 \end{enumerate}
\end{lemma}

\begin{proof}
 Consider an edge color $c \in C_E(G,\chi)$ of Type IIa.
 Let $A_1,\dots,A_\ell$ be the vertex sets of the connected components of $G[c]$.
 Observe that $\ell \geq 2$.
 Also fix a plane embedding of $G$.
 For every $i \in [\ell]$, the cycle $(G[c])[A_i]$ splits the plane into two regions, the interior region and the exterior region.
 We denote by $\exte_c(A_i)$ the set of vertices located in the exterior region, and $\inte_c(A_i)$ denotes the set of vertices located in the interior region.
 By Lemma \ref{la:path-between-color-components}, we may assume without loss of generality that
 \begin{equation}
  \label{eq:comps-ext}
  A_j \subseteq \exte_c(A_i)
 \end{equation}
 for all distinct $i,j \in[\ell]$.
 
 Let $D \coloneqq \{\chi(v,v) \mid v \in V(G[c])\}$.
 Recall that a path $u_0,\dots,u_m$ \emph{avoids $D$} if $\chi(u_i,u_i) \notin D$ for all $i \in [m-1]$, i.e., the color of no internal vertices of the path is contained in the set $D$.
 For a vertex $u \in V(G)$, we define
 \[\Theta_c(u) \coloneqq \{i \in [\ell] \mid \exists w \in A_i \colon \text{there is a path from $u$ to $w$ that avoids $D$}\}.\]
 Moreover, we define
 \[M(c) \coloneqq \{u \in V(G) \setminus V(G[c]) \mid |\Theta_c(u)| \geq 2\}.\]
 Since $2$-WL detects whether there is a path from $u$ to $w$ that avoids $D$ (see Observation \ref{obs:wl-knows-paths-avoiding-colors}), we conclude that $M(c)$ is $\chi$-invariant.
 Now, we pick an edge color $c \in C_E(G,\chi)$ of Type IIa for which $|M(c)|$ is minimal.
 Let $M \coloneqq M(c)$.
 A visualization is given in Figure \ref{fig:type-ii-a-short-connections}.
 
 \begin{claim}
  $M \subseteq \exte_c(A_i)$ for all $i \in [\ell]$.
 \end{claim}
 \begin{claimproof}
  This follows directly from the fact that $A_j \subseteq \exte_c(A_i)$ for all distinct $i,j \in [\ell]$.
  Indeed, suppose there are $u \in M$ and $i \in [\ell]$ such that $u \in \inte_c(A_i)$.
  By definition, $|\Theta_c(u)| \geq 2$ and there is some $i \neq j \in [\ell]$ and a vertex $w \in A_j$ such that $w$ is reachable from $u$ via a path that avoids $D$.
  But $w \in \exte_c(A_i)$ by \eqref{eq:comps-ext}.
  Hence, any path from $u$ to $w$ must visit some vertex $v \in A_i$ because of the plane embedding of $G$. This is a contradiction since $\chi(v,v) \in D$.
 \end{claimproof}
 
 Now, suppose there is no edge color $d \in C_E$ such that $c$ admits short $d$-connections (otherwise the first option of the lemma is satisfied).
 Recall that for an edge color $d \in C_E$ and a vertex $v \in V(G)$, we use the notation $N_d^+(v) \coloneqq \{w \in V(G) \mid \chi(v,w) = d\}$.
 
 \begin{claim}
  \label{claim:disjoint-neighborhoods-to-the-middle}
  For every $d \in C_E$ and every $v \neq w \in V(G[c])$, it holds that
  \[N_d^+(v) \cap N_d^+(w) \cap M = \emptyset.\]
 \end{claim}
 \begin{claimproof}
  Suppose towards a contradiction that there are an edge color $d \in C_E$ and distinct vertices $v,w \in V(G[c])$ such that there is some vertex $u \in N_d^+(v) \cap N_d^+(w) \cap M$.
  Let $W \coloneqq N_d^-(u)$.
  Observe that $v,w \in W$.
  Since $c$ does not admit short $d$-connections, there is some $i \in [\ell]$ such that $W \subseteq A_i$.
  All endpoints of $d$-colored edges in the graph $G[c]$ have the same vertex color $c_V$.
  We define $v_1,\dots,v_q$ to be the list of vertices in $A_i$ which have vertex color $c_V$, numbered according to the cyclic order along the cycle $(G[c])[A_i]$.
  
  First suppose $|W| \geq 3$.
  Then $G[d]$ contains a vertex of degree at least $3$, which means that $d$ does not have Type II.
  Since $G$ has Type IIa, we conclude that $d$ has Type~I.
  Let $U \coloneqq [u]_\chi$ denote the vertex color class of $u$.
  Also let $U_i \coloneqq \{u' \in U \mid N_d^-(u') \subseteq A_i\}$.
  The collection $\{N_d^-(u') \mid u' \in U_i\}$ forms a partition of $v_1,\dots,v_q$.
  Since $G$ is planar and $U_i \subseteq M \subseteq \exte_c(A_i)$, it follows that each set $N_d^-(u')$ forms an interval with respect to the cyclic order on $v_1,\dots,v_q$.
  But this is only possible if $U_i = \{u\}$ and $N_d^-(u) = \{v_1,\dots,v_q\}$ because $\chi$ is $2$-stable.
  Now, by Lemma \ref{la:path-between-color-components}, all sets $A_j$ for $i \neq j \in [\ell]$ are located in the same face of $(G[c,d])[A_i \cup U_i]$.
  Using the fact that $G$ is 3-connected, this provides vertices $x_1,x_2 \in A_i$ such that

%%%%%%%

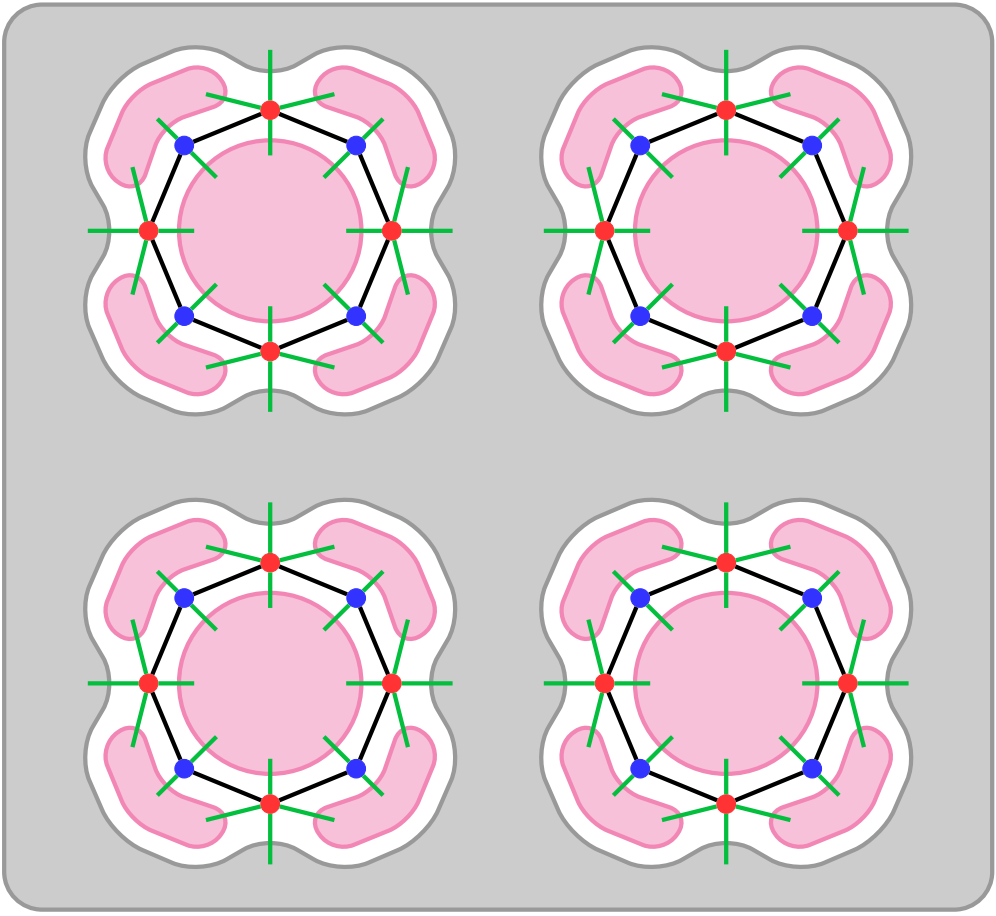
\begin{figure}
  \centering
  \begin{tikzpicture}
   
   \draw[rounded corners = 0.5cm, line width=1.6pt, gray!80, fill=gray!40] (-3.5,-3) rectangle (9.5,9);
   
   \foreach \p/\q in {0/0,6/0,0/6,6/6}{
    
    \draw[rounded corners = 0.4cm, line width=1.6pt, gray!80, fill=white] ($(\p,\q) + (0:2)$) -- ($(\p,\q) + (20:2.7)$) -- ($(\p,\q) + (45:2.9)$) -- ($(\p,\q) + (70:2.7)$) -- ($(\p,\q) + (90:2)$) -- ($(\p,\q) + (110:2.7)$) -- ($(\p,\q) + (135:2.9)$) -- ($(\p,\q) + (160:2.7)$) -- ($(\p,\q) + (180:2)$) -- ($(\p,\q) + (200:2.7)$) -- ($(\p,\q) + (225:2.9)$) -- ($(\p,\q) + (250:2.7)$) -- ($(\p,\q) + (270:2)$) -- ($(\p,\q) + (290:2.7)$) -- ($(\p,\q) + (315:2.9)$) -- ($(\p,\q) + (340:2.7)$) -- cycle;
    \draw[line width=1.6pt, magenta!60, fill=magenta!30] (\p,\q) circle (1.2cm);
    
    \foreach \i in {0,1,2,3}{
      \draw[rounded corners = 0.3cm, line width=1.6pt, magenta!60, fill=magenta!30] ($(\p,\q) + (15 + 90*\i:1.8)$) -- ($(\p,\q) + (45 + 90*\i:2)$) -- ($(\p,\q) + (75 + 90*\i:1.8)$) -- ($(\p,\q) + (70 + 90*\i:2.4)$) -- ($(\p,\q) + (45 + 90*\i:2.6)$) -- ($(\p,\q) + (20 + 90*\i:2.4)$) -- cycle;
    }
    
    \foreach \i in {0,1,2,3}{
     \node[vertex,red!80] (u\i-\p-\q) at ($(\p,\q) + (90*\i:1.6)$) {};
     \node[vertex,blue!80] (v\i-\p-\q) at ($(\p,\q) + (45 + 90*\i:1.6)$) {};
     
     \draw[line width=1.6pt, darkpastelgreen] (u\i-\p-\q) edge ($(\p,\q) + (90*\i:1.0)$);
     \draw[line width=1.6pt, darkpastelgreen] (v\i-\p-\q) edge ($(\p,\q) + (45 + 90*\i:1.0)$);
     \draw[line width=1.6pt, darkpastelgreen] (u\i-\p-\q) edge ($(\p,\q) + (-25 + 90*\i:2.0)$);
     \draw[line width=1.6pt, darkpastelgreen] (u\i-\p-\q) edge ($(\p,\q) + (90*\i:2.4)$);
     \draw[line width=1.6pt, darkpastelgreen] (u\i-\p-\q) edge ($(\p,\q) + (25 + 90*\i:2.0)$);
     \draw[line width=1.6pt, darkpastelgreen] (v\i-\p-\q) edge ($(\p,\q) + (45 + 90*\i:2.1)$);
    }
    \foreach \i/\j in {0/0,0/1,1/1,1/2,2/2,2/3,3/3,3/0}{
     \draw[line width=1.6pt] (v\i-\p-\q) edge (u\j-\p-\q);
    }
    
   }
  \end{tikzpicture}
  \caption{Visualization of planar graphs of Type IIa. The vertices of the graph $G[c]$ are colored {\sf red} and {\sf blue} ($G[c]$ is bicolored in the visualization) and edges of the graph $G[c]$ are colored {\sf black}.
   All other edges indicated in the visualization are colored green.
   The set $M$ is marked gray and the set $V(G) \setminus (V(G[c]) \cup M)$ is contained in the magenta regions.}
  \label{fig:type-ii-a-short-connections}
 \end{figure}

 %%%%%%%

  \begin{enumerate}
   \item $\chi(x_1,x_1) = \chi(x_2,x_2)$,
   \item there is some $i \neq j \in [\ell]$ and a vertex $z \in A_j \cup U_j$ such that $x_1$ is reachable from $z$ via a path that avoids $\{\chi(v,v) \mid v \in V(G[c,d])\}$, and
   \item there is no vertex $z \in A_j \cup U_j$, $i \neq j \in [\ell]$, such that $x_2$ is reachable from $z$ via a path that avoids $\{\chi(v,v) \mid v \in V(G[c,d])\}$.
  \end{enumerate}
  This contradicts $\chi$ being $2$-stable (see Observation \ref{obs:wl-knows-paths-avoiding-colors}).
  
  So $|W| = 2$.
  Using similar arguments as above, $N_d^-(u')$ again forms an interval with respect to the cyclic order on $v_1,\dots,v_q$ for all $u' \in U_i$.
  In this case, we conclude that either $d$ has Type II, or $(G[c])[A_i]$ is a bicolored cycle of length $4$ and $W = \{u\}$.
  In the latter case, we can use the same arguments as above (exploiting $3$-connectedness) to conclude that this case can not occur.
  
  So suppose $d$ has Type II.
  Observe that $d$ has Type IIa since $G$ has Type IIa.
  Also recall that $M \subseteq \exte_c(A_i)$ for all $i \in [\ell]$.
  Together, this means that $M(d) \subseteq M \setminus V(G[d])$ contradicting the fact that $|M| = |M(c)|$ is minimal among all edge colors of Type IIa.
 \end{claimproof}

 Now, fix some $i \in [\ell]$ and a vertex $v \in A_i$ such that $N_G(v) \cap M \neq \emptyset$.
 If there is an edge color $d \in C_E$ such that $|N_d(v) \cap M| = 1$, then $d$ defines a matching by the previous claim and since $M$ is $\chi$-invariant.
 So suppose that no such edge color exists.
 We argue that there is some $w \in N_G(v) \cap M$ such that $\Disc_G(w) = V(G)$, which completes the proof.
 
 Towards this end, consider the following auxiliary graph $H$ defined via
 \[V(H) \coloneqq V(G[c]) \cup M \cup \{z_1,\dots,z_\ell\}\]
 and
 \[E(H) \coloneqq E_G(M,V(G[c]) \cup M) \cup E(G[c]) \cup \{z_iv \mid i \in [\ell], v \in A_i\}.\]
 We equip $H$ with a vertex coloring $\lambda\coloneqq V(H) \rightarrow \{1,2,3\}$ defined via
 \[\lambda(u) \coloneqq \begin{cases}
                         1 &\text{if } u \in M,\\
                         2 &\text{if } u \in V(G[c]),\\
                         3 &\text{if } u \in \{z_1,\dots,z_\ell.\}
                        \end{cases}.\]
 
 \begin{claim}
  $H$ is planar and $3$-connected.
 \end{claim}
 \begin{claimproof}
  It is easy to see that $H$ is planar since $M \subseteq \exte_c(A_i)$ for all $i \in [\ell]$, so the vertex $z_i$ can simply be placed in the interior region of $(G[c])[A_i]$.
  
  So it remains to show that $H$ is $3$-connected.
  Let $v,w \in V(H) \setminus \{z_1,\dots,z_\ell\}$ be two distinct vertices.
  We argue that there are three internally vertex-disjoint paths from $v$ to $w$ in $H$.
  Since $G$ is $3$-connected, there are internally vertex-disjoint paths $P_1,P_2,P_3$ from $v$ to $w$ in $G$ by Menger's Theorem (see, e.g., \cite[Chapter 3.3]{Diestel18}).
  We translate those three paths into internally vertex-disjoint paths $Q_1,Q_2,Q_3$ from $v$ to $w$ in $H$.
  
  For $i \in [\ell]$, let $B_i \coloneqq \{u \in V(G) \setminus V(G[c]) \mid \Theta_c(u) = \{i\}\}$.
  Observe that $V(G) = (V(H)\setminus \{z_1,\dots,z_\ell\}) \cup \bigcup_{i \in [\ell]}B_i$.
  
  Consider a path $P_r$, $r \in \{1,2,3\}$.
  If $P_r$ never visits a vertex from $A_i \cup B_i$, for all $i \in [\ell]$, then the path $P_r$ also exists in $H$, and we can simply set $Q_r \coloneqq P_r$.
  So there is some $i \in [\ell]$ such that $P_r$ visits some vertex from $A_i \cup B_i$.
  Let $u$ be the first occurrence of a vertex from $A_i$ on $P_r$, and $u'$ the last occurrence of a vertex from $A_i$ on $P_r$.
  Observe that all occurrences of vertices from $A_i \cup B_i$ are on the subpath from $u$ to $u'$.
  Now, we simply replace the subpath from $u$ to $u'$ by a path from $u$ to $u'$ in the graph $H[A_i \cup \{z_i\}]$.
  Since $H[A_i \cup \{z_i\}]$ is $3$-connected, this operation can be performed for all three paths $P_1,P_2,P_3$ in parallel.
  
  Repeating this procedure for all $i \in [\ell]$, we obtain the desired paths $Q_1,Q_2,Q_3$.
  Overall, this implies that $H$ is $3$-connected.
 \end{claimproof}

 Recall that we fixed some $i \in [\ell]$ and a vertex $v \in A_i$ such that $N_G(v) \cap M \neq \emptyset$.
 Let $z_i,u_1,\dots,u_\delta$ be the neighbors of $v$ in the graph $H$, listed in clockwise orientation starting at $z_i$ for some fixed embedding of $H$.
 We have that $u_1,u_\delta \in A_i$, and $u_2,\dots,u_{\delta-1} \in M$.
 By the comments above, we have that $\delta \geq 4$, i.e., $v$ has at least two neighbors in the set $M$.
 Now, we define another auxiliary graph $\widehat{H}$ which is obtained from $H$ by contradicting $u_3,\dots,u_{\delta-1}$ to a single vertex $u^*$, and adding an edge between $u_2$ and $u^*$.
 Clearly, $\widehat{H}$ is still a planar graph.
 Let $\widehat{\lambda}$ be the corresponding vertex coloring obtained from $\lambda$ ($u_3,\dots,u_{\delta-1}$ receive the same color under $\lambda$ which also becomes the color of $u^*$).
 Let $(T,\beta)$ denote the decomposition into triconnected components of $\widehat{H}$.
 
 \begin{claim}
  \label{claim:3-connected-part-in-contracted-graph}
  There is some node $t \in V(T)$ such that $A_i \cup \{z_i,u_2,u^*\} \subseteq \beta(t)$.
 \end{claim}
 \begin{claimproof}
  Since $\widehat{H}[A_i \cup \{z_i\}]$ is $3$-connected, there is a unique node $t \in V(T)$ such that $A_i \cup \{z_i\} \subseteq \beta(t)$.
  We argue that $u_2,u^* \in \beta(t)$.
  
  First, consider the vertex $u^*$.
  Since $H$ is $3$-connected, there are three internally vertex-disjoint paths $P_1,P_2,P_3$ from $u_3$ to $z_i$ in the graph $H$.
  These paths can easily be lifted to paths in $\widehat{H}$.
  Indeed, for $r \in \{1,2,3\}$, let $u$ be the last appearance of a vertex from $u_3,\dots,u_{\delta-1}$ on the path $P_r$.
  We define $Q_r$ to be the subpath of $P_r$ that starts in $u$ (and ends in $z_i$) where we replace $u$ by $u^*$.
  It is easy to see that $Q_1,Q_2,Q_3$ give three internally vertex-disjoint paths from $u^*$ to $z_i$ in the graph $\widehat{H}$.
  This implies that $u^* \in \beta(t)$.
  
  For the vertex $u_2$, the argument is slightly more complicated.
  We have that $v,u_2,u^*$ forms a triangle in the graph $\widehat{H}$.
  Hence, there is some node $t' \in V(T)$ such that $v,u_2,u^* \in \beta(t')$.
  If $t=t'$, we are done.
  So suppose $t \neq t'$.
  Then $\{v,u^*\}$ separates $u_2$ from $z_i$.
  Let $W$ denote the vertex set of the connected component of the graph $G - \{v,u^*\}$ that contains $u_2$.
  Observe that $N_{\widehat{H}}(W) = \{v,u^*\}$.
  We have that $W \subseteq V(H)$.
  Let $j \in \{3,\dots,\delta-1\}$ be the highest index such that $u_j \in N_H(W)$.
  But now, by planarity, $\{v,u_j\}$ separates $W$ from $A_i \setminus \{v\}$ in the graph $H$.
  This contradicts $H$ being $3$-connected.
 \end{claimproof}

 Now, consider the coloring $\widehat{\kappa} \coloneqq \WL{1}{\widehat{H},\widehat{\lambda},v,u_2}$, i.e., the stable coloring computed by $1$-WL after individualizing $v$ and $u_2$ in the graph $\widehat{H}$.
 It is easy to see that $|[u^*]_{\widehat{\kappa}}| = 1$ since it is the only neighbor of $v$ of color $1$ (apart from $u_2$, which is individualized).
 Also, $v,u_2,u^*$ forms a triangular face in $G[[\beta(t)]]$ (where $t \in V(T)$ is the node from Claim \ref{claim:3-connected-part-in-contracted-graph}) since a planar embedding of $\widehat{H}$ is inherited from the embedding of $H$.
 Hence, by Lemma \ref{la:coloring-from-spring-embedding} and Claim \ref{claim:3-connected-part-in-contracted-graph}, we conclude that $|[x]_{\widehat{\kappa}}| = 1$ for all $x \in A_i \cup \{z_i,u_2,u^*\}$.
 
 Next, consider the coloring $\kappa \coloneqq \WL{1}{H,\lambda,v,u_2}$.
 Clearly, the set $\{u_3,\dots,u_{\delta-1}\}$ is $\kappa$-invariant.
 This implies that any refinement made in $(\widehat{H},\widehat{\lambda})$ can also be made in $(H,\lambda)$.
 In particular, $|[x]_{\kappa}| = 1$ for all $x \in A_i \cup \{z_i,u_2\}$.
 Since $H$ is $3$-connected, we conclude that $\kappa$ is discrete by Lemma \ref{la:fixing-cycle-and-extra-vertex}.
 
 To complete the proof, we now lift the analysis of colorings one more time from $H$ to~$G$.
 Consider the set $\Disc_G(u_2)$.
 From Claim \ref{claim:disjoint-neighborhoods-to-the-middle}, we can conclude that $v \in \Disc_G(u_2)$.
 Now, after individualizing $v$ and $u_2$, any refinement made by $1$-WL in the graph $H$ can also be made in $G$.
 This is trivial for all vertices except for the $z_i$.
 However, when $z_i$ and $z_j$ receive different colors, eventually the color sets on $A_i$ and $A_j$ will be disjoint in $G$.
 This allows us to ``simulate'' the refinement process on the $z_i$-vertices.
 
 So we get that $V(G[c]) \cup M \subseteq \Disc_G(u_2)$.
 It follows that $\Disc_G(u_2) = V(G)$ by Lemma \ref{la:fixing-cycle-and-extra-vertex}.
\end{proof}

\subsection{The Classification Theorem}

\begin{figure}
 \begin{center}
 \begin{subfigure}[b]{.27\linewidth}
  \begin{tikzpicture}[scale = 0.8]
   \node[vertex,red!80] (1) at ($(0,0)+(45:0.8)$) {};
   \node[vertex,blue!80] (2) at ($(0,0)+(135:0.8)$) {};
   \node[vertex,red!80] (3) at ($(0,0)+(225:0.8)$) {};
   \node[vertex,blue!80] (4) at ($(0,0)+(315:0.8)$) {};
   
   \node[vertex,blue!80] (5) at ($(0,0)+(45:2.4)$) {};
   \node[vertex,red!80] (6) at ($(0,0)+(135:2.4)$) {};
   \node[vertex,blue!80] (7) at ($(0,0)+(225:2.4)$) {};
   \node[vertex,red!80] (8) at ($(0,0)+(315:2.4)$) {};
   
   \draw[line width=1.6pt] (1) edge (2);
   \draw[line width=1.6pt] (1) edge (4);
   \draw[line width=1.6pt] (2) edge (3);
   \draw[line width=1.6pt] (3) edge (4);
   
   \draw[line width=1.6pt] (5) edge (6);
   \draw[line width=1.6pt] (5) edge (8);
   \draw[line width=1.6pt] (6) edge (7);
   \draw[line width=1.6pt] (7) edge (8);
   
   \draw[line width=1.6pt] (1) edge (5);
   \draw[line width=1.6pt] (2) edge (6);
   \draw[line width=1.6pt] (3) edge (7);
   \draw[line width=1.6pt] (4) edge (8);
   
  \end{tikzpicture}
  \caption{Bicolored cube}
  \label{fig:graph-bi-cube}
 \end{subfigure}
 \hspace{3cm}
 \begin{subfigure}[b]{.27\linewidth}
  \begin{tikzpicture}[scale=0.8]
   
   \node[vertex,red!80] (1a) at ($(0,0)+(45:0.8)+(45:0.2)$) {};
   \node[vertex,red!80] (1b) at ($(0,0)+(45:0.8)+(180:0.4)$) {};
   \node[vertex,red!80] (1c) at ($(0,0)+(45:0.8)+(270:0.4)$) {};
   \node[vertex,blue!80] (2) at ($(0,0)+(135:0.8)$) {};
   \node[vertex,red!80] (3a) at ($(0,0)+(225:0.8)+(0:0.4)$) {};
   \node[vertex,red!80] (3b) at ($(0,0)+(225:0.8)+(90:0.4)$) {};
   \node[vertex,red!80] (3c) at ($(0,0)+(225:0.8)+(225:0.2)$) {};
   \node[vertex,blue!80] (4) at ($(0,0)+(315:0.8)$) {};
   
   \node[vertex,blue!80] (5) at ($(0,0)+(45:2.4)$) {};
   \node[vertex,red!80] (6a) at ($(0,0)+(135:2.4)+(0:0.5)$) {};
   \node[vertex,red!80] (6b) at ($(0,0)+(135:2.4)+(270:0.5)$) {};
   \node[vertex,red!80] (6c) at ($(0,0)+(135:2.4)+(315:0.7)$) {};
   \node[vertex,blue!80] (7) at ($(0,0)+(225:2.4)$) {};
   \node[vertex,red!80] (8a) at ($(0,0)+(315:2.4)+(90:0.5)$) {};
   \node[vertex,red!80] (8b) at ($(0,0)+(315:2.4)+(135:0.7)$) {};
   \node[vertex,red!80] (8c) at ($(0,0)+(315:2.4)+(180:0.5)$) {};
   
   \draw[line width=1.6pt] (1b) edge (2);
   \draw[line width=1.6pt] (1c) edge (4);
   \draw[line width=1.6pt] (2) edge (3b);
   \draw[line width=1.6pt] (3a) edge (4);
   
   \draw[line width=1.6pt] (5) edge (6a);
   \draw[line width=1.6pt] (5) edge (8a);
   \draw[line width=1.6pt] (6b) edge (7);
   \draw[line width=1.6pt] (7) edge (8c);
   
   \draw[line width=1.6pt] (1a) edge (5);
   \draw[line width=1.6pt] (2) edge (6c);
   \draw[line width=1.6pt] (3c) edge (7);
   \draw[line width=1.6pt] (4) edge (8b);
   
   \foreach \i in {1,3,6,8}{
    \foreach \v/\w in {a/b,a/c,b/c}{
     \draw[line width=1.6pt,darkpastelgreen] (\i\v) edge (\i\w);
    }
   }
  \end{tikzpicture}
  \caption{Chamfered tetrahedron}
  \label{fig:graph-chamfered-tetrahedron}
 \end{subfigure}
 \end{center}
 \caption{A chamfered tetrahedron is obtained from a bicolored cube by truncating all {\sf red} vertices.}
 \label{fig:chamfered}
\end{figure}
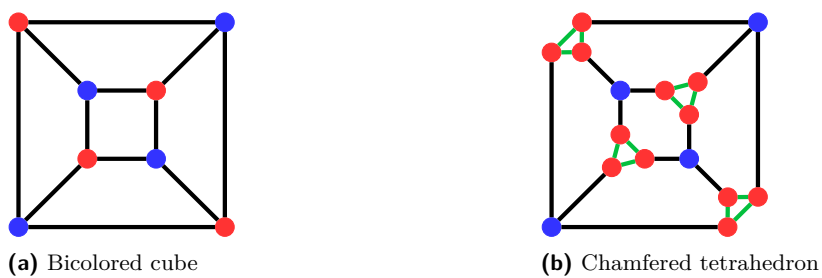

We now provide the main classification result of this section.
We start by defining the graphs that appear in it. 
Let $G$ be a $3$-connected planar graph and let $W \subseteq V(G)$ be a set of vertices.
We define $H$ to be the graph obtained from $G$ as follows. Let $w$ be a vertex in $W$. First, subdivide each edge incident to $w$ once. This gives $\ell \coloneqq \deg_G(w)$ new vertices, which we call $w_1, \dots, w_\ell$ according to the unique cyclic order in any embedding of $G$. We then remove all edges $ww_i$ and insert edges $w_iw_{(i+1) \bmod \ell}$ for $i \in [\ell]$, turning $w_1, \dots, w_\ell$ into a cycle. Each $w_i$ inherits the color of $w$. We call these steps the \emph{truncation} of $w$. One by one, we then truncate each vertex in $W$. (Note that their order does not matter for the final result.) The obtained graph is $H$; see Figure \ref{fig:chamfered} for an example.
Now, we can define the following graphs.
\begin{itemize}
 \item For every $3$-connected planar graph $G$, the \emph{truncated $G$} is obtained from $G$ by truncating all vertices.
 \item A \emph{chamfered tetrahedron} (Figure \ref{fig:graph-chamfered-tetrahedron}) is obtained from a bicolored cube (Figure \ref{fig:graph-bi-cube}) by truncating all {\sf red} vertices\footnote{The name \emph{chamfered tetrahedron} comes from an alternative construction that obtains a chamfered tetrahedron by truncation of all edges of a tetrahedron.}.
 \item A \emph{chamfered cube} is obtained from a rhombic dodecahedron (Figure \ref{fig:graph-rhombic-dodecahedron}) by truncating all {\sf blue} vertices.
 \item A \emph{chamfered octahedron} is obtained from a rhombic dodecahedron (Figure \ref{fig:graph-rhombic-dodecahedron}) by truncating all {\sf red} vertices.
 \item A \emph{chamfered dodecahedron} is obtained from a rhombic triacontahedron (Figure \ref{fig:graph-rhombic-triacontahedron}) by truncating all {\sf blue} vertices.
 \item A \emph{chamfered icosahedron} is obtained from a rhombic triacontahedron (Figure \ref{fig:graph-rhombic-triacontahedron}) by truncating all {\sf red} vertices.
\end{itemize}
Next, let $G$ be a planar graph. We define the \emph{$C_4$-subdivision of $G$} to be the graph obtained from $G$ by replacing each edge $vw \in E(G)$ with four vertices $(vw,1)$, $(vw,2)$, $(vw,3)$, $(vw,4)$ and edges $(vw,1)(vw,2)$, $(vw,2)(vw,3)$, $(vw,3)(vw,4)$, $(vw,4)(vw,1)$ and $v(vw,1)$, $w(vw,3)$.

Also, for $m \geq 2$, we define the graph $C_m^*$ with vertex set $V(C_m^*) \coloneqq [m] \times [4]$ and edge set
\begin{align*}
 E(C_m^*) {}\coloneqq{} \phantom{{}\cup{}} &\{(i,1)(i,2), (i,2)(i,3), (i,3)(i,4), (i,4)(i,1) \mid i \in [m]\}\\
                    {}\cup{} &\{(i,3)(i+1,1) \mid i \in [m-1]\} \cup \{(m,3)(1,1)\}.
\end{align*}
Finally, for $h \geq 3$, we define $K_{2,h}^*$ to be the graph with vertex set $V(K_{2,h}^*) \coloneqq \{u_1,u_2\} \uplus ([h] \times [4])$ and edge set
\begin{align*}
 E(K_{2,h}^*) {}\coloneqq{} \phantom{{}\cup{}} &\{(i,1)(i,2), (i,2)(i,3), (i,3)(i,4), (i,4)(i,1) \mid i \in [h]\}\\
                       {}\cup{} &\{u_1(i,1) \mid i \in [h]\} {}\cup{} \{u_2(i,3) \mid i \in [h]\}.
\end{align*}
Examples for the last three constructions can be found in Figure \ref{fig:c4-extensions}. 

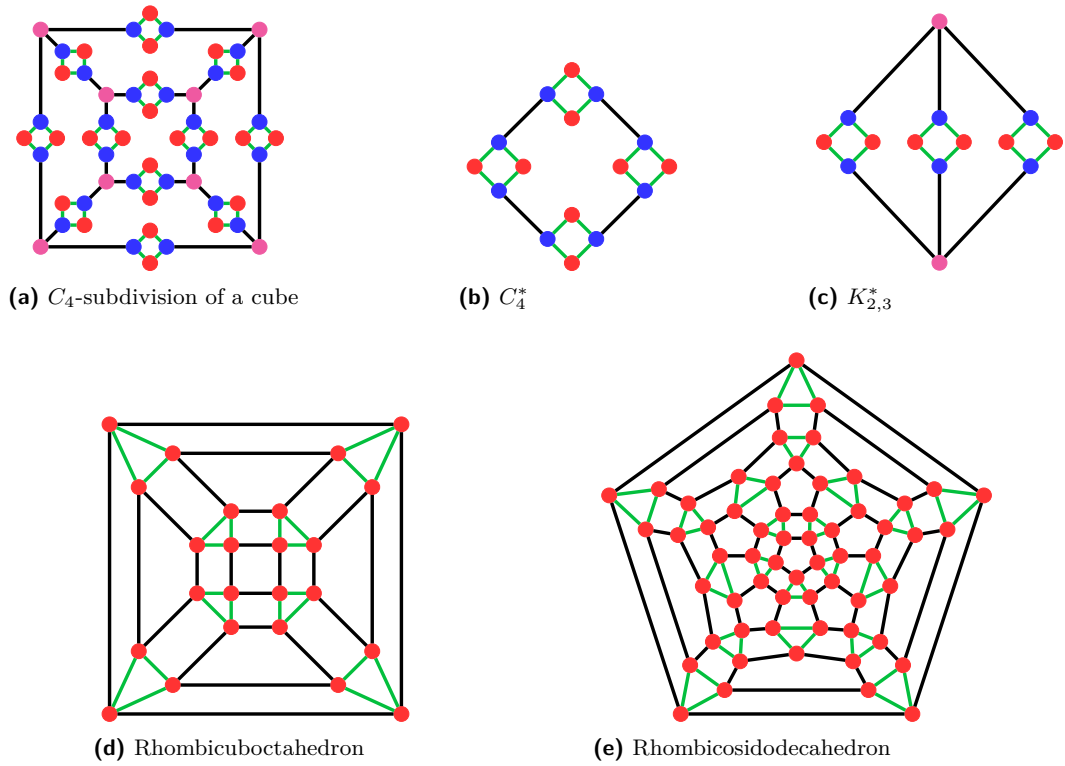
\begin{figure}
 \centering
 \begin{subfigure}[b]{.3\linewidth}
  \scalebox{\figscalelarge}{
  \begin{tikzpicture}[scale=0.9]
   \node[vertex,magenta!80] (1) at (0.8,0.8) {};
   \node[vertex,magenta!80] (2) at (-0.8,0.8) {};
   \node[vertex,magenta!80] (3) at (-0.8,-0.8) {};
   \node[vertex,magenta!80] (4) at (0.8,-0.8) {};
   
   \node[vertex,magenta!80] (5) at (2.0,2.0) {};
   \node[vertex,magenta!80] (6) at (-2.0,2.0) {};
   \node[vertex,magenta!80] (7) at (-2.0,-2.0) {};
   \node[vertex,magenta!80] (8) at (2.0,-2.0) {};
   
   \node[vertex,red!80] (9-1) at ($(0,0.8)+(0.0,0.3)$) {};
   \node[vertex,blue!80] (9-2) at ($(0,0.8)+(0.3,0.0)$) {};
   \node[vertex,red!80] (9-3) at ($(0,0.8)+(0.0,-0.3)$) {};
   \node[vertex,blue!80] (9-4) at ($(0,0.8)+(-0.3,0.0)$) {};
   
   \node[vertex,red!80] (10-1) at ($(0,-0.8)+(0.0,0.3)$) {};
   \node[vertex,blue!80] (10-2) at ($(0,-0.8)+(0.3,0.0)$) {};
   \node[vertex,red!80] (10-3) at ($(0,-0.8)+(0.0,-0.3)$) {};
   \node[vertex,blue!80] (10-4) at ($(0,-0.8)+(-0.3,0.0)$) {};
   
   \node[vertex,blue!80] (11-1) at ($(0.8,0)+(0.0,0.3)$) {};
   \node[vertex,red!80] (11-2) at ($(0.8,0)+(0.3,0.0)$) {};
   \node[vertex,blue!80] (11-3) at ($(0.8,0)+(0.0,-0.3)$) {};
   \node[vertex,red!80] (11-4) at ($(0.8,0)+(-0.3,0.0)$) {};
   
   \node[vertex,blue!80] (12-1) at ($(-0.8,0)+(0.0,0.3)$) {};
   \node[vertex,red!80] (12-2) at ($(-0.8,0)+(0.3,0.0)$) {};
   \node[vertex,blue!80] (12-3) at ($(-0.8,0)+(0.0,-0.3)$) {};
   \node[vertex,red!80] (12-4) at ($(-0.8,0)+(-0.3,0.0)$) {};
   
   \node[vertex,red!80] (13-1) at ($(0,2.0)+(0.0,0.3)$) {};
   \node[vertex,blue!80] (13-2) at ($(0,2.0)+(0.3,0.0)$) {};
   \node[vertex,red!80] (13-3) at ($(0,2.0)+(0.0,-0.3)$) {};
   \node[vertex,blue!80] (13-4) at ($(0,2.0)+(-0.3,0.0)$) {};
   
   \node[vertex,red!80] (14-1) at ($(0,-2.0)+(0.0,0.3)$) {};
   \node[vertex,blue!80] (14-2) at ($(0,-2.0)+(0.3,0.0)$) {};
   \node[vertex,red!80] (14-3) at ($(0,-2.0)+(0.0,-0.3)$) {};
   \node[vertex,blue!80] (14-4) at ($(0,-2.0)+(-0.3,0.0)$) {};
   
   \node[vertex,blue!80] (15-1) at ($(2.0,0)+(0.0,0.3)$) {};
   \node[vertex,red!80] (15-2) at ($(2.0,0)+(0.3,0.0)$) {};
   \node[vertex,blue!80] (15-3) at ($(2.0,0)+(0.0,-0.3)$) {};
   \node[vertex,red!80] (15-4) at ($(2.0,0)+(-0.3,0.0)$) {};
   
   \node[vertex,blue!80] (16-1) at ($(-2.0,0)+(0.0,0.3)$) {};
   \node[vertex,red!80] (16-2) at ($(-2.0,0)+(0.3,0.0)$) {};
   \node[vertex,blue!80] (16-3) at ($(-2.0,0)+(0.0,-0.3)$) {};
   \node[vertex,red!80] (16-4) at ($(-2.0,0)+(-0.3,0.0)$) {};
   
   \node[vertex,blue!80] (17-1) at ($(1.4,1.4)+(0.2,0.2)$) {};
   \node[vertex,red!80] (17-2) at ($(1.4,1.4)+(-0.2,0.2)$) {};
   \node[vertex,blue!80] (17-3) at ($(1.4,1.4)+(-0.2,-0.2)$) {};
   \node[vertex,red!80] (17-4) at ($(1.4,1.4)+(0.2,-0.2)$) {};
   
   \node[vertex,red!80] (18-1) at ($(-1.4,1.4)+(0.2,0.2)$) {};
   \node[vertex,blue!80] (18-2) at ($(-1.4,1.4)+(-0.2,0.2)$) {};
   \node[vertex,red!80] (18-3) at ($(-1.4,1.4)+(-0.2,-0.2)$) {};
   \node[vertex,blue!80] (18-4) at ($(-1.4,1.4)+(0.2,-0.2)$) {};
   
   \node[vertex,blue!80] (19-1) at ($(-1.4,-1.4)+(0.2,0.2)$) {};
   \node[vertex,red!80] (19-2) at ($(-1.4,-1.4)+(-0.2,0.2)$) {};
   \node[vertex,blue!80] (19-3) at ($(-1.4,-1.4)+(-0.2,-0.2)$) {};
   \node[vertex,red!80] (19-4) at ($(-1.4,-1.4)+(0.2,-0.2)$) {};
   
   \node[vertex,red!80] (20-1) at ($(1.4,-1.4)+(0.2,0.2)$) {};
   \node[vertex,blue!80] (20-2) at ($(1.4,-1.4)+(-0.2,0.2)$) {};
   \node[vertex,red!80] (20-3) at ($(1.4,-1.4)+(-0.2,-0.2)$) {};
   \node[vertex,blue!80] (20-4) at ($(1.4,-1.4)+(0.2,-0.2)$) {};
   
   \foreach \i in {9,...,20}{
    \foreach \v/\w in {1/2,2/3,3/4,1/4}{
     \draw[line width=1.6pt,darkpastelgreen] (\i-\v) edge (\i-\w);
    }
   }
   
   \foreach \v/\w in {1/9-2,2/9-4,2/12-1,3/12-3,3/10-4,4/10-2,4/11-3,1/11-1,5/13-2,6/13-4,6/16-1,7/16-3,7/14-4,8/14-2,8/15-3,5/15-1,
                      1/17-3,5/17-1,2/18-4,6/18-2,3/19-1,7/19-3,4/20-2,8/20-4}{
    \draw[line width=1.6pt] (\v) edge (\w);
   }
   
  \end{tikzpicture}
  }
  \caption{$C_4$-subdivision of a cube}
  \label{fig:graph-c4-cube}
 \end{subfigure}
 \hfill
 \begin{subfigure}[b]{.205\linewidth}
  \scalebox{\figscalelarge}{
  \begin{tikzpicture}
   
   \node[vertex,blue!80] (1-1) at ($(0,1.2)+(0.4,0)$) {};
   \node[vertex,red!80] (1-2) at ($(0,1.2)+(0,0.4)$) {};
   \node[vertex,blue!80] (1-3) at ($(0,1.2)+(-0.4,0)$) {};
   \node[vertex,red!80] (1-4) at ($(0,1.2)+(0,-0.4)$) {};
   
   \node[vertex,red!80] (2-1) at ($(1.2,0)+(0.4,0)$) {};
   \node[vertex,blue!80] (2-2) at ($(1.2,0)+(0,0.4)$) {};
   \node[vertex,red!80] (2-3) at ($(1.2,0)+(-0.4,0)$) {};
   \node[vertex,blue!80] (2-4) at ($(1.2,0)+(0,-0.4)$) {};
   
   \node[vertex,blue!80] (3-1) at ($(0,-1.2)+(0.4,0)$) {};
   \node[vertex,red!80] (3-2) at ($(0,-1.2)+(0,0.4)$) {};
   \node[vertex,blue!80] (3-3) at ($(0,-1.2)+(-0.4,0)$) {};
   \node[vertex,red!80] (3-4) at ($(0,-1.2)+(0,-0.4)$) {};
   
   \node[vertex,red!80] (4-1) at ($(-1.2,0)+(0.4,0)$) {};
   \node[vertex,blue!80] (4-2) at ($(-1.2,0)+(0,0.4)$) {};
   \node[vertex,red!80] (4-3) at ($(-1.2,0)+(-0.4,0)$) {};
   \node[vertex,blue!80] (4-4) at ($(-1.2,0)+(0,-0.4)$) {};
   
   \foreach \i in {1,2,3,4}{
    \foreach \v/\w in {1/2,2/3,3/4,1/4}{
     \draw[line width=1.6pt,darkpastelgreen] (\i-\v) edge (\i-\w);
    }
   }
   
   \draw[line width=1.6pt] (1-1) edge (2-2);
   \draw[line width=1.6pt] (1-3) edge (4-2);
   \draw[line width=1.6pt] (2-4) edge (3-1);
   \draw[line width=1.6pt] (3-3) edge (4-4);
   
  \end{tikzpicture}
  }
  \caption{$C_4^*$}
  \label{fig:graph-c-4-star}
 \end{subfigure}
 \hfill
 \begin{subfigure}[b]{.24\linewidth}
  \scalebox{\figscalelarge}{
  \begin{tikzpicture}
   
   \node[vertex,red!80] (1-1) at ($(0,0)+(0.4,0)$) {};
   \node[vertex,blue!80] (1-2) at ($(0,0)+(0,0.4)$) {};
   \node[vertex,red!80] (1-3) at ($(0,0)+(-0.4,0)$) {};
   \node[vertex,blue!80] (1-4) at ($(0,0)+(0,-0.4)$) {};
   
   \node[vertex,red!80] (2-1) at ($(1.5,0)+(0.4,0)$) {};
   \node[vertex,blue!80] (2-2) at ($(1.5,0)+(0,0.4)$) {};
   \node[vertex,red!80] (2-3) at ($(1.5,0)+(-0.4,0)$) {};
   \node[vertex,blue!80] (2-4) at ($(1.5,0)+(0,-0.4)$) {};
   
   \node[vertex,red!80] (3-1) at ($(-1.5,0)+(0.4,0)$) {};
   \node[vertex,blue!80] (3-2) at ($(-1.5,0)+(0,0.4)$) {};
   \node[vertex,red!80] (3-3) at ($(-1.5,0)+(-0.4,0)$) {};
   \node[vertex,blue!80] (3-4) at ($(-1.5,0)+(0,-0.4)$) {};
   
   \node[vertex,magenta!80] (u1) at (0,2) {};
   \node[vertex,magenta!80] (u2) at (0,-2) {};
   
   \foreach \i in {1,2,3}{
    \foreach \v/\w in {1/2,2/3,3/4,1/4}{
     \draw[line width=1.6pt,darkpastelgreen] (\i-\v) edge (\i-\w);
    }
   }
   
   \draw[line width=1.6pt] (u1) edge (1-2);
   \draw[line width=1.6pt] (u1) edge (2-2);
   \draw[line width=1.6pt] (u1) edge (3-2);
   \draw[line width=1.6pt] (u2) edge (1-4);
   \draw[line width=1.6pt] (u2) edge (2-4);
   \draw[line width=1.6pt] (u2) edge (3-4);
   
  \end{tikzpicture}
  }
  \caption{$K_{2,3}^*$}
  \label{fig:graph-k23-star}
 \end{subfigure}
 \vspace{0.22cm}
 \begin{center}
 \begin{subfigure}[b]{.31\linewidth}
  \scalebox{\figscalelarge}{
  \begin{tikzpicture}[scale = 0.8]
   \node[vertex,red!80] (1) at (0.5,0.5) {};
   \node[vertex,red!80] (2) at (-0.5,0.5) {};
   \node[vertex,red!80] (3) at (0.5,-0.5) {};
   \node[vertex,red!80] (4) at (-0.5,-0.5) {};
   
   \node[vertex,red!80] (5) at (1.2,0.5) {};
   \node[vertex,red!80] (6) at (0.5,1.2) {};
   \node[vertex,red!80] (7) at (-1.2,0.5) {};
   \node[vertex,red!80] (8) at (-0.5,1.2) {};
   \node[vertex,red!80] (9) at (1.2,-0.5) {};
   \node[vertex,red!80] (10) at (0.5,-1.2) {};
   \node[vertex,red!80] (11) at (-1.2,-0.5) {};
   \node[vertex,red!80] (12) at (-0.5,-1.2) {};
   
   \node[vertex,red!80] (13) at (2.4,1.7) {};
   \node[vertex,red!80] (14) at (1.7,2.4) {};
   \node[vertex,red!80] (15) at (-2.4,1.7) {};
   \node[vertex,red!80] (16) at (-1.7,2.4) {};
   \node[vertex,red!80] (17) at (2.4,-1.7) {};
   \node[vertex,red!80] (18) at (1.7,-2.4) {};
   \node[vertex,red!80] (19) at (-2.4,-1.7) {};
   \node[vertex,red!80] (20) at (-1.7,-2.4) {};
   
   \node[vertex,red!80] (21) at (3.0,3.0) {};
   \node[vertex,red!80] (22) at (-3.0,3.0) {};
   \node[vertex,red!80] (23) at (3.0,-3.0) {};
   \node[vertex,red!80] (24) at (-3.0,-3.0) {};
   
   \foreach \i/\j in {1/2,2/4,3/4,1/3,5/9,6/8,7/11,10/12,5/13,6/14,7/15,8/16,9/17,10/18,11/19,12/20,13/17,14/16,15/19,18/20,21/22,22/24,23/24,21/23}{
    \draw[line width=1.6pt] (\i) edge (\j);
   }
   \foreach \i/\j in {1/5,1/6,2/7,2/8,3/9,3/10,4/11,4/12,5/6,7/8,9/10,11/12,13/21,14/21,15/22,16/22,17/23,18/23,19/24,20/24,13/14,15/16,17/18,19/20}{
    \draw[line width=1.6pt,darkpastelgreen] (\i) edge (\j);
   }
   
  \end{tikzpicture}
  }
  \caption{Rhombicuboctahedron}
  \label{fig:graph-rhombicuboctahedron}
 \end{subfigure}
 \hspace{2cm}
 \begin{subfigure}[b]{.37\linewidth}
  \scalebox{\figscalelarge}{
  \begin{tikzpicture}[scale = 0.9]
   
   \node[vertex,red!80] (1) at ($(0,0)+(54:0.4)$) {};
   \node[vertex,red!80] (2) at ($(0,0)+(126:0.4)$) {};
   \node[vertex,red!80] (3) at ($(0,0)+(198:0.4)$) {};
   \node[vertex,red!80] (4) at ($(0,0)+(270:0.4)$) {};
   \node[vertex,red!80] (5) at ($(0,0)+(342:0.4)$) {};
   
   \foreach[count=\i] \v in {6,...,15}{
    \node[vertex,red!80] (\v) at ($(0,0)+(36*\i:0.8)$) {};
   }
   \foreach[count=\i] \v in {16,...,25}{
    \node[vertex,red!80] (\v) at ($(0,0)+(36*\i:1.4)$) {};
   }
   
   \foreach[count=\i] \v in {26,...,30}{
    \node[vertex,red!80] (\v) at ($(0,0)+(18+72*\i:1.7)$) {};
   }
   \foreach[count=\i] \v in {31,...,35}{
    \node[vertex,red!80] (\v) at ($(0,0)+(-18+72*\i:1.8)$) {};
   }
   
   \node[vertex,red!80] (36) at ($(0,0)+(82:2.2)$) {};
   \node[vertex,red!80] (37) at ($(0,0)+(98:2.2)$) {};
   \node[vertex,red!80] (38) at ($(0,0)+(154:2.2)$) {};
   \node[vertex,red!80] (39) at ($(0,0)+(170:2.2)$) {};
   \node[vertex,red!80] (40) at ($(0,0)+(226:2.2)$) {};
   \node[vertex,red!80] (41) at ($(0,0)+(242:2.2)$) {};
   \node[vertex,red!80] (42) at ($(0,0)+(298:2.2)$) {};
   \node[vertex,red!80] (43) at ($(0,0)+(314:2.2)$) {};
   \node[vertex,red!80] (44) at ($(0,0)+(10:2.2)$) {};
   \node[vertex,red!80] (45) at ($(0,0)+(26:2.2)$) {};
   
   \node[vertex,red!80] (46) at ($(0,0)+(82:2.8)$) {};
   \node[vertex,red!80] (47) at ($(0,0)+(98:2.8)$) {};
   \node[vertex,red!80] (48) at ($(0,0)+(154:2.8)$) {};
   \node[vertex,red!80] (49) at ($(0,0)+(170:2.8)$) {};
   \node[vertex,red!80] (50) at ($(0,0)+(226:2.8)$) {};
   \node[vertex,red!80] (51) at ($(0,0)+(242:2.8)$) {};
   \node[vertex,red!80] (52) at ($(0,0)+(298:2.8)$) {};
   \node[vertex,red!80] (53) at ($(0,0)+(314:2.8)$) {};
   \node[vertex,red!80] (54) at ($(0,0)+(10:2.8)$) {};
   \node[vertex,red!80] (55) at ($(0,0)+(26:2.8)$) {};
   
   \foreach[count=\i] \v in {56,...,60}{
    \node[vertex,red!80] (\v) at ($(0,0)+(18+72*\i:3.6)$) {};
   }
   
   \foreach \i/\j in {1/2,2/3,3/4,4/5,1/5,7/8,9/10,11/12,13/14,15/6,6/16,7/17,8/18,9/19,10/20,11/21,12/22,13/23,14/24,15/25,
                      17/26,18/26,19/27,20/27,21/28,22/28,23/29,24/29,25/30,16/30,31/36,31/45,32/37,32/38,33/39,33/40,34/41,34/42,35/43,35/44,
                      36/46,37/47,38/48,39/49,40/50,41/51,42/52,43/53,44/54,45/55,46/55,47/48,49/50,51/52,53/54,56/57,57/58,58/59,59/60,56/60}{
    \draw[line width=1.6pt] (\i) edge (\j);
   }
   \foreach \i/\j in {1/6,1/7,2/8,2/9,3/10,3/11,4/12,4/13,5/14,5/15,6/7,8/9,10/11,12/13,14/15,16/17,18/19,20/21,22/23,24/25,
                      16/31,17/31,18/32,19/32,20/33,21/33,22/34,23/34,24/35,25/35,26/36,26/37,27/38,27/39,28/40,28/41,29/42,29/43,30/44,30/45,
                      36/37,38/39,40/41,42/43,44/45,46/47,48/49,50/51,52/53,54/55,46/56,47/56,48/57,49/57,50/58,51/58,52/59,53/59,54/60,55/60}{
    \draw[line width=1.6pt,darkpastelgreen] (\i) edge (\j);
   }
   
  \end{tikzpicture}
  }
  \caption{Rhombicosidodecahedron}
  \label{fig:graph-rhombicosidodecahedron}
 \end{subfigure}
 \end{center}
 \caption{Examples for the constructions from the classification for graphs of Type IIa.}
 \label{fig:c4-extensions}
\end{figure}

Let $H$ be a graph and $f\colon E(H) \rightarrow \NN$ be a function.
The \emph{$f$-subdivision of $H$} is the graph $H^{(f)}$ obtained from $H$ by replacing each edge $e$ with $f(e)$ parallel paths of length~$2$ (if $f(e) = 0$, the edge $e$ remains unaltered).
Formally, $H^{(f)}$ is the graph with vertex set $V(H^{(f)}) \coloneqq V(H) \uplus \{(e,i) \mid e \in E(H), i \in [f(e)]\}$ and edge set
$E(H^{(s)}) \coloneqq \{e \in E(H) \mid f(e) = 0\} \cup \{v(e,i) \mid e \in E(H), v \in e, i \in [f(e)]\}$.
A graph $G$ is a \emph{parallel subdivision of $H$} if there is a function $f\colon E(H) \rightarrow \NN$ such that $G$ is isomorphic to $H^{(f)}$.

\begin{theorem}
 \label{thm:type-ii-connected-subgraphs-classification}
 Let $G$ be a $3$-connected planar graph of Type IIa and let $\chi \coloneqq \WL{2}{G}$ be the coloring computed by $2$-WL.
 Then one of the following options holds.
 \begin{enumerate}[label=(\Alph*)]
  \item\label{item:type-ii-connected-subgraphs-classification-1} There is a vertex $v \in V(G)$ such that $\Disc_G(v) = V(G)$,
  \item there is an edge color $c \in C_E(G,\chi)$ that defines a matching, or
  \item\label{item:type-ii-connected-subgraphs-classification-3} there are colors $c,d \in C_E(G,\chi)$ such that $G[c,d]$ is isomorphic to a parallel subdivision of one of the following graphs:
  \begin{enumerate}[label=\arabic*.]
   \item a truncated tetrahedron, a truncated cube, a truncated octahedron, a truncated dodecahedron, a truncated icosahedron,
   \item an $m$-side prism for $m \geq 3$,
   \item a cuboctahedron (with two edge colors), a rhombicuboctahedron, a rhombicosidodecahedron,
   \item a $C_4$-subdivision of one of the graphs from Figure \ref{fig:graph-k4} -- \ref{fig:graph-icosahedron},
   \item a $C_m^*$ for $m \geq 2$,
   \item a $K_{2,h}^*$ for $h \geq 3$,
   \item a chamfered tetrahedron, a chamfered cube, a chamfered octahedron, a chamfered dodecahedron, or a chamfered icosahedron.
  \end{enumerate}
 \end{enumerate}
\end{theorem}

\begin{remark}
 There are four Archimedean solids that are explicitly listed neither in Theorem \ref{thm:type-iii-classification} nor in Theorem \ref{thm:type-ii-connected-subgraphs-classification}.
 These are the truncated cuboctahedron, the truncated icosidodecahedron, the snub cube, and the snub dodecahedron.
 The graphs corresponding to these solids have fixing number $1$ under $1$-WL and hence, they implicitly appear in Theorem \ref{thm:type-ii-connected-subgraphs-classification}, Option \ref{item:type-ii-connected-subgraphs-classification-1}.
\end{remark}

We now dive into the lengthy proof of Theorem \ref{thm:type-ii-connected-subgraphs-classification}.
Using Lemma \ref{la:cycle-connections}, we may assume that there are colors $c,d \in C_E(G,\chi)$ such that $c$ has Type IIa and admits short $d$-connections in~$G$.
The following lemma turns out to be useful to show that the graph $G[c,d]$ is connected.

Recall that a set of edge colors $C \subseteq C_E(G,\chi)$ is locally determined if $C \subseteq C_E(G[A],\chi)$ for every connected component $A$ of $G[C]$, i.e., all edge colors from $C$ appear in every connected component of $G[C]$.

\begin{lemma}
 \label{la:disconnected-color-subgraphs}
 Let $G$ be a connected graph and let $\chi \coloneqq \WL{2}{G}$ be the coloring computed by $2$-WL.
 Let $C \subseteq C_E(G,\chi)$ be a locally determined set of edge colors such that $G[C]$ is disconnected.
 Let $A$ be the vertex set of a connected component of $G[C]$.
 Then there is a $\chi$-invariant set $W \subseteq V(G)$ and a vertex set $Z$ of a connected component of $G - A$ such that $V(G[C]) \setminus A \subseteq Z$ and $N_G(Z) = W \cap A$.
\end{lemma}

We remark that the lemma is proved by the same arguments already used in the first part of the proof of Lemma \ref{la:3-faces-connected}.

\begin{proof}
 Let $A_1,\dots,A_\ell$ denote the connected components of $G[C]$.
 Without loss of generality assume that $A = A_1$.
 By Lemma \ref{la:path-between-color-components}, there is a vertex set $Z$ of some connected component of the graph $G - A$ such that $A_i \subseteq Z$ for all $i \in \{2,\dots,\ell\}$.
 Now let $v \in N_G(Z)$ and $w \in A$ such that $\chi(v,v) = \chi(w,w)$.
 To prove the lemma, it suffices to show that $w \in N_G(Z)$.
 Let $D \coloneqq C_V(G[C],\chi)$ denote the set of vertex colors present in the graph $G[C]$.
 Observe that $\bigcup_{i \in [\ell]} A_i = \bigcup_{d \in D} V_d(G,\chi)$.
 Since $v \in N_G(Z)$, there is some $i \in \{2,\dots,\ell\}$ and a vertex $v' \in A_i$ such that there is a path from $v$ to $v'$ that avoids $D$.
 Now let $w' \in V(G)$ such that $\chi(v,v') = \chi(w,w')$ (such a vertex exists since $\chi(v,v) = \chi(w,w)$).
 Then $\chi(v',v') = \chi(w',w')$ and hence, $w' \in \bigcup_{i \in [\ell]} A_i$.
 Since $2$-WL knows the connected components of $G[C]$ and $v,v'$ are in distinct connected components of $G[C]$, it follows that $w,w'$ are also in distinct connected components of $G[C]$.
 In other words, there is some $j \in \{2,\dots,\ell\}$ such that $w' \in A_j$.
 In particular, $w' \in Z$.
 Moreover, by Observation \ref{obs:wl-knows-paths-avoiding-colors}, there is a path from $w$ to $w'$ that avoids $D$.
 All this is only possible if $w \in N_G(Z)$.
\end{proof}

Before continuing, let us briefly explain how the previous lemma can be applied to show that certain graphs $G[C]$ are connected assuming $G$ is planar.
As a simple example, suppose that $C = \{c\}$ and every connected component of $G[c]$ is isomorphic to $K_4$.
In particular, $G[c]$ is unicolored.
Suppose towards a contradiction that $G[c]$ is disconnected and let $A$ be the vertex set of a connected component of $G[c]$.
Then the last lemma provides a connected component $Z$ of $G - A$ and a $\chi$-invariant set $W$ such that $N_G(Z) = W \cap A$.
Since $N_G(Z) \neq \emptyset$ ($G$ is connected) and $\chi(v,v) = \chi(w,w)$ for all $v,w \in A$, it follows that $N_G(Z) = A$.
However, $G[A]$ is isomorphic to $K_4$ and hence, contracting $Z$ to a single vertex, we obtain a $K_5$-minor, which contradicts by Theorem \ref{thm:wagner} the planarity of $G$.
So $G[c]$ is connected.

This argument can be applied whenever for every connected component $A$ of $G[C]$ and every vertex color $d \in C_V(G[C],\chi)$, there is no face $F$ of $G[C]$ (where the embedding is inherited from some embedding of $G$) such that every vertex from $A \cap V_d$ is incident to $F$.
We shall repeatedly use this argument in the classification below.

Also, we need one more auxiliary result to cover a certain case where Lemma \ref{la:disconnected-color-subgraphs} is not applicable.

\begin{lemma}
 \label{la:c-m-star-connected}
 Let $G$ be a $3$-connected planar graph and let $\chi \coloneqq \WL{2}{G}$ be the coloring computed by $2$-WL.
 Suppose there are colors $c,d \in C_E(G,\chi)$ and a vertex set $A$ of a connected component of $G[c,d]$ such that $(G[c,d])[A]$ is isomorphic to a parallel subdivision of $C_m^*$ for some $m \geq 2$.
 Then $G[c,d]$ is connected.
\end{lemma}

\begin{proof}
 The proof is identical to the second part of the proof of Lemma \ref{la:3-faces-connected} covering the case of an $s$-subdivision of a cycle $C_\ell$ for some $\ell \geq 3$ and $s \geq 2$.
\end{proof}

Now, we are ready to prove the second main lemma of this section, which together with Lemma \ref{la:cycle-connections} will yield Theorem \ref{thm:type-ii-connected-subgraphs-classification}.

\begin{lemma}
 \label{la:type-ii-connected-subgraphs-classification}
 Let $G$ be a $3$-connected planar graph and let $\chi \coloneqq \WL{2}{G}$ be the coloring computed by $2$-WL.
 Suppose that $G$ has Type IIa.
 Also suppose there are colors $c,d \in C_E(G,\chi)$ such that $c$ has Type IIa and admits short $d$-connections in $G$.
 Then there is a vertex $v \in V(G)$ such that $\Disc_G(v) = V(G)$, or there is an edge color $c' \in C_E(G,\chi)$ that defines a matching, or $G[c,d]$ is isomorphic to one of the following graphs:
 \begin{enumerate}
  \item a (parallel subdivision of a) truncated tetrahedron, a truncated cube, a truncated octahedron, a truncated dodecahedron, a truncated icosahedron,
  \item a (parallel subdivision of an) $m$-side prism for $m \geq 3$,
  \item a (parallel subdivision of a) cuboctahedron (with two edge colors), a rhombicuboctahedron, or a rhombicosidodecahedron,
  \item a $C_4$-subdivision of one of the graphs from Figure \ref{fig:graph-k4} -- \ref{fig:graph-icosahedron},
  \item a (parallel subdivision of a) $C_m^*$ for $m \geq 2$,
  \item a $K_{2,h}^*$ for $h \geq 3$,
  \item a (parallel subdivision of a) chamfered tetrahedron, a chamfered cube, a chamfered octahedron, a chamfered dodecahedron, or a chamfered icosahedron.
 \end{enumerate}
\end{lemma}

\begin{proof}
 Let $A_1,\dots,A_\ell$ denote the vertex sets of the connected components of $G[c]$.
 Also, let $\widehat{A}_i$ denote the set of vertices in $A_i$ that are incident to a $d$-colored edge.
 Observe that either $G[c]$ is unicolored, in which case $A_i = \widehat{A}_i$, or $G[c]$ is bicolored, which means that $|\widehat{A}_i| = |A_i|/2$.
 Also note that $\bigcup_{i \in [\ell]}\widehat{A}_i$ is $\chi$-invariant.
 Finally, observe that $\ell \geq 2$ since $c$ has Type IIa.
 
 For the remainder of the proof, we usually assume that $A_i = \widehat{A}_i$.
 If this is not the case, we can restrict ourselves to $\widehat{A}_i$ by defining a cycle on this set by connecting vertices of distance$2$ in $G[c]$.
 Note that reverting this step is covered by the classification due to allowing parallel subdivisions (each vertex in $A_i \setminus \widehat{A}_i$ has degree $2$ in $G[c,d]$).
 There is one case where this strategy does not exactly work out, namely if $G[c]$ is a disjoint union of bicolored cycles of length $4$.
 In this case, $|\widehat{A}_i| = 2$ and it is not possible to define a cycle on $\widehat{A}_i$.
 
 Similarly to the proof of Lemma \ref{la:cycle-connections}, we fix a planar embedding of $G$.
 For every $i \in [\ell]$, the cycle $(G[c])[A_i]$ splits the plane into two regions, the interior region and the exterior region.
 We denote by $\exte_c(A_i)$ the set of vertices located in the exterior region, and $\inte_c(A_i)$ denotes the set of vertices located in the interior region.
 By Lemma \ref{la:path-between-color-components}, we may assume without loss of generality that $A_j \subseteq \exte_c(A_i)$ for all distinct $i,j \in[\ell]$.
 This also means that $V(G[c,d]) \setminus V(G[c]) \subseteq \exte_c(A_i)$ for all $i \in [\ell]$.
 
 By Lemma \ref{la:factor-graph-2-wl}, it holds that $\chi/c$ is a $2$-stable coloring for $G/c$.
 Also, there is a unique color $\widetilde{d} \in \im(\chi/c)$ such that $d \in \widetilde{d}$ (recall that colors in the image of $\chi/c$ are multisets of colors from the image of $\chi$).
 Consider the graph $(G/c)[\widetilde{d}]$ obtained from $G[c,d]$ by contracting all sets $A_i$ to single vertices.
 We use the results from Section \ref{sec:edge-transitive} to classify $(G/c)[\widetilde{d}]$, and then reconstruct $G[c,d]$.
 Towards this end, we distinguish several cases.
 In many cases, we rely on Lemma \ref{la:disconnected-color-subgraphs} to argue that $G[c,d]$ is connected.
 Let us remark that the lemma is applicable since $\{c,d\}$ is locally determined by the choice of the colors $c,d$.
 
 \begin{description}
  \item[{Case $(G/c)[\widetilde{d}]$ is unicolored:}]
   This means every connected component of $(G/c)[\widetilde{d}]$ is isomorphic to $K_2$, or to $C_m$ for some $m \geq 3$, or to one of the graphs from Figure \ref{fig:graph-k4} -- \ref{fig:graph-icosidodecahedron} using Theorem~\ref{thm:classification-edge-transitive}.
   Observe that every vertex from $(G/c)[\widetilde{d}]$ is one of the sets $A_i$.
   In particular $\ell = |V((G/c)[\widetilde{d}])|$.
   
   We first cover the case that $G[d]$ is bicolored.
   
   \begin{claim}
    If $G[d]$ is bicolored, then $d$ defines a matching in $G$.
   \end{claim}
   \begin{claimproof}
    Suppose that $G[d]$ is bicolored.
    First observe that $V(G[c]) = V(G[c,d])$.
    This implies that $G[c]$ is bicolored as well.
    Suppose $C_V(G[c]) = \{c_V,d_V\}$ is the set of the two vertex colors appearing in $G[c,d]$.
    Let $V_1 \coloneqq V_{c_V}$ and $V_2 \coloneqq V_{d_V}$ denote the corresponding color classes.
    Note that $G[c,d]$ is bipartite with bipartition $(V_1,V_2)$.
    Also observe that $|V_1| = |V_2|$ since $G[c]$ has Type II, i.e., it is a disjoint union of bicolored cycles.
    Since $2$-WL cannot distinguish between vertices from $V_i$, there is some number $r_i$ such that $\deg_{G[c,d]}(v_i) = r_i$ for all $v_i \in V_i$, $i \in \{1,2\}$.
    Moreover, $|E(G[c,d])| = r_1 \cdot |V_1| = r_2 \cdot |V_2|$.
    Since $|V_1| = |V_2|$, we conclude that $r \coloneqq r_1 = r_2$, i.e., $G[c,d]$ is $r$-regular.
    Since $G[c,d]$ is planar and bipartite, it follows that $r \leq 3$ using Corollary \ref{cor:planar-degree-bound}.
    On the other hand, $|N_c(v)| = 2$ and $|N_d(v)| \geq 1$ for every $v \in V(G[c,d])$.
    So together, $r = 3$ and $|N_d(v)| = 1$ for every $v \in V(G[c,d])$.
    But $|N_d(v)| = 1$ for all $v \in V(G[c,d])$ means that $d$ defines a matching.
   \end{claimproof}

   So we may assume that $G[d]$ is unicolored and hence, it is isomorphic to a disjoint union of $K_2$'s, or a disjoint union of cycles $C_m$ for $m \geq 3$ (recall that $G$ has Type IIa).
   
   \medskip
   
   First suppose that $G[d]$ is isomorphic to a disjoint union of $K_2$'s.
   Consider a connected component $B$ of the graph $G[c,d]$ and the corresponding component $\widetilde{B}$ in the graph $(G/c)[\widetilde{d}]$.
   If $((G/c)[\widetilde{d}])[\widetilde{B}]$ is isomorphic to $K_2$ and $|\widehat{A}_i| \geq 3$ for all $i \in [\ell]$, then it is easy to see that $(G[c,d])[B]$ is isomorphic to a parallel subdivision of an $m$-side prism for some $m \geq 3$ ($m = |\widehat{A}_i|$ for all $i \in [\ell]$).
   It follows that $G[c,d]$ is connected by Lemma \ref{la:disconnected-color-subgraphs} (using the arguments outlined below the lemma).
   If $((G/c)[\widetilde{d}])[\widetilde{B}]$ is isomorphic to $K_2$ and $|\widehat{A}_i| = 2$ for all $i \in [\ell]$, then $(G[c,d])[B]$ is isomorphic to $C_2^*$.
   In particular, $G[c,d]$ is connected by Lemma \ref{la:c-m-star-connected}.
   
   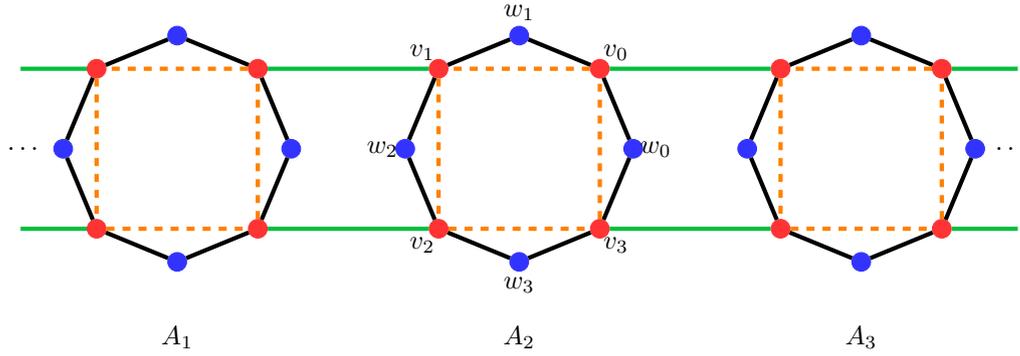
\begin{figure}
    \centering
    \begin{tikzpicture}
     
     \foreach \p/\q in {0/1,1/2,2/3}{
      \foreach \i in {0,1,2,3}{
       \node[vertex,blue!80] (u\i-\p) at ($(4.5*\p,0) + (90*\i:1.5)$) {};
       \node[vertex,red!80] (v\i-\p) at ($(4.5*\p,0) + (45 + 90*\i:1.5)$) {};
      }
      \node at (4.5*\p,-2.5) {$A_{\q}$};
      \foreach \i/\j in {0/0,0/1,1/1,1/2,2/2,2/3,3/3,3/0}{
       \draw[line width=1.6pt] (v\i-\p) edge (u\j-\p);
      }
      \foreach \i/\j in {0/1,1/2,2/3,3/0}{
       \draw[line width=1.6pt, dashed, orange] (v\i-\p) edge (v\j-\p);
      }
     }
     \draw[line width=1.6pt, darkpastelgreen] (v1-0) edge ($(v1-0)+(-1,0)$);
     \draw[line width=1.6pt, darkpastelgreen] (v2-0) edge ($(v2-0)+(-1,0)$);
     \draw[line width=1.6pt, darkpastelgreen] (v0-0) edge (v1-1);
     \draw[line width=1.6pt, darkpastelgreen] (v3-0) edge (v2-1);
     \draw[line width=1.6pt, darkpastelgreen] (v1-2) edge (v0-1);
     \draw[line width=1.6pt, darkpastelgreen] (v2-2) edge (v3-1);
     \draw[line width=1.6pt, darkpastelgreen] (v0-2) edge ($(v0-2)+(1,0)$);
     \draw[line width=1.6pt, darkpastelgreen] (v3-2) edge ($(v3-2)+(1,0)$);
     
     \node at (-2,0) {$\dots$};
     \node at (11,0) {$\dots$};
     
     \foreach \i in {0,1,2,3}{
       \node at ($(4.5,0) + (90*\i:1.8)$) {$w_\i$};
       \node at ($(4.5,0) + (45 + 90*\i:1.8)$) {$v_\i$};
      }
     
    \end{tikzpicture}
    \caption{Visualization for the proof of Lemma \ref{la:type-ii-connected-subgraphs-classification} assuming $(G/c)[\widetilde{d}]$ is unicolored, $G[d]$ is isomorphic to a disjoint union of $K_2$'s, and $((G/c)[\widetilde{d}])[\widetilde{B}]$ is isomorphic to $C_m$ for some $m \geq 3$.
     The vertices from the sets $\widehat{A}_i$ are colored {\sf red} and the vertices from the sets $A_i \setminus \widehat{A}_i$ are colored {\sf blue}.
     The edges of color $c$ are {\sf black} and edges of color $d$ are {\sf green}.
     Moreover, orange dashed edges indicate the cycles defined on the sets $\widehat{A}_i$.
     We have $\widehat{A}_{2,1} = \{v_1,v_2\}$ and $\widehat{A}_{2,3} = \{v_0,v_3\}$.
     There is a walk from $v_1$ to $v_2$ taking first a green edge, then a sequence of black edges, and finally one more green edge.
     Such a walk does not exist from $v_0$ to $v_1$.
     Hence, $\chi(v_0,v_1) \neq \chi(v_1,v_2)$.
     But this contradicts $\chi(v_1,w_1) = \chi(v_2,w_2)$.}
    \label{fig:c-m-star-case}
   \end{figure}
   
   Next, suppose that $((G/c)[\widetilde{d}])[\widetilde{B}]$ is isomorphic to $C_m$ for some $m \geq 3$.
   We first claim that $|\widehat{A}_i| = 2$ for all $i \in [\ell]$.
   Recall that each vertex of $(G/c)[\widetilde{d}]$ is one of the sets $A_i$, $i \in [\ell]$.
   Without loss of generality, suppose that $A_1A_2, A_2A_3$ are edges of $((G/c)[\widetilde{d}])[\widetilde{B}]$.
   Suppose towards a contradiction that $|\widehat{A}_2| \geq 3$ (see also Figure \ref{fig:c-m-star-case}).
   Every vertex $v \in \widehat{A}_2$ has exactly one neighbor $v'$ in the graph $G[d]$.
   Since $A_1$ and $A_3$ are the only neighbors of $A_2$ in the graph $(G/c)[\widetilde{d}]$, it holds that $v' \in \widehat{A}_1 \cup \widehat{A}_3$.
   For $j \in \{1,3\}$, define
   \[\widehat{A}_{2,j} \coloneqq \{v \in \widehat{A}_2 \mid N_{G[d]}(v) \cap \widehat{A}_j \neq \emptyset\}.\]
   Then $\{\widehat{A}_{2,1},\widehat{A}_{2,3}\}$ forms a partition of $\widehat{A}_2$.
   Since $A_j \subseteq \exte_c(A_i)$ for all distinct $i,j \in[\ell]$ and $G$ is planar, each set $\widehat{A}_{2,j}$, $j \in \{1,3\}$ gives a connected subgraph of $(G[c])[A_2]$ (when defining a cycle on $\widehat{A}_2$ as described above).
   Now let $j \in \{1,3\}$ such that $|\widehat{A}_{2,j}| \geq 2$ (recall that $|\widehat{A}_2| \geq 3$).
   Then edges from $G[c]$ located within the subgraph corresponding to $\widehat{A}_{2,j}$ receive different colors by $2$-WL than those edges not being part of either of the two subgraphs.
   But this contradicts the stability of $\chi$.
   Hence, $|\widehat{A}_2| = 2$.
   Since $|\widehat{A}_i| = |\widehat{A}_j|$ for all $i,j \in [\ell]$, it follows that $|\widehat{A}_i| = 2$ for all $i \in [\ell]$.
   So every connected component of $G[c]$ is a bicolored cycle of length $4$.
   Now, it is easy to see that $(G[c,d])[B]$ is isomorphic to $C_m^*$.
   As before, by Lemma \ref{la:c-m-star-connected}, this implies that $G[c,d]$ is connected and hence, $G[c,d]$ is isomorphic to $C_m^*$.
   
   So consider the case that $((G/c)[\widetilde{d}])[\widetilde{B}]$ is isomorphic to one of the graphs from Figure \ref{fig:graph-k4} -- \ref{fig:graph-icosidodecahedron}.
   First of all, by Lemma \ref{la:disconnected-color-subgraphs}, it follows that $(G/c)[\widetilde{d}]$ is connected.
   So $G[c,d]$ is also connected.
   Let $r$ denote the degree of every vertex in $(G/c)[\widetilde{d}])$.
   Using the same arguments as above, we get that $|\widehat{A}_i| = r$ for all $i \in [\ell]$.
   It follows that $G[c,d]$ is obtained from $(G/c)[\widetilde{d}]$ by truncating all vertices (and possibly subdividing edges).
   In other words, $G[c,d]$ is isomorphic to a parallel subdivision of a truncated tetrahedron, a truncated cube, a truncated octahedron, a truncated dodecahedron, a truncated icosahedron, a truncated cuboctahedron, or a truncated icosidodecahedron.
   For the last two options, it can be checked that $G[c,d]$ has fixing number $1$ under $1$-WL.
   Hence, there is some $v \in V(G)$ such that $\Disc_G(v) = V(G)$ by Lemma \ref{la:fixing-cycle-and-extra-vertex}.
   This completes the subcase that $G[d]$ is isomorphic to a disjoint union of $K_2$'s.
   
   \medskip
   
   So assume that $G[d]$ is isomorphic to a disjoint union of $C_m$'s for some $m \geq 3$.
   First observe that $G[d]$ is disconnected since $G$ has Type IIa.
   
   \begin{claim}
    Let $i \in [\ell]$.
    For $v,w \in \widehat{A}_i$ define $v \sim w$ if $v,w$ are contained in the same connected component of $G[d]$.
    Then $\sim$ is an equivalence relation.
    Moreover, $\sim$ is trivial, i.e., there is only one equivalence class or every equivalence class has size $1$. 
   \end{claim}
   
   \begin{claimproof}
    Clearly, $\sim$ is an equivalence relation with the equivalence classes being the intersections of the connected components of $G[d]$ with $\widehat{A}_i$.
    Suppose towards a contraction that $\sim$ is non-trivial, i.e, there are connected components of $B_1,B_2$ of $G[d]$ such that $|B_1 \cap \widehat{A}_i| \geq 2$ and  $|B_2 \cap \widehat{A}_i| \geq 1$.
    First, observe that $|B_1 \cap \widehat{A}_i| = |B_2 \cap \widehat{A}_i|$ since otherwise $2$-WL would distinguish between vertices from $B_1 \cap \widehat{A}_i$ and vertices from $B_2 \cap \widehat{A}_i$.
    So $|B_2 \cap \widehat{A}_i| \geq 2$.
    In particular, $|\widehat{A}_i| \geq 4$.
    Actually, let $B_1,\dots,B_k$ denote those connected components of $G[d]$ that have a non-empty intersection with $\widehat{A}_i$.
    Note that $B_j \setminus \widehat{A}_i \subseteq \exte_c(A_i)$ for all $j \in [k]$.
    By planarity, we conclude that $B_j \cap \widehat{A}_i$ induces a connected set on $\widehat{A}_i$ for each $j \in [k]$ when defining a cycle on $\widehat{A}_i$ in the natural way.
    But now, (virtual) edges of the cycle on $\widehat{A}_i$ within a component $B_j$ receive different $2$-WL-colors than (virtual) edges between distinct sets $B_j$, $j \in [k]$.
    This contradicts the stability of the coloring $\chi$.
   \end{claimproof}
   
   As before, consider a connected component $B$ of the graph $G[c,d]$ and the corresponding component $\widetilde{B}$ in the graph $(G/c)[\widetilde{d}]$.
   Let $A_i \in \widetilde{B}$ and let $\sim$ be the equivalence relation on $\widehat{A}_i$ defined in the last claim.
   First suppose that $\sim$ has only one equivalence class, i.e., there is some connected component $X$ of $G[d]$ such that $\widehat{A}_i \subseteq X$.
   Since $c$ admits short $d$-connections, there is another component $A_{i'}$, $i \neq i' \in [\ell]$, such that $\widehat{A}_{i'} \cap X \neq \emptyset$.
   Since $2$-WL does not distinguish between vertices from $\widehat{A}_i$ and $\widehat{A}_{i'}$, we conclude that $\widehat{A}_{i'} \subseteq X$.
   It follows that $\widetilde{B} = \{A_i,A_{i'}\}$ and $(G[c,d])[B]$ is isomorphic to a parallel subdivision of an antiprism.
   By Lemma \ref{la:disconnected-color-subgraphs}, we conclude that $G[c,d]$ is connected.
   But now, $G[d]$ is connected, which contradicts $G$ having Type IIa.
   Hence, this case can actually not occur.
   
   So assume that every equivalence class of $\sim$ contains only a single element.
   Actually, we may assume the analogous statement holds for all $i \in [\ell]$.
   Let $X$ be a connected component of $G[d]$.
   Then every vertex of $X$ is contained in a different component $A_i$, $i \in [\ell]$, of $G[c]$.
   So the cycle $(G[d])[X]$ gives a cycle of $(G/c)[\widetilde{d}]$, which actually forms a face cycle of $(G/c)[\widetilde{d}]$ (where the embedding is inherited from $G$ in the natural way).
   Recall that $((G/c)[\widetilde{d}])[\widetilde{B}]$ is isomorphic to a cycle $C_m$, $m \geq 3$, or one of the graphs from Figure \ref{fig:graph-k4} -- \ref{fig:graph-icosidodecahedron}.
   We obtain that $(G[c,d])[B]$ is a parallel subdivision of the truncated version of one the graphs from Figure \ref{fig:graph-k4} -- \ref{fig:graph-icosidodecahedron}, a parallel subdivision of an $m$-side prism, or a parallel subdivision of a cuboctahedron (with two edge colors), a rhombicuboctahedron, or a rhombicosidodecahedron. 
   For a parallel subdivision of a truncated cuboctahedron or a truncated icosidodecahedron, we again obtain that $\Disc_G(v) = V(G)$ for some $v \in V(G)$.
   For all other cases, as usual, we apply Lemma \ref{la:disconnected-color-subgraphs} to deduce that $G[c,d]$ is connected.
  \item[{Case $(G/c)[\widetilde{d}]$ is bicolored:}]
   Let $W$ denote the set of vertices of $(G/c)[\widetilde{d}]$ that are not obtained from contracting the sets $A_1,\dots,A_\ell$.
   Hence, $(G/c)[\widetilde{d}]$ has vertex color classes $W$ and $\{A_1,\dots,A_\ell\}$.
   In particular, all vertices $w \in W$ have the same degree $\widetilde{r}$ in the graph $(G/c)[\widetilde{d}]$.
   Observe that $\widetilde{r} \geq 2$, since $c$ admits short $d$-connections.
   Similarly, all vertices $w \in W$ have the same degree $r$ in the graph $G[c,d]$.
   Note that $r \geq \widetilde{r}$.
   
   If $r = 2$, then we can repeat the analysis of the previous case to obtain that $G[c,d]$ is a parallel subdivision of one of the graphs considered in the previous case.
   
   So suppose $r \geq 3$.
   Since $G$ has Type IIa, it follows that $d$ has Type I in $G$.
   We first argue that $\widetilde{d}$ does not have Type I in $G/c$.
   
   \begin{claim}
    \label{claim:contracted-color-type-iii}
    $\widetilde{d}$ has not Type I in the graph $G/c$.
   \end{claim}
   \begin{claimproof}
    Suppose towards a contradiction that $\widetilde{d}$ has Type I in the graph $G/c$.
    Since $\widetilde{r} \geq 2$, we conclude that for every $i \in [\ell]$, there is some $w_i \in V(G)$ such that $\widehat{A}_i \subseteq N_d(w_i)$.
    Also, there are distinct $i,i' \in [\ell]$ such that $w_i = w_{i'}$.
    Let $I \coloneqq \{i'' \mid w_i = w_{i''}\}$.
    Without loss of generality, suppose that $I = \{1,\dots,k\}$ for some $2 \leq k \leq \ell$ and define $w \coloneqq w_1$.
    Observe that $N_d(w) = \bigcup_{i \in [k]} \widehat{A}_i$.
    
    First suppose $G[c,d]$ is disconnected.
    Let $B$ denote the connected component of $G[c,d]$ that contains $w$.
    By Lemma \ref{la:disconnected-color-subgraphs}, there is a $\chi$-invariant set $M \subseteq V(G)$ and a connected component $Z$ of $G - B$ such that $N_G(Z) = M \cap B$.
    Recall that $W$ forms a vertex color class of $(G,\chi)$ and $B \cap W = \{w\}$ (since $\widetilde{d}$ has Type I).
    Since $G$ is $3$-connected, it follows that
    \begin{enumerate}[label=(\roman*)]
     \item\label{item:red-connected} $\bigcup_{i \in [k]} \widehat{A}_i \subseteq N_G(Z)$, or
     \item\label{item:blue-connected} $G[c]$ is bicolored and $\bigcup_{i \in [k]} A_i \setminus\widehat{A}_i \subseteq N_G(Z)$.
    \end{enumerate}
    If $|\widehat{A}_i| \geq 3$, then we obtain a $K_5$.minor (with two vertices $w$ and $Z \cup \bigcup_{i \in \{2,\dots,k\}} A_i$ as well as a triangle obtained from $A_1$), which contradicts by Theorem \ref{thm:wagner} the planarity of $G$.
    So suppose $|\widehat{A}_i| = 2$, i.e., $G[c]$ is bicolored.
    If Option \ref{item:blue-connected} is satisfied, then we obtain a minor $K_{3,3}$ (with the four vertices from $A_1$, $w$, and $Z \cup \bigcup_{i \in \{2,\dots,k\}} A_i$; see Figure \ref{fig:c-d-disconnected-k33}), which again contradicts with Theorem \ref{thm:wagner} the planarity of $G$.
    So assume that Option \ref{item:blue-connected} is not satisfied, which implies that \ref{item:red-connected} is satisfied.
    We consider the set $A_1$.
    Suppose $\widehat{A}_1 = \{v_1,v_2\}$ and $A_1 \setminus \widehat{A}_1 = \{u_1,u_2\}$.
    Since $G$ is $3$-connected, the graph $G - \{v_1,v_2\}$ is connected, and hence there is a path $P_1$ from $u_1$ to $W \cup \bigcup_{i \in \{2,\dots,\ell\}} A_i$.
    We may assume without loss of generality that no internal vertex of $P_1$ is contained in $W \cup \bigcup_{i \in \{2,\dots,\ell\}} A_i$.
    Also, by possibly switching the roles of $u_1$ and $u_2$, we may assume that no internal vertex of $P_1$ is contained in $A_1$.
    Hence, $P_1$ avoids $C_V(G[c,d],\chi)$.
    In particular, $P_1$ is disjoint from $Z$ since all neighbors of $Z$ in $G - \{v_1,v_2\}$ are contained in $W \cup \bigcup_{i \in \{2,\dots,\ell\}} A_i$.
    First suppose that $P_1$ ends in the set $W$.
    This means that $P_1$ ends in $w$ since it is disjoint from $Z$ using Lemma \ref{la:disconnected-color-subgraphs}.
    Using Observation \ref{obs:wl-knows-paths-avoiding-colors}, there is also a path $P_2$ from $u_2$ to $W$ that avoids $C_V(G[c,d],\chi)$.
    By the same arguments, $P_2$ also ends in $w$ and is disjoint from $Z$.
    But this gives once again a $K_{3,3}$-minor (with vertices $u_1,u_2,v_1,v_2,w$, and $Z \cup \bigcup_{i \in \{2,\dots,k\}} A_i$ where we use the paths $P_1$ and $P_2$ to connect $w$ to $u_1$ and $u_2$).
    So $P_1$ ends in the set $\bigcup_{i \in \{2,\dots,\ell\}} A_i$.
    Using Observation \ref{obs:wl-knows-paths-avoiding-colors} again, there is a path $P_2$ from $u_2$ to $\bigcup_{i \in \{2,\dots,\ell\}} A_i$ that avoids $C_V(G[c,d],\chi)$.
    Again, we obtain a minor $K_{3,3}$ (with vertices $u_1,u_2,v_1,v_2,w$, and $Z \cup \bigcup_{i \in \{2,\dots,k\}} A_i$ where we use the paths $P_1$ and $P_2$ to connect $Z \cup \bigcup_{i \in \{2,\dots,k\}} A_i$ to $u_1$ and $u_2$).
    This completes the case that $G[c,d]$ is disconnected.
    
    \begin{figure}
     \centering
     \begin{subfigure}[b]{.37\linewidth}
      \begin{tikzpicture}
       \node[vertex,blue!80] (u1) at (0:0.8) {};
       \node[vertex,red!80] (u2) at (90:0.8) {};
       \node[vertex,blue!80] (u3) at (180:0.8) {};
       \node[vertex,red!80] (u4) at (270:0.8) {};
       
       \node at (0,-1.4) {$A_1$};
       
       \node[vertex,magenta!80] (w) at (1.6,0) {};
       \node at (1.6,-0.3) {$w$};
       
       \node[vertex,blue!80] (v1) at ($(3.2,0)+(0:0.8)$) {};
       \node[vertex,red!80] (v2) at ($(3.2,0)+(90:0.8)$) {};
       \node[vertex,blue!80] (v3) at ($(3.2,0)+(180:0.8)$) {};
       \node[vertex,red!80] (v4) at ($(3.2,0)+(270:0.8)$) {};
       
       \node at (3.2,-1.4) {$A_2$};
       
       \foreach \i/\j in {1/2,2/3,3/4,4/1}{
        \draw[line width=1.6pt] (u\i) edge (u\j);
        \draw[line width=1.6pt] (v\i) edge (v\j);
       }
       \draw[line width=1.6pt, darkpastelgreen] (w) edge (u2);
       \draw[line width=1.6pt, darkpastelgreen] (w) edge (u4);
       \draw[line width=1.6pt, darkpastelgreen] (w) edge (v2);
       \draw[line width=1.6pt, darkpastelgreen] (w) edge (v4);
       
       \draw[line width=1.6pt, mYellow, fill = mYellow!40] (1.6,2) ellipse (1.6cm and 0.6cm);
       \node at (1.6,2) {$Z$};
       
       \draw[line width=1.6pt, gray] (u3) edge (0.2,1.9);
       \draw[line width=1.6pt, gray] (u1) edge (1.3,1.6);
       \draw[line width=1.6pt, gray] (v3) edge (1.9,1.6);
       \draw[line width=1.6pt, gray] (v1) edge (3.0,1.9);
      \end{tikzpicture}
      \caption{Contracting $Z \cup A_2$ to one vertex results in a graph isomorphic to $K_{3,3}$.}
      \label{fig:c-d-disconnected-k33}
     \end{subfigure}
     \hfill
     \begin{subfigure}[b]{.54\linewidth}
      \begin{tikzpicture}
       \draw[line width=1.6pt, mYellow, fill = mYellow!40] (0,0) ellipse (1.2cm and 2.4cm);
       \node at (-1.6,0) {$Z$};
       
       \node[vertex,magenta!80] (w) at (2.4,0) {};
       \node at (2.4,-0.3) {$w$};
       
       \foreach \q/\k in {-1.4/1,0/2,1.4/3}{
        \node[vertex,blue!80] (u-\k-1) at ($(0,\q)+(0:0.4)$) {};
        \node[vertex,red!80] (u-\k-2) at ($(0,\q)+(90:0.4)$) {};
        \node[vertex,blue!80] (u-\k-3) at ($(0,\q)+(180:0.4)$) {};
        \node[vertex,red!80] (u-\k-4) at ($(0,\q)+(270:0.4)$) {};
        
        \foreach \i/\j in {1/2,2/3,3/4,4/1}{
         \draw[line width=1.6pt] (u-\k-\i) edge (u-\k-\j);
        }
       }
       
       \draw[line width=1.6pt, darkpastelgreen] (w) edge (u-1-2);
       \draw[line width=1.6pt, darkpastelgreen, bend left=20] (w) edge (u-1-4);
       \draw[line width=1.6pt, darkpastelgreen] (w) edge (u-2-2);
       \draw[line width=1.6pt, darkpastelgreen] (w) edge (u-2-4);
       \draw[line width=1.6pt, darkpastelgreen, bend right=20] (w) edge (u-3-2);
       \draw[line width=1.6pt, darkpastelgreen] (w) edge (u-3-4);
       
       \node[vertex,blue!80] (v1) at ($(4.8,0)+(0:0.8)$) {};
       \node[vertex,red!80] (v2) at ($(4.8,0)+(90:0.8)$) {};
       \node[vertex,blue!80] (v3) at ($(4.8,0)+(180:0.8)$) {};
       \node[vertex,red!80] (v4) at ($(4.8,0)+(270:0.8)$) {};
       
       \node at (4.8,-1.4) {$A_1$};
       
       \foreach \i/\j in {1/2,2/3,3/4,4/1}{
        \draw[line width=1.6pt] (v\i) edge (v\j);
       }
       \draw[line width=1.6pt, darkpastelgreen] (w) edge (v2);
       \draw[line width=1.6pt, darkpastelgreen] (w) edge (v4);
       
       \draw[line width=1.6pt, gray] (v2) edge (0.4,2.1);
       \draw[line width=1.6pt, gray] (v4) edge (0.4,-2.1);
       
       \node[smallvertex] (p1) at (3.1,0.5) {};
       \node[smallvertex] (p2) at (3.6,0.1) {};
       \node[smallvertex] (p3) at (3.7,0.8) {};
       \node[smallvertex] (p4) at (4.5,1.7) {};
       \node[smallvertex] (p5) at (5.3,1.4) {};
       \node[smallvertex] (p6) at (5.6,0.7) {};
       
       \draw[line width=0.8pt, gray] (w) edge (p1);
       \draw[line width=0.8pt, gray] (p1) edge (p2);
       \draw[line width=0.8pt, gray] (p2) edge (v3);
       \draw[line width=0.8pt, gray] (p1) edge (p3);
       \draw[line width=0.8pt, gray] (p3) edge (p4);
       \draw[line width=0.8pt, gray] (p4) edge (p5);
       \draw[line width=0.8pt, gray] (p5) edge (p6);
       \draw[line width=0.8pt, gray] (p6) edge (v1);
      \end{tikzpicture}
      \caption{Contracting $Z$ to one vertex results in a minor isomorphic to $K_{3,3}$.
       The paths $P_1$ and $P_2$ are given by white vertices and they are contracted to $w$.}
      \label{fig:c-d-connected-k33}
     \end{subfigure}
     \caption{Visualizations for the proof of Claim \ref{claim:contracted-color-type-iii}.}
    \end{figure}
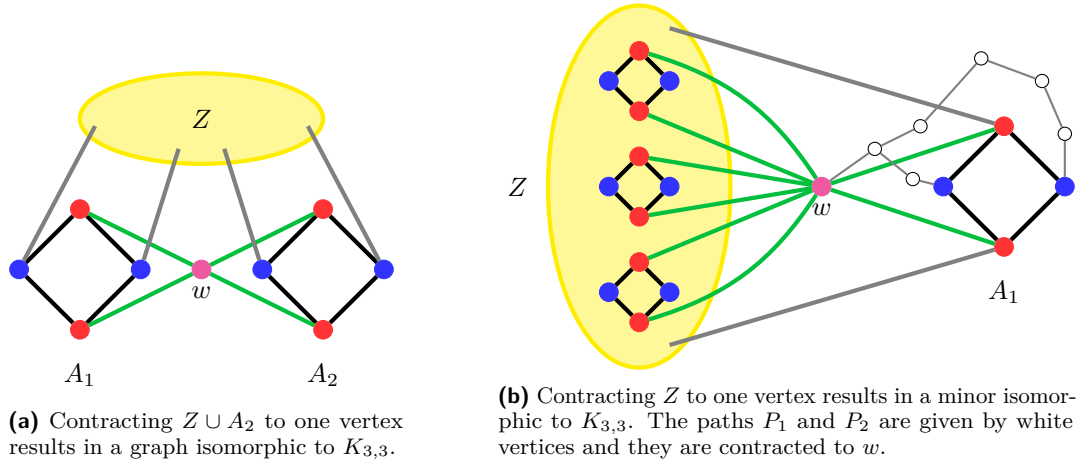
    
    \medskip
    
    Next, suppose that $G[c,d]$ is connected.
    This means $W = \{w\}$.
    Consider the graph $G' \coloneqq G - W$.
    Observe that $G'$ is still connected since $G$ is $3$-connected.
    Note that $\chi|_{(V(G'))^2}$ is a coloring that is $2$-stable with respect to $G'$.
    Consider the set $A_1$.
    By Lemma \ref{la:disconnected-color-subgraphs}, there is a $\chi$-invariant set $M \subseteq V(G')$ and a connected component $Z$ of $G' - A_1$ such that $N_{G'}(Z) = M \cap A_1$.
    Also, $\bigcup_{i \in \{2,\dots,\ell\}}A_i \subseteq Z$.
    This means that $w \in N_G(Z)$ because $w \in N_G(A_i)$ for all $i \in [\ell]$.
    If $|\widehat{A}_i| \geq 3$ for all $i \in [\ell]$, then we obtain a $K_5$-minor (with two vertices $w$ and $Z$ as well as a triangle obtained from $A_1$), which contradicts by Theorem \ref{thm:wagner} the planarity of $G$.
    Otherwise, $|\widehat{A}_i| = 2$ for all $i \in [\ell]$, which means that $G[c]$ is bicolored.
    If $A_1 \setminus \widehat{A}_1 \subseteq N_G(Z)$, then we obtain a $K_{3,3}$-minor (with the four vertices from $A_1$, $w$ and $Z$), which contradicts by Theorem \ref{thm:wagner} the planarity of $G$.
    So assume that $N_G(Z) = \{w\} \cup \widehat{A}_1$.
    Suppose $\widehat{A}_1 = \{v_1,v_2\}$ and $A_1 \setminus \widehat{A}_1 = \{u_1,u_2\}$.
    Since $G$ is $3$-connected, the graph $G - \{v_1,v_2\}$ is connected, and hence there is a path $P_1$ from $u_1$ to $w$.
    Note that $P_1$ does not contain any vertex from $Z$.
    Also, by possibly switching the roles of $u_1$ and $u_2$, we may assume that no internal vertex of $P_1$ is contained in $A_1$.
    Hence, $P_1$ avoids $C_V(G[c,d],\chi)$.
    Using Observation \ref{obs:wl-knows-paths-avoiding-colors}, there is also a path $P_2$ from $u_2$ to $w$ that avoids $C_V(G[c,d],\chi)$.
    Once again, $P_2$ is disjoint from $Z$.
    But this gives a $K_{3,3}$-minor (with vertices $u_1,u_2,v_1,v_2,w$, and $Z$ where we use the paths $P_1$ and $P_2$ to connect $w$ to $u_1$ and $u_2$; see Figure \ref{fig:c-d-connected-k33}), which contradicts the planarity of $G$.
   \end{claimproof}
   
   \begin{claim}
    \label{claim:one-neighbor-in-cycle}
    Let $w \in W$ and $i \in [\ell]$ such that $wA_i \in E((G/c)[\widetilde{d}])$.
    Then $|N_{G[d]}(w) \cap A_i| = 1$.
   \end{claim}
   \begin{claimproof}
    Let $w_1,\dots,w_k$ denote the set of neighbors of $A_i$ in the graph $(G/c)[\widetilde{d}]$.
    Observe that $k \geq 2$ by Claim \ref{claim:contracted-color-type-iii}.
    Since $d$ has Type I, every vertex in $\widehat{A}_i$ is adjacent to exactly one of the vertices $w_1,\dots,w_k$ in the graph $G[d]$.
    For $j \in [k]$, let
    \[\widehat{A}_{i,j} \coloneqq N_{G[d]}(w_j) \cap A_i.\]
    Then $\widehat{A}_{i,1},\dots,\widehat{A}_{i,k}$ forms a partition of $\widehat{A}_i$.

    Using Lemma \ref{la:factor-graph-2-wl}, we get that $|\widehat{A}_{i,j}| = |\widehat{A}_{i,j'}|$ for all $j,j' \in [k]$ (all cardinalities equal to the number of occurrences of $d$ in the multiset $\widetilde{d}$).
    Since $w_1,\dots,w_k \in \exte(A_i)$, it holds that every $\widehat{A}_{i,j}$ gives a connected subgraph of $(G[c])[A_i]$ (when defining a cycle on $\widehat{A}_i$ as described above).
    Now suppose towards a contradiction that $|\widehat{A}_{i,1}| \geq 2$.
    Then there are vertices $v_1,v_2,v_3 \in \widehat{A}_i$ such that $v_1,v_2,v_3$ appear consecutively along the cycle $(G[c])[A_i]$ and $v_1,v_2 \in \widehat{A}_{i,1}$, but $v_3 \notin \widehat{A}_{i,1}$. 
    But then $\chi(v_1,v_2) \neq \chi(v_2,v_3)$, since $\chi$ is $2$-stable.

    But this contradicts the fact that all edges of $G[c]$ receive the same color and $(G[c])[A_i]$ is a cycle.
   \end{claimproof}
   
   In particular, since $d$ has Type I and $r \geq 3$, Claim \ref{claim:one-neighbor-in-cycle} implies that $|\widehat{A}_i| = \deg_{(G/c)[\widetilde{d}]}(A_i)$.
   Also, $r = \widetilde{r}$.
   This means that $\widetilde{d}$ has Type III in the graph $G/c$ using Claim \ref{claim:contracted-color-type-iii}.   
   
   Now, first suppose that $(G/c)[\widetilde{d}]$ has minimum degree $3$.
   Then every connected component of $(G/c)[\widetilde{d}]$ is isomorphic to a bicolored cube (see Figure \ref{fig:graph-bi-cube}), a rhombic triacontahedron (see Figure \ref{fig:graph-rhombic-triacontahedron}), or a rhombic dodecahedron (see Figure \ref{fig:graph-rhombic-dodecahedron}) by Lemma \ref{la:edge-transitive-two-vertex-colors}.
   Using Lemma \ref{la:disconnected-color-subgraphs}, we conclude that $(G/c)[\widetilde{d}]$ is connected.
   It follows that $G[c,d]$ is isomorphic to a parallel subdivision of a chamfered tetrahedron, a chamfered cube, a chamfered octahedron, a chamfered dodecahedron, or a chamfered icosahedron.
   
   So assume that every vertex $A_i$, $i \in [\ell]$, has degree $2$ in the graph $(G/c)[\widetilde{d}]$.
   Then $|\widehat{A}_i| = 2$ for all $i \in [\ell]$, and hence every connected component of $G[c]$ is a bicolored cycle of length~$4$.
   Let $B$ be the vertex set of a connected component of $G[c,d]$ and let $\widetilde{B}$ be the corresponding set in the graph $(G/c)[\widetilde{d}]$.
   By Theorem \ref{thm:classification-edge-transitive} (and the comments below the theorem), $((G/c)[\widetilde{d}])[\widetilde{B}]$ is isomorphic to an $s$-subdivision of $K_2$, a cycle $C_m$ for some $m \geq 3$, or one of the graphs from Figure \ref{fig:graph-k4} -- \ref{fig:graph-icosidodecahedron}.
   
   \begin{claim}
    \label{claim:one-cycle-per-subdivision}
    If $|W \cap B| \geq 3$, then $s = 1$.
   \end{claim}
   \begin{claimproof}
    Suppose towards a contradiction that $s \geq 2$.
    Let $w_1,w_2 \in W \cap B$ be two vertices with distance $2$ in the graph $(G/c)[\widetilde{d}]$, and, without loss of generality, assume that $A_1,\dots,A_s$ are those vertices that are adjacent to both $w_1$ and $w_2$ in $(G/c)[\widetilde{d}]$.
    For $i \in [s]$, suppose that $\widehat{A}_i = \{v_1^i,v_2^i\}$ such that $w_1v_1^i,w_2v_2^i \in E(G[d])$ for all $i \in [s]$.
    Also suppose that $A_i \setminus \widehat{A}_i = \{u_1^i,u_2^i\}$ (recall that $G[c]$ is bicolored).
    
    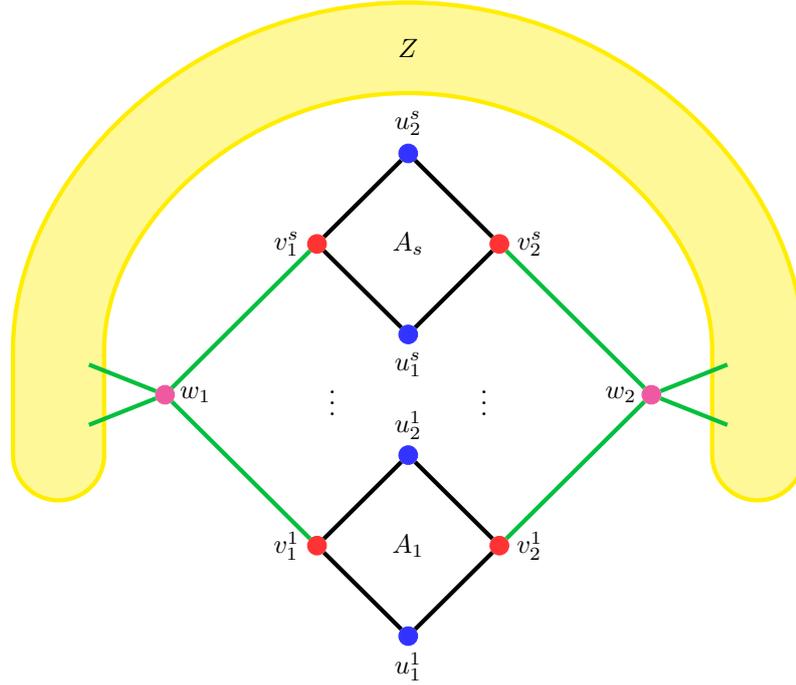
\begin{figure}
     \centering
     \begin{tikzpicture}[rotate=90]
      \draw[line width=1.6pt, rounded corners = 0.6cm, mYellow, fill = mYellow!40] (0.6,4) -- (2,4) arc (90:-90:4) -- (0.6,-4) -- (0.6,-5.2) -- (2,-5.2) arc (-90:90:5.2) -- (0.6,5.2) -- cycle;
      \node at (6.6,0) {$Z$};
      
      \foreach \k/\q in {0/1,1/s}{
       \node[vertex,blue!80] (u-\k-1) at ($(4*\k,0)+(0:1.2)$) {};
       \node[vertex,red!80] (u-\k-2) at ($(4*\k,0)+(90:1.2)$) {};
       \node[vertex,blue!80] (u-\k-3) at ($(4*\k,0)+(180:1.2)$) {};
       \node[vertex,red!80] (u-\k-4) at ($(4*\k,0)+(270:1.2)$) {};
       
       \node at ($(4*\k,0)+(0:1.6)$) {$u_2^{\q}$};
       \node at ($(4*\k,0)+(90:1.6)$) {$v_1^{\q}$};
       \node at ($(4*\k,0)+(180:1.6)$) {$u_1^{\q}$};
       \node at ($(4*\k,0)+(270:1.6)$) {$v_2^{\q}$};
       
       \node at (4*\k,0) {$A_{\q}$};
       
       \foreach \i/\j in {1/2,2/3,3/4,4/1}{
        \draw[line width=1.6pt] (u-\k-\i) edge (u-\k-\j);
       }
      } 
      
      \node at (2,1) {$\vdots$};
      \node at (2,-1) {$\vdots$};
      
      \node[vertex,magenta!80] (w-1) at (2,3.2) {};
      \node[vertex,magenta!80] (w-2) at (2,-3.2) {};
      \node at (2,2.8) {$w_1$};
      \node at (2,-2.8) {$w_2$};
      
      \draw[line width=1.6pt, darkpastelgreen] (w-1) edge (u-0-2);
      \draw[line width=1.6pt, darkpastelgreen] (w-2) edge (u-0-4);
      \draw[line width=1.6pt, darkpastelgreen] (w-1) edge (u-1-2);
      \draw[line width=1.6pt, darkpastelgreen] (w-2) edge (u-1-4);
      
      \draw[line width=1.6pt, darkpastelgreen] (w-1) edge (1.6,4.2);
      \draw[line width=1.6pt, darkpastelgreen] (w-1) edge (2.4,4.2);
      \draw[line width=1.6pt, darkpastelgreen] (w-2) edge (1.6,-4.2);
      \draw[line width=1.6pt, darkpastelgreen] (w-2) edge (2.4,-4.2);
      
     \end{tikzpicture}
     \caption{Visualization for the proof of Claim \ref{claim:one-cycle-per-subdivision}.}
     \label{fig:c-4-subdivision-analysis}
    \end{figure}

    Let $X \coloneqq \{w_1,w_2\} \cup \bigcup_{i \in [s]} A_i$ denote the set of all vertices listed above.
    Note that $X \subseteq B$.
    We first argue that $G - X$ contains a connected component $Z$ such that $V(G[c,d]) \setminus X \subseteq Z$.
    By the classification of the graph $((G/c)[\widetilde{d}])[\widetilde{B}]$ and since $|W \cap B| \geq 3$, we conclude that $((G[c,d])[B]) - X$ is connected (and non-empty).
    Let $Z$ denote the connected component of $G - X$ that contains $B \setminus X$.
    Since $\{\chi(v,v) \mid v \in B\} = \{\chi(v,v) \mid v \in B \setminus X\}$, Lemma~\ref{la:disconnected-color-subgraphs} implies that actually $V(G[c,d]) \setminus X \subseteq Z$.
    
    Since $G$ is $3$-connected, we get that $|N_G(Z)| \geq 3$.
    Let $X' \coloneqq X \setminus \{w_1,w_2\} = \bigcup_{i \in [s]} A_i$.
    So $N_G(Z) \cap X' \neq \emptyset$.
    Let $x \in N_G(Z) \cap X'$ and let $x' \in X'$ such that $\chi(x,x) = \chi(x',x')$.
    We claim that $x' \in N_G(Z)$.
    Let $D \coloneqq C_V(G[c,d],\chi)$ denote the set of vertex colors appearing in $G[c,d]$.
    Since $x \in N_G(Z)$, there is some $y \in Z$ such that $\chi(y,y) \in D$ and there is a path from $x$ to $y$ in $G$ that avoids $D$. 
    Also, because $\chi(x,x) = \chi(x'x')$, there is some $y' \in V(G)$ such that $\chi(x,y) = \chi(x',y')$.
    Then $\chi(y',y') = \chi(y,y) \in D$ and there is a path from $x'$ to $y'$ in $G$ that avoids $D$ by Observation \ref{obs:wl-knows-paths-avoiding-colors}.
    Also, using Lemma \ref{la:factor-graph-2-wl} and the structural classification of $((G/c)[\widetilde{d}])[\widetilde{B}]$, we conclude that $y' \notin X$ (as pairs in the set $X' \times X$ receive different colors than pairs from the set $X' \times (V(G) \setminus X)$).
    So $y' \in Z$.
    But this means that $x' \in N_G(Z)$.
    Overall, this means that
    \begin{enumerate}[label=(\roman*)]
     \item $\{v_1^i,v_2^i \mid i \in [s]\} \subseteq N_G(Z)$, or
     \item $\{u_1^i,u_2^i \mid i \in [s]\} \subseteq N_G(Z)$.
    \end{enumerate}
    In the latter case, we immediately obtain a $K_{3,2s}$-minor with vertices $Z$, $\{w_1,v_1^1,\dots,v_1^s\}$ and $\{w_2,v_2^1,\dots,v_2^s\}$ on the left side and vertices $u_j^i$, $i \in [s]$ and $j \in \{1,2\}$, on the right side.
    By Theorem \ref{thm:wagner}, this contradicts the planarity, since $s \geq 2$ by assumption.
    So we may focus on the first option.
    Actually, we may assume that $u_j^i \notin N_G(Z)$ for all $i \in [s]$ and $j \in \{1,2\}$.
    In this case, we obtain a $K_{3,s}$-minor with vertices $Z$, $w_1$ and $w_2$ on the left side and vertices $A_i$, $i \in [s]$, on the right side.
    By planarity, it follows that $s = 2$.
    
    Now, consider the graph $G' \coloneqq G - \{v_1^1,v_2^1\}$.
    Observe that $G'$ is connected because $G$ is $3$-connected.
    So there is a path $P_1$ from a vertex from $\{u_1^1,u_2^1\}$ to some vertex from $V(G[c,d]) \setminus \{u_1^1,u_2^1,v_1^1,v_2^1\}$ that avoids $D$.
    Without loss of generality, assume that $P_1$ starts in $u_1^1$ and let $u_1'$ denote its end vertex.
    By the comments above, we conclude that $u_1' \notin Z$.
    This means that $u_1' \in X$.
    In other words, $u_1' \in \{w_1,w_2,v_1^2,v_2^2,u_1^2,u_2^2\} = \{w_1,w_2\} \cup A_2$.
    We also conclude that $P_1$ does not visit any vertex from $Z$.
    Using the same arguments as before, there is also a path $P_2$ from $u_2^1$ to some vertex $u_2' \in \{w_1,w_2\} \cup A_2$ that avoids $D$.
    Again, $P_2$ is disjoint from $Z$.
    But now, this gives a $K_{3,3}$-minor with vertices $u_1^1$, $u_2^1$ and $Z$ on the left and $v_1^1$, $v_2^1$ and $\{w_1,w_2\} \cup A_2$ on the right.
    Note that we can contract the paths $P_1$ and $P_2$ onto $\{w_1,w_2\} \cup A_2$ to connect $u_1^1$ and $u_2^1$ to $\{w_1,w_2\} \cup A_2$.
    By Theorem \ref{thm:wagner}, this is a contradiction to the planarity of $G$.
   \end{claimproof}
   
   First suppose that $((G/c)[\widetilde{d}])[\widetilde{B}]$ is isomorphic to an $s$-subdivision of $K_2$.
   Note that $s \geq 3$ in this case because $\widetilde{r} \geq 3$.
   It follows that $(G[c,d])[B]$ is isomorphic to $K_{2,s}^*$.
   Also, $G[c,d]$ is connected using Lemma \ref{la:disconnected-color-subgraphs} (observe that $G$ is $3$-connected and hence, $|N_G(Z)| \geq 3$ where $Z$ is the vertex set of the connected component provided by Lemma \ref{la:disconnected-color-subgraphs}).
   So together $G[c,d]$ is isomorphic to $K_{2,s}^*$.
   
   Otherwise $((G/c)[\widetilde{d}])[\widetilde{B}]$ is isomorphic to an $s$-subdivision of $C_m$ for $m \geq 3$ or one of the graphs from Figure \ref{fig:graph-k4} -- \ref{fig:graph-icosidodecahedron}.
   So $s = 1$ by Claim \ref{claim:one-cycle-per-subdivision}.
   Since $\widetilde{r} \geq 3$, we conclude that $((G/c)[\widetilde{d}])[\widetilde{B}]$ is isomorphic to a $1$-subdivision of one of the graphs from Figure \ref{fig:graph-k4} -- \ref{fig:graph-icosidodecahedron}.
   So $(G/c)[\widetilde{d}]$ is connected by Lemma \ref{la:disconnected-color-subgraphs}, and hence $G[c,d]$ is connected as well.
   To complete the proof, it only remains to show that $G[c,d]$ cannot be isomorphic to a $C_4$-subdivision of a cuboctahedron (Figure \ref{fig:graph-cuboctahedron}) or an icosidodecahedron (Figure \ref{fig:graph-icosidodecahedron}).
   
   We show that $G[c,d]$ cannot be isomorphic to a $C_4$-subdivision of a cuboctahedron, the analysis for the icosidodecahedron is analogous.
   We remark that the arguments are similar to those already used in the proof of Theorem \ref{thm:type-iii-classification}.
   
   Suppose towards a contradiction that $G[c,d]$ is isomorphic to a $C_4$-subdivision of a cuboctahedron.
   Let $H$ denote the corresponding cuboctahedron on vertex set $W$.
   Let $w,w' \in W$ be neighbors in $H$ and without loss of generality, suppose that the $4$-cycle on $A_1$ is adjacent to $w$ and $w'$.
   A visualization is given in Figure \ref{fig:c4-cuboctahedron}.
   Let us fix a planar embedding of $G$ and consider the induced embedding of $G[c,d]$.
   Suppose that $\widehat{A}_1 = \{v,v'\}$ such that $vw,v'w' \in E(G[d])$ and $A_1 \setminus \widehat{A}_1 = \{u_1,u_2\}$.
   Let $F,F_1,F_2$ denote the three faces of $G[c,d]$ so that $(G[c])[A_1]$ bounds $F$, $u_1$ is incident to $F_1$ but not $F_2$, and $u_2$ is incident to $F_2$ but not $F_1$.
   We have that $|F| = 4$ and $\{|F_1|,|F_2|\} = \{12,16\}$.
   Without loss of generality, suppose that $|F_1| = 12$.
   
   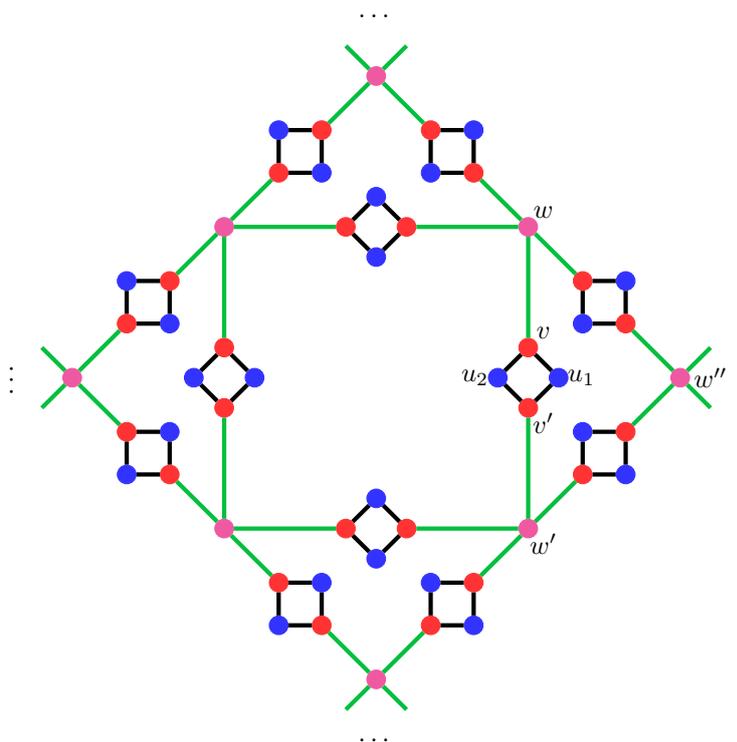
\begin{figure}
    \centering
    \begin{tikzpicture}
     \node[vertex,magenta!80] (1) at (2,2) {};
     \node[vertex,magenta!80] (2) at (-2,2) {};
     \node[vertex,magenta!80] (3) at (-2,-2) {};
     \node[vertex,magenta!80] (4) at (2,-2) {};
   
     \node[vertex,magenta!80] (5) at (4,0) {};
     \node[vertex,magenta!80] (6) at (0,4) {};
     \node[vertex,magenta!80] (7) at (-4,0) {};
     \node[vertex,magenta!80] (8) at (0,-4) {};

     \foreach \x/\y/\r/\v/\w [count = \k] in {2/0/0/1/4, 0/2/90/2/1, -2/0/0/2/3, 0/-2/90/3/4, 3/1/45/1/5, 3/-1/-45/5/4, 1/3/45/6/1, -1/3/-45/6/2, -3/1/-45/2/7, -3/-1/45/7/3, -1/-3/45/3/8, 1/-3/-45/4/8}{
      
      \node[vertex,blue!80] (u-\k-1) at ($(\x,\y)+(0+\r:0.4)$) {};
      \node[vertex,red!80] (u-\k-2) at ($(\x,\y)+(90+\r:0.4)$) {};
      \node[vertex,blue!80] (u-\k-3) at ($(\x,\y)+(180+\r:0.4)$) {};
      \node[vertex,red!80] (u-\k-4) at ($(\x,\y)+(270+\r:0.4)$) {};
      
      \foreach \i/\j in {1/2,2/3,3/4,4/1}{
       \draw[line width=1.6pt] (u-\k-\i) edge (u-\k-\j);
      }
      
      \draw[line width=1.6pt, darkpastelgreen] (u-\k-2) edge (\v);
      \draw[line width=1.6pt, darkpastelgreen] (u-\k-4) edge (\w);
     }
   
     \draw[line width=1.6pt, darkpastelgreen] (5) edge (4.4,0.4);
     \draw[line width=1.6pt, darkpastelgreen] (5) edge (4.4,-0.4);
     \node[rotate=90] at (4.8,0) {$\dots$};
     \draw[line width=1.6pt, darkpastelgreen] (6) edge (0.4,4.4);
     \draw[line width=1.6pt, darkpastelgreen] (6) edge (-0.4,4.4);
     \node at (0,4.8) {$\dots$};
     \draw[line width=1.6pt, darkpastelgreen] (7) edge (-4.4,0.4);
     \draw[line width=1.6pt, darkpastelgreen] (7) edge (-4.4,-0.4);
     \node[rotate=90] at (-4.8,0) {$\dots$};
     \draw[line width=1.6pt, darkpastelgreen] (8) edge (0.4,-4.4);
     \draw[line width=1.6pt, darkpastelgreen] (8) edge (-0.4,-4.4);
     \node at (0,-4.8) {$\dots$};
     
     \node at (2.2,2.2) {$w$};
     \node at (2.2,-2.2) {$w'$};
     \node at (4.4,0) {$w''$};
     \node at (2.2,0.6) {$v$};
     \node at (2.2,-0.6) {$v'$};
     \node at (2.7,0) {$u_1$};
     \node at (1.3,0) {$u_2$};
     
    \end{tikzpicture}
    \caption{Visualization for the proof of Theorem \ref{la:type-ii-connected-subgraphs-classification}. The figure shows a part of a $C_4$-subdivision of a cuboctahedron.}
    \label{fig:c4-cuboctahedron}
   \end{figure}
   
   Consider the connected components of $G - V(G[c,d])$.
   Let $C \coloneqq \{\chi(x,x') \mid xx' \in E(G - V(G[c,d]))\}$.
   Note that the connected components of $G - V(G[c,d])$ are precisely those $G[C]$.

   For the purpose of analysis, we consider the graph $G/C$ obtained from contracting each component of $G - V(G[c,d])$ to one vertex.
   The graph $G/C$ is still $3$-connected, since $G[c,d]$ is $2$-connected.
   Also recall that $\chi/C$ is a $2$-stable coloring of $G/C$ by Lemma \ref{la:factor-graph-2-wl}.
   
   Now, since $G$ is $3$-connected, there is some vertex set $Z_1$ of some connected component of $G - V(G[c,d])$ such that $\{u_1,u_2\} \cap N_G(Z) \neq \emptyset$ and $N_G(Z) \nsubseteq A_1$.
   Because $\chi(u_1,u_1) = \chi(u_2,u_2)$, we may assume without loss of generality that $u_1 \in N_G(Z_1)$ (if $u_2 \in N_G(Z_1)$ then there is another such component $Z_1'$ with $u_1 \in N_G(Z_1)$ as otherwise $2$-WL could distinguish between $u_1$ and $u_2$).
   This means that $Z_1$ is located in the face $F_1$.
   
   Now, first suppose that $N_G(Z_1) \nsubseteq \{w,w',v,v',u_1\}$.
   Using Lemma \ref{la:factor-graph-2-wl}, there is a vertex set of another component $Z_2$ of $G - V(G[c,d])$ such that
   \[(\chi/C)(u_1,Z_1) = (\chi/C)(u_2,Z_2)\]
   In particular, $(\chi/C)(Z_1,Z_1) = (\chi/C)(Z_2,Z_2)$.
   However, it can be checked that this is impossible.
   For example, let $w''$ be chosen as in Figure \ref{fig:c4-cuboctahedron} and suppose $w'' \in N_G(Z_1)$, then there are two walks from $u_1$ to $w''$ of length $6$ in the graph $G[c,d]$.
   By stability of the coloring $\chi$, there also must exist such a vertex $w''' \in N_G(Z_2)$ such that there are two walks from $u_2$ to $w''$ of length $6$ in the graph $G[c,d]$.
   But $Z_2$ has to be located in the face $F_2$ and no vertex $w'''$ incident to $F_2$ satisfies this property.
   
   So $N_G(Z_1) \subseteq \{w,w',v,v',u_1\}$ and $\{w,w'\} \cap N_G(Z_1) \neq \emptyset$ by the assumption above.
   Without loss of generality, suppose that $w \in N_G(Z_1)$.
   Since $\chi(u_1,v) = \chi(u_1,v')$, there is also a vertex set $Z_1'$ of a connected component of $G - V(G[c,d])$ (possibly $Z_1 = Z_1'$) such that $N_G(Z_1') \subseteq \{w,w',v,v',u_1\}$ and $w' \in N_G(Z_1')$.
   Using $\chi(u_1,u_1) = \chi(u_2,u_2)$, the argument can be repeated to obtain components $Z_2,Z_2'$ of $G - V(G[c,d])$ such that $N_G(Z_2) \subseteq \{w,w',v,v',u_2\}$ and $w \in N_G(Z_2)$ as well as $N_G(Z_2') \subseteq \{w,w',v,v',u_2\}$ and $w' \in N_G(Z_2')$.
   But now $\{w,w'\}$ forms a $2$-separator of $G$, since, by planarity, no other component of $G - V(G[c,d])$ can be connected to a vertex from $A_1$ as well as some vertex from $V(G[c,d]) \setminus (A_1 \cup \{w,w'\})$.   
   This contradicts $G$ being $3$-connected.\qedhere
 \end{description}
\end{proof}

Together, Lemma \ref{la:cycle-connections} and Lemma \ref{la:type-ii-connected-subgraphs-classification} prove Theorem \ref{thm:type-ii-connected-subgraphs-classification}.

\medskip

Let us point out that, for Option \ref{item:type-ii-connected-subgraphs-classification-3} from Theorem \ref{thm:type-ii-connected-subgraphs-classification}, using the same arguments as for edge colors of Type III, we obtain that $\Aut(G)$ is isomorphic to a subgroup of $\Aut(G[c,d])$.

\section{Conclusion}

Overall, by combining Lemmas \ref{la:face-cycle-with-star-edges} and \ref{la:directed-cycle} and Theorems \ref{thm:type-iii-classification} and \ref{thm:type-ii-connected-subgraphs-classification}, we obtain that every $3$-connected planar graph $G$ satisfies one of the following options.
\begin{enumerate}[label=(\Alph*)]
 \item\label{item:outcome-1} There is some $v \in V(G)$ such that $\Disc_G(v) = V(G)$, which implies that $2$-WL determines pair orbits of $G$ by Lemma \ref{la:wl-pair-orbits-from-fixing-number},
 \item\label{item:outcome-2} there is an edge color $c \in C_E(G,\WL{2}{G})$ that defines a matching, or
 \item\label{item:outcome-3} there is a set $C \subseteq C_E(G,\WL{2}{G})$ such that $|C| \leq 2$ and $G[C]$ is essentially a Platonic or Archimedean solid, or stems from a small number of infinite families of connected graphs.
\end{enumerate}
Option \ref{item:outcome-3} contains the graphs listed in Theorems \ref{thm:type-iii-classification} and \ref{thm:type-ii-connected-subgraphs-classification}, as well as the class of bipyramids from Lemma \ref{la:directed-cycle} and the class of cycles to cover graphs of Type IIb.

It remains an important open question whether $2$-WL identifies every planar graph.
With the structural insights from this paper, it now suffices to focus on Case \ref{item:outcome-3} and, as explained in the introduction, the classification of the subgraphs $G[C]$ appearing in this case should be a crucial step to determining the WL dimension of planar graphs.

\bibliography{literature}

\end{document}